\documentclass[12pt]{amsart}

\usepackage[T1]{fontenc}
\usepackage{vmargin}
\usepackage{amssymb}
\usepackage{amscd}
\usepackage{mathtools}
\usepackage{mathabx}
\usepackage{bbm}
\usepackage{a4wide}
\usepackage{supertabular}
\usepackage{array}
\usepackage{multicol}
\usepackage{enumitem}
\usepackage{hyperref}

\setmarginsrb{3cm}{2cm}{3cm}{3cm}{75pt}{20pt}{20pt}{30mm}
\newtheorem{Thm}{Theorem}
\newtheorem{Lem}[Thm]{Lemma}

\newtheorem{Cor}[Thm]{Corollary}
\newtheorem{Prop}[Thm]{Proposition}

\numberwithin{equation}{section}
\numberwithin{Thm}{section}

\theoremstyle{definition}
\newtheorem{Def}[Thm]{Definition}
\newtheorem{Rk}[Thm]{Remark}

\theoremstyle{remark}

\newtheorem*{Thm*}{Theorem}
\newtheorem*{Lem*}{Lemma}
\newtheorem*{Conj*}{Conjecture}
\newtheorem*{Cor*}{Corollary}
\newtheorem*{Def*}{Definition}
\newtheorem*{Prop*}{Proposition}
\newtheorem*{Exo*}{Exercise}
\newtheorem*{Exs*}{Examples}
\newtheorem*{Ex*}{Example}
\newtheorem*{Rk*}{Remark}
\newtheorem*{Rks*}{Remarks}

\newcommand{\N}{{\mathbb{N}}}

\newcommand{\R}{{\mathbb{R}}}
\newcommand{\Rp}{{\mathbb{R}^+}}
\newcommand{\Rop}{{\mathbb{R}_0^+}}

\newcommand{\banach}{{\mathcal{X}}}
\newcommand{\Fspace}{{\mathcal{F}}}
\newcommand{\bFspace}{{\mathcal{B}(\mathcal{F})}}
\newcommand{\bbanach}{{\mathcal{B}(\mathcal{X})}}

\newcommand{\hilbert}{{\mathcal{H}}}
\newcommand{\bhilbert}{{\mathcal{B}(\mathcal{H})}}

\newcommand{\interval}{{-\infty<t_1<t_2<\infty}}

\newcommand{\Ker}{{\mathrm{Ker}}}
\newcommand{\Dom}{{\mathrm{Dom}}}

\def\signnb{\bigskip \begin{flushleft} {\sc Norbert Barankai \par\vspace{3mm}
MTA-ELTE Theoretical Physics Research Group \par
P\'azm\'any P\'eter s\'et\'any 1/A\par
H-1117 Budapest\par
HUNGARY\par\vspace{3mm}
e-mail:} \tt{barankai@caesar.elte.hu} \end{flushleft}}

\def\signzz{\bigskip \begin{flushleft} {\sc Zolt\'an Zimbor\'as \par\vspace{3mm}
Department of Theoretical Physics\par
Wigner Research Centre for Physics \par
 Hungarian Academy of Sciences \par
P.O. Box 49, H-1525 Budapest\par
HUNGARY\par \vspace{1.5mm}
and \par \vspace{1.5mm}
Mathematical Institute\\ 
Budapest University of Technology and Economics\\
Egry J\'ozsef u 1. \par
H-1111 Budapest \par
HUNGARY \par \vspace{3mm}
e-mail:} \tt{zimboras.zoltan@wigner.mta.hu} \end{flushleft}}

\begin{document}

\title[]{Generalized quantum Zeno dynamics and ergodic means \\[1mm]
}

\vspace*{-8mm}

\author{N. Barankai}
\author{Z. Zimbor\'as}

\thanks{The authors are thankful to Francesco Giacosa whose talk on non-exponential decay phenomena and its interrelation with the Zeno effect initiated this research.  This research was supported by the National Research, Development and Innovation Office - NKFIH within the Quantum Technology National Excellence Program (Project No. 2017-1.2.1-NKP-2017-00001), and by Grants No. K108676, K124351, K124152.}

\begin{abstract}
We prove the existence of uniform limits for certain sequences of products of contractions and elements of a family of uniformly continuous  propagators acting on a Hilbert or a Banach space. From the point of view of Quantum Physics, the considered sequences can represent the evolution of a system whose dynamics, described by a continuous propagator, is disturbed by a sequence of generic quantum operations (e.g., projective measurements or unitary pulses). This includes and also generalizes  the so-called quantum Zeno dynamics. The time-evolution obtained from the limits of the considered sequences are described by propagators generated by ergodic means. The notion of such ergodic means is generalized in this paper to also include  propagators of time-dependent generators, and thus our results can be used to develop new forms of active decoherence suppression. In a similar way, we also consider the effective time-evolution obtained from adding to the possibly time-dependent generator of the original propagator a generator of a uniformly continuous contraction semi-group with asymptotically increasing weight in the sum. By proving a generalized adiabatic theorem, it is shown that also for this set-up the resulting time-evolution is governed by a suitable ergodic mean.
\end{abstract}



\maketitle

\vspace*{-5mm}

\tableofcontents

\vspace*{-5mm}

\bigskip

\section{Introduction}
\label{SEC1}

Performing frequent measurements on a quantum system, in order to determine whether it is in a particular subspace, forces the system to remain in that subspace. This phenomenon is  called the quantum Zeno effect \cite{misra1977zeno,peres1980zeno} and the limiting time evolution within the projected subspace is referred to as the quantum Zeno dynamics \cite{facchi2000quantum}.   The quantum Zeno effect has been demonstrated in several experiments \cite{kwiat1999high, nagels1997quantum,schafer2014experimental,
signoles2014confined}, and there has been a considerable effort to analyze it from a theoretical and mathematical point of view, see \cite{facchi2008quantum} and references therein. In recent years, its possible application to quantum information processing tasks, e.g.,  active decoherence avoidance and error correction, has also been studied. Motivated by some current works along this research line \cite{bernad2017dynamical,burgarth2018generalized,burgarth2018quantum,burgarth2014exponential,dominy2013analysis,paz2012zeno}, we give a generalization of these results concerning propagators with possibly time-dependent generators acting on Banach spaces.

\subsection{Statement of the main results}
We study two central questions concerning the quantum Zeno dynamics and the adiabatic theorem.\par

First, let $\hilbert$ be a Hilbert space and $\bhilbert$ be the space of bounded linear transformations of $\hilbert$ endowed with the usual sup-norm $\|\cdot\|_\bhilbert$. Let $T\in\bhilbert$ be a contraction on $\hilbert$ and $K\in\bhilbert$ be the generator of $(F(t))_{t\geq 0}$, a uniformly continuous semi-group on $\hilbert$. We ask the question, what is the uniform limit and speed of convergence of the sequence
\begin{equation}
C_n\cdot \bigl(TF(t/n)\bigr)^n\qquad \qquad n\in\mathbb{N} \,,  \; t\in\mathbb{R}_0^+,
\label{Q1a}
\end{equation}
as $n$ tends to infinity, where the sequence $(C_n)_{n\in\mathbb{N}}\subset\bhilbert$ is a `natural' choice which guarantees an interesting, non-trivial limit in the uniform topology of $\bhilbert$. 

More generally, let $\banach$ be a Banach space, let $K:[t_{1},t_{2}]\rightarrow \bbanach$ be continuous on the finite, closed interval $[t_1,t_2]$. Define the propagator $F(\cdot,\cdot)$ as the solution of the initial value problem
\begin{eqnarray}
\partial_{1}F(u,v)&=&K(u)F(u,v)\qquad\qquad t_{1}\leq v\leq u\leq t_{2}, \nonumber\\
F(v,v)&=&\mathbbm{1}_{\banach}.
\label{EQ1}
\end{eqnarray}
Assume that for each $s\in [t_1,t_2]$ and for each $n\in\N$, a partition $\mathfrak{T}_n(s)$ of $[t_{1},s]$ defined through the sequences $t_{1}=s_{n,0}<s_{n,1}<\dots <s_{n,n}=s$ is given, such that the norms of these partitions converges to zero as $n$ tends to infinity. Then, the same question can be raised regarding the uniform limit of the products
\begin{equation}
C_n\cdot \prod_{l=0}^{n-1}TF(s_{n,l+1},s_{n,l})\qquad \qquad n\in\mathbb{N}
\label{Q1b}
\end{equation}
in $\bbanach$. 

Our first main result connects the existence of the limit in (\ref{Q1b}) to the existence of an ergodic mean. For every infinite sequence $\alpha=(\alpha_1,\alpha_2,\cdots)$ of sequences $(\alpha_{n,l})_{0\leq l<n}$ of positive numbers satisfying an admissibility criterion, and for every isometric isomorphism $T\in\bbanach$ we associate an ergodic mean $\mathfrak{P}_{\alpha,T}$ which operates on the Banach space $C^0([t_1,t_2],\bbanach)$ of the continuous, $\bbanach$ valued functions (see section \ref{SEC2} for details). Then, the generalized quantum Zeno theorem is stated roughly as\par
\vspace{2mm}
\noindent\textbf{Theorem} (Generalized quantum Zeno dynamics of contractions). \textit{Let $\banach$ be a Banach space with a decomposition $\banach=\banach_{\mathrm{I}}\oplus\banach_{\mathrm{C}}$ such that both $\banach_{\mathrm{I}}$ and $\banach_{\mathrm{C}}$ are closed subspaces of $\banach$ and the continuous projection $P_{\mathrm{I}}$ corresponding to $\banach_{\mathrm{I}}$ and its complementary projection $P_{\mathrm{C}}$ are norm one. Assume that $T\in\bbanach$ is of the form $T=T_{\mathrm{I}}\oplus T_{\mathrm{C}}$, where $T_\mathrm{I}\in\mathcal{B}(\banach_{\mathrm{I}})$ is an isometric isomorphism, while $T_\mathrm{C}\in\mathcal{B}(\banach_{\mathrm{C}})$ is a strict contraction.  Let $K\in C^0([t_1,t_2],\bbanach)$ and define the propagator $F(\cdot,\cdot)$ as the solution of the initial value problem (\ref{EQ1}). Assume that  $\mathfrak{P}_{\alpha,T_\mathrm{I}}[P_\mathrm{I}KP_\mathrm{I}]$ exists and let $G(\cdot,t_1)$ be the solution of the initial value problem 
\begin{eqnarray}
\partial_1G(u,t_1)&=&\mathfrak{P}_{\alpha,T_\mathrm{I}}[P_\mathrm{I}KP_\mathrm{I}](u)G(u,t_1)\qquad\qquad t_{1}\leq u\leq t_{2},\nonumber\\
G(t_1,t_1)&=&P_{\mathrm{I}},
\end{eqnarray}
If $\alpha$ satisfies a certain condition on the variance of its members and $[t_1,t_2]\ni t\mapsto P_\mathrm{I}K(t)P_\mathrm{I}$ admits a certain approximability condition, then
\begin{equation}
\lim_{n\rightarrow\infty}T^{-n}_\mathrm{I} \prod_{l=0}^{n-1}TF\circ(\rho^{\alpha}_{n,l+1}\times\rho^{\alpha}_{n,l})=G(\cdot,t_1)
\end{equation}
in the $C^0$ norm.
}\par
\vspace{2mm}
\noindent The functions $\rho^{\alpha}_{n,l}:[t_1,t_2]\mapsto [t_1,t_2]$, $n\in\mathbb{N}$, $0\leq l\leq n$ are defined through $\rho^{\alpha}_{n,l}(s)=t_1+(s-t_1)\sum_{k=0}^{l-1}\alpha_{n,k}$. Therefore, for every $s\in [t_1,t_2]$ and for every $n\in\mathbb{N}$ the sequence $t_1=\rho^{\alpha}_{n,0}(s)<\rho^{\alpha}_{n,1}(s)<\cdots<\rho^{\alpha}_{n,n}(s)=s$ is a partition of the interval $[t_1,s]$. \par
\vspace{3mm}
Second, assume that $A\in\bhilbert$ generates a uniformly continuous contraction semi-group and let $K\in\bhilbert$. Then, if $\gamma>0$, we are interested in the large $\gamma$ behavior of the semi-group
\begin{equation}
F_\gamma(t)\vcentcolon=\exp\bigl((\gamma A+K)t\bigr)\qquad \qquad t\in\mathbb{R}_0^+.
\label{Q2}
\end{equation}

More generally, let $\banach$ be a Banach space, $A\in\bbanach$ be a generator of a uniformly continuous contraction semi-group and let $K\in C^0([t_1,t_2],\bbanach)$. Define the propagator $F_\gamma(\cdot,\cdot)$ as the solution of the initial value problem
\begin{eqnarray}
\partial_{1}F_\gamma(u,v)&=&\bigl(\gamma A+K(u)\bigr)F_\gamma(u,v)\qquad\qquad t_{1}\leq v\leq u\leq t_{2},\nonumber\\
F_\gamma(v,v)&=&\mathbbm{1}_{\banach}.
\label{Q2b}
\end{eqnarray}
Then, the same question can be raised regarding the large $\gamma$ behavior of (\ref{Q2b}). 

For every uniformly continuous group $(T(t))_{t\in\mathbb{R}}\subset\bbanach$ we can associate an ergodic mean
\begin{equation}
\mathfrak{P}_{T}[K]\vcentcolon=\lim_{0<S\rightarrow\infty}\frac{1}{S}\int_{0}^{S}T^{-1}(s)KT(s)\,\mathrm{d}s\qquad\qquad K\in\bbanach,
\end{equation}
whenever this limit exists in the uniform topology of $\bbanach$. The generalized adiabatic theorem is the following statement.\par
\vspace{2mm}
\noindent\textbf{Theorem} (Generalized adiabatic theorem of contractions). \textit{Let $\banach$ be a Banach space with a decomposition $\banach=\banach_{\mathrm{I}}\oplus\banach_{\mathrm{C}}$ such that both $\banach_{\mathrm{I}}$ and $\banach_{\mathrm{C}}$ are closed subspaces of $\banach$ and the continuous projection $P_{\mathrm{I}}$ corresponding to $\banach_{\mathrm{I}}$ and its complementary projection $P_{\mathrm{C}}$ are norm one. Assume that the semi-group $(T(t))_{t\geq 0}\subset\bbanach$, generated by $A\in\bbanach$ is of the form $T=T_{\mathrm{I}}\oplus T_{\mathrm{C}}$, where $(T_\mathrm{I}(t))_{t\geq 0}\subset\mathcal{B}(\banach_{\mathrm{I}})$ is a restriction of a group of isometric isomorphisms of $\banach_\mathrm{I}$ to $\mathbb{R}_0^+$,  while $(T_\mathrm{C}(t))_{t\geq 0}\in\mathcal{B}(\banach_{\mathrm{C}})$ is a semi-group of strict contractions acting on $\mathcal{X}_\mathrm{C}$.  Let $K\in C^0([t_1,t_2],\bbanach)$, $0\leq t_1<t_2<\infty$ and let  $F_\gamma(\cdot,\cdot)$ be the solution of the initial value problem (\ref{Q2b}). Assume that for all $t\in[t_1,t_2]$, $\mathfrak{P}_{P_\mathrm{I}T_\mathrm{I}P_\mathrm{I}}[K(t)]$ exists and let $G(\cdot,t_1)$ be the solution of the initial value problem 
\begin{eqnarray}
\partial_1G(u,t_1)&=&\mathfrak{P}_{T_\mathrm{I}}[P_\mathrm{I}K(u)P_\mathrm{I}]G(u,t_1)\qquad\qquad t_{1}\leq u\leq t_{2},\nonumber\\
G(t_1,t_1)&=&P_{\mathrm{I}}.
\end{eqnarray}
Then, if $[t_1,t_2]\ni t\mapsto P_\mathrm{I}K(t)P_\mathrm{I}$ admits a certain approximability condition, then for every $t\in (t_1,t_2]$
\begin{equation}
\lim_{0<\gamma\rightarrow\infty}\left\|F_\gamma(t,t_1)-T_{\mathrm{I},\gamma}(t)G(t,t_1)\right\|_{\bbanach}=0.
\end{equation}
and the convergence is uniform on every compact subset of $(t_1,t_2]$, whenever $K$ is constant valued.
}\par
\vspace{2mm}
\noindent Here, $T_{\mathrm{I},\gamma}(t)=T_{\mathrm{I}}(\gamma t)$ for all $t\in\mathbb{R}_0^+$.

\subsection{Previous results} 
Concerning $(\ref{Q1a})$, if $T$ is a projection $P$ onto a closed subspace of $\hilbert$, and $(F(t))_{t\in\R}$ is a strongly continuous unitary group, then the limit of (\ref{Q1a}) with the choice $C_n=P$ has been studied many times in the past fifty years and there are numerous results concerning it, see for instance \cite{Exner2005,exner2007zeno,friedman1972semigroup,misra1977zeno} and references therein. Limiting our interest to a uniformly continuous unitary group with the skew-adjoint generator $-\mathrm{i}H$, then it is long known that  
\begin{equation}
\lim_{n\rightarrow\infty}\bigl(PF(t/n)\bigr)^n=P\exp\bigl(-\mathrm{i}H_{\mathrm{Z}}t\bigr)\qquad\qquad t\in\R,
\label{EQ34}
\end{equation}
where the Hamiltonian $H_\mathrm{Z}=PHP$ is called the Zeno Hamiltonian and the time evolution generated by $H_\mathrm{Z}$ is usually referred to as the Zeno dynamics. For experimental investigations of the Zeno dynamics, see \cite{balzer2002relaxationless, balzer2000quantum,bernu2008freezing,fischer2001observation,gleyzes2016quantum,
itano1990quantum,koshino2005quantum,kwiat1999high,nagels1997quantum,raimond2012quantum,raimond2010phase,schafer2014experimental,
signoles2014confined} and \cite{facchi2008quantum,home1997conceptual} for reviews of the subject from the physical point of view.

The physical interpretation of the time evolution given by (\ref{EQ34}) represents the evolution of a system whose unitary dynamics is disturbed by a sequence of projective measurements imitating a type of `continuous observation' in the large $n$ limit \cite{misra1977zeno}. Recently, generalization of this scenario to weak measurement protocols in the limited case of systems undergoing  computational time evolution in a quantum error correction framework has been studied \cite{dominy2013analysis,paz2012zeno}. Without dwelling on the fine details, \cite{dominy2013analysis} and \cite{paz2012zeno}  studied a family of contractions of the form $P\oplus T_\varepsilon\in\mathcal{B}(\mathcal{S}_1)$, $\varepsilon\in\R$ representing weak measurements on the Banach space of the first Schatten class of $\bhilbert$. Here $P$ is a bounded projection and $T_\varepsilon$ is strictly contractive for all $\varepsilon\in\R$. The unitary time evolution has been given by (\ref{EQ1}) with $K(\cdot)=-\mathrm{i}[L(\cdot),\bullet]$, where $L(\cdot)$ is a continuous family of  bounded self-adjoint operators in $\bhilbert$ and $[\bullet,\bullet]$ is the commutator of the arguments. The form of $L(\cdot)$ has been chosen to be relevant for quantum computing. It has been shown, that the limit 
\begin{equation}
\rho_\infty(t)\equiv\lim_{n\rightarrow\infty}\rho_{n}(t)=\lim_{n\rightarrow\infty}\prod_{l=0}^{n-1}T_{\varepsilon}F(t_{n,l+1},t_{n,l})\ \rho_0\qquad t_{n,l}\vcentcolon=l\frac{t}{n} \, , \;  t\in\Rop
\end{equation}
exists for all initial density matrices $\rho_0$ with respect to the $\|\cdot\|_1$ norm and the speed of convergence is bounded by 
\begin{equation}
\|\rho_n(t)-\rho_\infty(t)\|_1\leq \frac{p_1(t)\mathrm{e}^{h_1t}}{n}+\frac{p_2(t)\mathrm{e}^{h_2t}}{n}f(\varepsilon)+\mathcal{O}(n^{-2}),
\label{Lidar}
\end{equation}
where $h_1$, $h_2$ are constants depending on the dimension of the computational space and on the norms of various parts of $L\in C^{0}([0,t],\bhilbert)$. The real valued functions $p_1(\cdot)$ and $p_2(\cdot)$ are second order, $n$-independent polynomials of $t$, and $f(\varepsilon)$ is a positive function. This suggests that the convergence with respect to the topology induced by $\|\cdot\|_1$ should be uniform in $\mathcal{B}(\mathcal{S}_1)$. However, as (\ref{Lidar}) suggests, this convergence can be exponentially slow in the time variable, which could raise doubts concerning the applicability of (\ref{Lidar}) for large times in concrete practical situations. \par

If the contraction $T$ is equal to a unitary $U\in\bhilbert$ and the uniformly continuous semi-group $(F(t))_{t\geq 0}$ is  unitary, the latter being generated by $-\mathrm{i}H$, then the resulting time evolution given by (\ref{Q1a}) is called `single unitary dynamical control' in quantum computing theory. In \cite{bernad2017dynamical}, the limit of (\ref{Q1a}) was studied rigorously and the connection to ergodic means has been pointed out. In particular, it was shown that if the ergodic mean 
\begin{equation}
\mathfrak{P}_U[H]\vcentcolon=\lim_{n\rightarrow\infty}\frac{1}{n}\sum_{l=0}^{n-1}U^{*l}HU^l
\label{EQ35}
\end{equation}
exists in the uniform topology, then 
\begin{equation}
\lim_{n\rightarrow\infty}U^{*n}(UF(t/n))^n=\exp(-\mathrm{i}\mathfrak{P}_U[H]t)\qquad \qquad t\in\R
\label{EQ51}
\end{equation}
uniformly. Furthermore, for any Hamiltonian $H$ of the form $H=H_0+U^*H_1U-H_1$, where $H_0,H_1\in\bhilbert$ and $\mathfrak{P}_U[H_0] $ exists and non-vanishing, the author found
\begin{eqnarray}
&&\left\|U^{*n}(UF(t/n))^n-\exp\bigl(-\mathrm{i}\mathfrak{P}_U[H]t\bigr)\right\|_\bhilbert\nonumber\\
&&\qquad\qquad\leq \bigl(p_1(|t|)+p_2(|t|)\mathrm{e}^{2\|H_1\|_\bhilbert|t|}\bigr)\frac{e^{\|H\|_\bhilbert|t|}}{n}+\mathcal{O}(n^{-2}),
\end{eqnarray}
where $p_1(\cdot)$ and $p_2(\cdot)$ are $n$-independent third and second order polynomials of $|t|$, respectively. For a general $H\in\bhilbert$, for which $\mathfrak{P}_U[H]$ exists, the convergence of (\ref{EQ51}) has been also proved. Similar statements have been proved concerning the in-time non-equidistant control protocols.  \par

If the spectrum $\sigma_U$ of $U$ contains only isolated eigenvalues, then the ergodic mean (\ref{EQ35}) exists for all $H\in\bhilbert$ and it is equal to
\begin{equation}
\mathfrak{P}_U[H]=\sum_{\lambda\in\sigma_U}P_\lambda HP_{\lambda},
\label{EQ53} 
\end{equation}
where $P_\lambda$ is the spectral projection of $\lambda\in\sigma_U$. Surprisingly, this ergodic mean shows up if one considers the large $\gamma>0$ behavior of the semi-group $(F_\gamma(t))_{t\geq 0}$ generated by the sum $\gamma A+B$, where $A,B\in\bhilbert$, $\hilbert$ being a finite dimensional Hilbert-space and the spectrum $\sigma_A$ of $A$ is contained in $\mathbb{C}^{-}$, such that its purely imaginary eigenvalues are diagonalizable: In a recent study \cite{burgarth2018generalized}, the authors found that if these conditions on $A$ are fulfilled, then
\begin{eqnarray}
\qquad\lim_{0<\gamma\rightarrow\infty}\left\|\exp\bigl((\gamma A+B)t\bigr)-P_\mathrm{I}\exp(\gamma A t)\exp\left(\sum_{\lambda\in \sigma_{A}\cap\,\mathrm{i}\R}P_{\lambda}BP_\lambda t\right)\right\|_\bhilbert=0
\label{Facchi1}
\end{eqnarray}
holds for all $t\in \Rop$, where 
\begin{equation}
P_\mathrm{I}=\sum_{\lambda\in \sigma_{A}\cap\,\mathrm{i}\R}P_{\lambda}.
\end{equation}
Now, assume that $A$ shares the assumptions above, let us denote its generated semi-group by $(T(t))_{t\geq 0}$ and define $P_\mathrm{C}\vcentcolon=\mathbbm{1}_\hilbert-P_\mathrm{I}$. Then, $\hilbert$ decomposes as $\hilbert=P_\mathrm{I}\hilbert\oplus P_\mathrm{C}\hilbert$ with a corresponding block-diagonal decomposition $T(\cdot)=T(\cdot)_{\mathrm{I}}\oplus T(\cdot)_{\mathrm{C}}$, where the semi-groups $(T_\mathrm{I}(t))_{t\geq 0}$ and $(T_{\mathrm{C}}(t))_{t\geq 0}$ are acting on $P_\mathrm{I}\hilbert$ and $P_\mathrm{C}\hilbert$, respectively. Furthermore, they satisfy
\begin{eqnarray}
T^{\ }_\mathrm{I}(t)T^*_\mathrm{I}(t)&=&T^*_\mathrm{I}(t)T^{\ }_\mathrm{I}(t)=\mathbbm{1}_{P_\mathrm{I}\hilbert}\qquad t\in\R,\nonumber\\
\|T_\mathrm{C}(t)\|_\bhilbert&\leq& M\exp(-\omega t) \ \,\qquad\qquad t\in\Rop,
\end{eqnarray}
with some $M,\omega\in\Rp$ and $M\geq 1$. Then, we observe 
\begin{eqnarray}
&&\left\|\frac{1}{S}\int_{0}^{S}T^*(s)BT(s)\,\mathrm{d}s-\frac{1}{S}\int_{0}^{S}T_{\mathrm{I}}^*(s)BT_{\mathrm{I}}(s)\,\mathrm{d}s\right\|_\bhilbert\nonumber\\
&&\qquad\qquad\qquad\qquad\qquad\qquad\leq \frac{2\|B\|_\bhilbert M}{S}\int_{0}^S\mathrm{e}^{-\omega s}\,\mathrm{d}s+\frac{\|B\|_\bhilbert M^2}{S}\int_{0}^S\mathrm{e}^{-2\omega s}\,\mathrm{d}s,
\end{eqnarray}
thus 
\begin{equation}
\mathfrak{P}_{T_\mathrm{I}}[B]\vcentcolon=\lim_{S\rightarrow\infty}\frac{1}{S}\int_{0}^{S}T_\mathrm{I}^*(s)BT_\mathrm{I}(s)\,\mathrm{d}s=\lim_{S\rightarrow\infty}\frac{1}{S}\int_{0}^{S}T^*(s)BT(s)\,\mathrm{d}s,
\end{equation}
where $\mathfrak{P}_{T_\mathrm{I}}[B]$ is the ergodic mean of $B$ with respect to the unitary group $(F_\mathrm{I}(t))_{t\in\R}$. This can be calculated easily using (\ref{EQ53}):
\begin{equation}
\mathfrak{P}_{T_\mathrm{I}}[B]=\sum_{\lambda\in \sigma_{A}\cap\,\mathrm{i}\R}P_{\lambda}BP_\lambda.
\label{EQ36}
\end{equation}
With this observation in hand, the speed of convergence of (\ref{Facchi1}) found by \cite{burgarth2018generalized} can be written as
\begin{eqnarray}
&&\Biggl\|\exp\bigl((\gamma A+B)t\bigr)-P_\mathrm{I}\exp(\gamma A t)\exp\bigl(\mathfrak{P}_{T_\mathrm{I}}[B]t\bigr)\Biggr\|_\bhilbert\nonumber\\[10pt]
&&\qquad\qquad\qquad\leq \frac{c_{A,B}(M_A+1)}{\gamma} \frac {M_A\|B\|_\bhilbert\mathrm{e}^{tM_{A}\|B\|_\bhilbert}-\|\mathfrak{P}_{T_\mathrm{I}}[B]\|_\bhilbert\mathrm{e}^{t\|\mathfrak{P}_{T_\mathrm{I}}[B]\|_\bhilbert}}{M_A\|B\|_\bhilbert-\|\mathfrak{P}_{T_\mathrm{I}}[B]\|_\bhilbert}\nonumber\\[5pt]
&&\qquad\qquad\qquad\qquad+\frac{M_A\mathrm{e}^{M_A\|B\|_\bhilbert}}{\gamma}\int_{0}^{\infty}\mathrm{e}^{-\eta s}p(s)\,\mathrm{d}s
+\,\mathrm{e}^{-\gamma t}p(\gamma t),
\label{EQ52}
\end{eqnarray}
where $c_{A,B}$ is a positive constant, $\eta=\min\{|\Re(\lambda)|^2:\,\lambda\in\sigma_{A}\setminus \mathrm{i}\R\}$, $M_A\geq 1$ and $p: \Rop\rightarrow \Rp$ satisfy the inequalities
\begin{equation}
\|F(t)\|_\bhilbert\leq M_A,\qquad\qquad \|F(t)P_\mathrm{C}\|_\bhilbert\leq \mathrm{e}^{-\eta t}p(t),\qquad\qquad t\in \Rop.
\end{equation} 
If $A$ is self-adjoint, then, provided that $\sigma_{A}\subset \mathbb{C}^-$, we can set $M_A=1$, $p\equiv 1$ and take the limit $\eta\rightarrow 0$ in (\ref{EQ52}), so 
\begin{eqnarray}
&&\Biggl\|\exp\bigl((\gamma A+B)t\bigr)-P_\mathrm{I}\exp(\gamma A t)\exp\left(\mathfrak{P}_{T_\mathrm{I}}[B]t\right)\Biggr\|_\bhilbert\qquad\qquad\qquad\qquad\qquad\nonumber\\[10pt]
&&\qquad\qquad\qquad\qquad\leq \frac{2c_{A,B}}{\gamma} \frac {\|B\|_\bhilbert\mathrm{e}^{t\|B\|_\bhilbert}-\|\mathfrak{P}_{T_\mathrm{I}}[B]\|_\bhilbert\mathrm{e}^{t\|\mathfrak{P}_{T_\mathrm{I}}[B]\|_\bhilbert}}{\|B\|_\bhilbert-\|\mathfrak{P}_{T_\mathrm{I}}[B]\|_\bhilbert}+\mathrm{e}^{-\gamma t}.
\end{eqnarray}

The same conclusions arises if one studies - instead of the uniformly continuous semi-group $\exp((\gamma A+B)t)$, but the products 
\begin{equation}
(\mathcal{E}\exp(\mathcal{L}t/n))^n,
\label{EQ68}
\end{equation}
where $\mathcal{E}$ is a completely positive operator acting on a Hilbert-Schmidt space $\mathcal{H}_{\mathrm{HS}}$ of finite dimensional square matrices and $\mathcal{L}$ is a generator of a completely positive trace preserving semi-group on the same space. The spectrum of $\mathcal{E}$ is contained within the closed unit disk $\mathbb{D}$ and all the eigenvalues located on $\partial\mathbb{D}$ are diagonalizable within $\mathcal{B}(\mathcal{H}_{\mathrm{HS}})$. Again, we have the decomposition $\mathcal{E}=\mathcal{E}_{I}\oplus\mathcal{E}_C$ given by the complementary projections 
\begin{equation}
P_\mathrm{I}=\sum_{\lambda\in \sigma_{\mathcal{E}}\cap \partial\mathbb{D}}P_{\lambda} \, , \quad P_{\mathrm{C}}=\mathbbm{1}_{\mathcal{H}_{\mathrm{HS}}}-P_{\mathrm{I}}.
\end{equation}
It has been shown \cite{burgarth2018quantum}, that the limit of the product (\ref{EQ68}) is given by 
\begin{equation}
\lim_{n\rightarrow\infty}\mathcal{E}_{\mathrm{I}}^{-n}(\mathcal{E}\exp(\mathcal{L}t/n))^n=\exp\left(\mathcal{L}_{\mathrm{Z}}t\right)\, , \quad \mathcal{L}_{\mathrm{Z}}=\sum_{\lambda\in \sigma_{\mathcal{E}}\cap \partial\mathbb{D}}P_\lambda \mathcal{L}P_\lambda .
\label{EQ69}
\end{equation}
A short calculation gives 
\begin{equation}
\mathfrak{P}_{\mathcal{E}_\mathrm{I}}[\mathcal{L}]\vcentcolon=\lim_{n\rightarrow\infty}\frac{1}{n}\sum_{l=0}^{n-1}{\mathcal{E}_\mathrm{I}}^{*l}\mathcal{L}{\mathcal{E}_\mathrm{I}}^{l}=\mathcal{L}_{\mathrm{Z}},
\label{EQ70}
\end{equation}
and the convergence is uniform in $\mathcal{B}(\mathcal{H}_{\mathrm{HS}})$. Therefore, an ergodic mean $\mathfrak{P}_{\mathcal{E}_\mathrm{I}}[\mathcal{L}]$ pops up again in the limit (\ref{EQ69}).
 
In this study, we give further insights into the appearance of the ergodic means such like (\ref{EQ35}), (\ref{EQ36}) and (\ref{EQ70}). The observation that the Zeno Dynamics, single unitary dynamical control and adiabatic theorems are interrelated is not new and dates back to at least \cite{facchi2004unification} and \cite{facchi2002quantum}. It is surprising that an enlightening, rigorous mathematical explanation of this connection in arbitrary Banach spaces has been missing, nonetheless the correspondence has been verified even experimentally \cite{schafer2014experimental}. Fortunately, the recent results of \cite{burgarth2018generalized} and \cite{burgarth2018quantum} give a well developed foundation of the theory in finite dimension.

\subsection{Content and structure of the paper} 
In this paper, we take further steps into the direction of a complete mathematical description of the phenomenon and widen its range of applicability to systems whose time evolution propagator is generated by continuously time dependent, bounded operators. To avoid technical difficulties, we concentrate here on uniform convergence and assume contractivity of the control operator $T$ throughout the paper. Uniform convergence appears to be a rather restrictive condition, but from an experimental point of view, it is well confirmed: In most cases, the initial datum of systems coupled to their environments are not known exactly. 

The paper is organized as follows. In section \ref{SEC2}, a generalization of ergodic means to time dependent operator families is introduced and some of its elementary properties are discussed. In section \ref{SEC3}, we introduce a weakening of the definition of growth bound of a semi-group \cite{engel1999one} which allows us to prove Lyapunov-type (exponential) upper bounds on products of the form (\ref{EQ1}) by using the Trotter product formula. This definition agrees with the usual one if the semi-group in question does not have a `hump' \cite{moler2003nineteen} and clearly weaker if it does. In many physical situations, the Banach space $\banach$ on which (\ref{EQ1}) acts as well as the contraction $T$ admit a decomposition $\banach=\banach_{\mathrm{I}}\oplus \banach_{\mathrm{C}}$ and $T=T_{\mathrm{I}}\oplus T_{\mathrm{C}}$, where $T_{\mathrm{I}}:\banach_{\mathrm{I}}\rightarrow \banach_{\mathrm{I}}$ is an isometric isomorphism and $T_{\mathrm{C}}:\banach_{\mathrm{C}}\rightarrow\banach_{\mathrm{C}}$ is a strict contraction. In section \ref{SEC4}, we present technical results which enables us to show that the existence of the limit of (\ref{EQ1}) depends only on the isometric part of $T$. We prove the main theorems in section \ref{SEC5} and section \ref{SEC6}. Some frequently used facts are listed in the appendix. 

\subsection{Notation}
Throughout the paper, indices of products of operators descend from left to the right, i.e.~for every integral pair $n<m$ and any sequence $K_n,\dots, K_m$ of operators, the product $\prod_{l=n}^{m}K_l$ is meant to be equal to $K_mK_{m-1}\cdots K_n$. The empty product always takes the unit value and the empty sum is equal to zero. The number sets $\mathbb{N}_0$, $\mathbb{R}^+$ and $\mathbb{R}^+_0$ denote the non-negative integers, positive reals and non-negative reals, respectively. If $\banach$ is a Banach space, then the index of the operator norm $\|\cdot\|_{\bbanach}$ is dropped in the proofs, but not in the statements. Therefore, if $\|\cdot\|$ appears in a proof, it always means $\|\cdot\|_{\bbanach}$.

\section{Ergodic mean in \texorpdfstring{$C([t_1,t_2],\bbanach)$}{C}}
\label{SEC2}
Application of operator theoretic methods in ergodic theory originated in the work of J.~von Neumann \cite{neumann1932proof}. After that, the subject has been grown enormously, see \cite{eisner2015operator} for an encyclopedic survey of the field. Here, we introduce an ergodic mean in $C^0([t_1,t_2],\bbanach)$ that is suitable for our purposes and consider some of its elementary properties involving an adaptation of L.~W.~Cohen's ergodic theorem \cite{cohen1940mean}. We have gained inspiration from the classic texts containing the elementary proof of the mean ergodic theorem of contractions on Hilbert spaces of F.~Riesz \cite{riesz1938some} and from the Banach space generalization of the theorem due to K.~Yosida \cite{yosida1938mean}, see  E.~Hopf's short book \cite{hopf1937ergodentheorie} for an early review. In the time independent case, the ergodic mean introduced in (\ref{EQ35}) resembles the entangled ergodic means, see \cite{accardi1998notions,eisner2010entangled,eisner2017pointwise,fidaleo2007entangled,fidaleo2009entangled,liebscher1999note} for recent results in this direction.

Let $[t_1,t_2]$, $t_1<t_2$ be a closed, finite interval and $\banach$ be a Banach space. The notation $C^{r}([t_1,t_2],\bbanach)$ stands for the Banach space of $r$ times continuously differentiable, $\bbanach$ valued functions with the usual $C^r$ norm:
\begin{equation}
\|K\|_{C^r}=\sum_{p=0}^{r}\|\partial_t^pK\|_{C^{0}},
\end{equation}
where $\partial_t^p$ stands for the $p$th ordinary derivative of $K(\cdot)$. The notation $C^{0,1}([t_1,t_2],$ $\bbanach)$ stands for the Banach space of Lipschitz continuous, $\bbanach$ valued functions with the usual norm:
\begin{equation}
\|K\|_{C^{0,1}}=\|K\|_{C^{0}}+[K]_1,\qquad [K]_1=\sup_{\substack{s,s'\in [t_1,t_2]\\ s\neq s'}}\frac{\|K(s)-K(s')\|_{\bbanach}}{|s-s'|},
\end{equation}
whenever $[K]_1$ is finite. We again warn the reader that in the proofs, for the sake of brevity, $\|\cdot\|$ stands for the uniform norm of $\bbanach$.

\subsection{Definition and elementary properties}
\begin{Def}
Let $\alpha=(\alpha_1,\alpha_2,\dots)$ be a sequence of sequences of positive numbers $(\alpha_{n,l})_{0\leq l<n}$. We say that $\alpha$ is admissible, if
\begin{equation}
\sum_{l=0}^{n-1}\alpha_{n,l}=1,\qquad\lim_{n\rightarrow\infty}\max_{0\leq l<n}\alpha_{n,l}=0,\qquad\lim_{n\rightarrow\infty}\sum_{l=0}^{n-2}|\alpha_{n,l+1}-\alpha_{n,l}|=0
\label{EQ47}
\end{equation}
hold for all $n\in\mathbb{N}$.
\end{Def}

For a given admissible $\alpha$, we define $\gamma^{\alpha}_{n,l}\vcentcolon=\sum_{p=0}^{l-1}\alpha_{n,p}$ and the maps $\rho^{\alpha}_{n,l}:[t_1,t_2]\mapsto[t_1,t_2]$, $\rho^{\alpha}_{n,l}(s)\vcentcolon =t_1+(s-t_1)\gamma^{\alpha}_{n,l}$ for each $n\in\mathbb{N}$ and $0\leq l<n$. 

Given an admissible $\alpha$ and an isometric isomorphism $T:\banach\rightarrow \banach$, we define the functions $\mathfrak{P}^{(n,k)}_{\alpha,T}[K]\in C([t_1,t_2],\bbanach)$, $0<k< n$, $n\in\N$ as
\begin{equation}
\mathfrak{P}^{(n,k)}_{\alpha,T}[K](s)\vcentcolon=\sum_{l=0}^{k}T^{-l}\bigl(\gamma^{\alpha}_{n,l+1}K\circ\rho^{\alpha}_{n,l+1}-\gamma^{\alpha}_{n,l}K\circ \rho^{\alpha}_{n,l}\bigr)(s)T^{l}\qquad s\in[t_1,t_2],
\end{equation}
for every $K\in C([t_1,t_2],\bbanach)$. Particularly, we use the notation $\mathfrak{P}^{(n)}_{\alpha,T}$ for $\mathfrak{P}^{(n,n-1)}_{\alpha,T}$.\par

We need some computational properties of $\mathfrak{P}^{(n,k)}_{\alpha,T}$. 

\begin{Lem} \label{LEM2.1} Let $n\in\mathbb{N}$, $0<k<n$ an integer. If $K\in C^{0,1}([t_1,t_2],\bbanach)$, then
\begin{equation}
\bigl\|\mathfrak{P}^{(n,k)}_{\alpha,T}[K](s)\bigr\|_{\bbanach}\leq (s-t_1)[K]_1+\|K\|_{C^{0}}.
\end{equation}
Furthermore, if $K\in C^{r+1}([t_1,t_2],\bbanach)$ and $0\leq p\leq r$ is an integer, then 
\begin{equation}
\bigl\|\bigl(\partial_t^p\mathfrak{P}^{(n,k)}_{\alpha,T}[K]\bigr)(s)\bigr\|_{\bbanach}\leq (s-t_1)\|\partial_t^{p+1}K\|_{C^0}+(p+1)\|\partial_t^pK\|_{C^0}.
\end{equation}
\end{Lem}

\begin{proof} Let $s\in [t_1,t_2]$, then
\begin{align}
\bigl\|\mathfrak{P}^{(n,k)}_{\alpha,T}[K](s)\bigr\|&=\Biggl\|\sum_{l=0}^{k}\gamma^{\alpha}_{n,l+1}T^{-l}\bigl(K(\rho^{\alpha}_{n,l+1}(s))-K(\rho^{\alpha}_{n,l}(s))\bigr)T^{l}\nonumber\\
&\qquad\qquad+\sum_{l=0}^{k}\left(\gamma^{\alpha}_{n,l+1}-\gamma^{\alpha}_{n,l}\right)T^{-l}K(\rho^{\alpha}_{n,l}(s))T^{l}\Biggr\|\nonumber\\
&\leq\sum_{l=0}^{k}\gamma^{\alpha}_{n,l+1}(\rho^{\alpha}_{n,l+1}(s)-\rho^{\alpha}_{n,l}(s))[K]_1+\sum_{l=0}^{k}\alpha_{n,l}\|K\|_{C^{0}}\nonumber,\\
&\leq(s-t_1)\sum_{l=0}^{k}\gamma^{\alpha}_{n,l+1}\alpha_{n,l}[K]_1+\sum_{l=0}^{n-1}\alpha_{n,l}\|K\|_{C^{0}}\nonumber,\\
&\leq(s-t_1)\sum_{l=0}^{n-1}\alpha_{n,l}[K]_1+\sum_{l=0}^{n-1}\alpha_{n,l}\|K\|_{C^{0}}\nonumber,\\
&\leq(s-t_1)[K]_1+\|K\|_{C^{0}},
\end{align}
where $0\leq \gamma^{\alpha}_{n,l}\leq 1$, $0\leq l\leq n$ has been used. The second part of the statement is proved in a similar way:
\begin{eqnarray}
\bigl\|\bigl(\partial_t^p\mathfrak{P}^{(n,k)}_{\alpha,T}[K]\bigr)(s)\bigr\|&=&\Biggl\|\sum_{l=0}^{k}\left(\gamma^{\alpha}_{n,l+1}\right)^{p+1}T^{-l}\bigl(\partial_t^pK(\rho^{\alpha}_{n,l+1}(s))-\partial_t^pK(\rho^{\alpha}_{n,l}(s))\bigr)T^{l}\nonumber\\
&&\qquad\qquad+\sum_{l=0}^{k}\left(\left(\gamma^{\alpha}_{n,l+1}\right)^{p+1}-\left(\gamma^{\alpha}_{n,l}\right)^p\right)T^{-l}\partial_t^pK(\rho^{\alpha}_{n,l}(s))T^{l}\Biggr\|\nonumber\\
&\leq&\sum_{l=0}^{k}\left(\gamma^{\alpha}_{n,l+1}\right)^{p+1}\Biggl\|T^{-l}\int_{\rho^{\alpha}_{n,l}(s)}^{\rho^{\alpha}_{n,l+1}(s)}\partial_t^{p+1}K(v)\, \mathrm{d}v\,T^{l}\Biggr\|\nonumber\\
&&\qquad\qquad+\sum_{l=0}^{k}\left(\left(\gamma^{\alpha}_{n,l+1}\right)^{p+1}-\left(\gamma^{\alpha}_{n,l}\right)^{p+1}\right)\Biggl\|T^{-l}\partial_t^pK(\rho^{\alpha}_{n,l}(s))T^{l}\Biggr\|\nonumber\\
&=&\sum_{l=0}^{k}\left(\gamma^{\alpha}_{n,l+1}\right)^{p+1}\Biggl\|\int_{\rho^{\alpha}_{n,l}(s)}^{\rho^{\alpha}_{n,l+1}(s)}\partial_t^{p+1}K(v)\, \mathrm{d}v\Biggr\|\nonumber\\
&&\qquad\qquad+\sum_{l=0}^{k}(p+1)\xi^{p}_l \alpha_{n,l}\|\partial_t^pK(\rho^{\alpha}_{n,l}(s))\|\nonumber\\
&\leq&(s-t_1)\sum_{l=0}^{k}\gamma^{\alpha}_{n,l+1}\alpha_{n,l}\|\partial_t^{p+1}K\|_{C^0}+(p+1)\sum_{l=0}^{n-1}\alpha_{n,l}\|\partial_t^pK\|_{C^0}\nonumber\\
\qquad\qquad\qquad&=&(s-t_1)\sum_{l=0}^{k}\sum_{m=0}^{l}\alpha_{n,l}\alpha_{n,m}\|\partial_t^{p+1}K\|_{C^0}+(p+1)\|\partial_t^pK\|_{C^0}\nonumber\\
&\leq&(s-t_1)\left(\sum_{l=0}^{n-1}\alpha_{n,l}\right)^2\|\partial_t^{p+1}K\|_{C^0}+(p+1)\|\partial_t^pK\|_{C^0}\nonumber\\[6pt]
&=&(s-t_1)\|\partial_t^{p+1}K\|_{C^0}+(p+1)\|\partial_t^pK\|_{C^0},
\end{eqnarray}
where $0\leq \gamma^{\alpha}_{n,l}\leq \xi_l\leq \gamma^{\alpha}_{n,{l+1}}\leq 1$ for all $0\leq l<n$ and Taylor's theorem has been used.
\end{proof}

\begin{Lem} \label{LEM2.2} Let $K\in C^{0}([t_1,t_2],\bbanach)$. Then, for any $s\in[t_1,t_2]$, $0< k<n$,
\begin{equation}
\left\|\int_{t_1}^{s}\mathfrak{P}^{(n,k)}_{\alpha,T}[K](v)\, \mathrm{d}v\right\|_{\bbanach}\leq (s-t_1)\|K\|_{C^0}.
\end{equation}
\end{Lem}
\begin{proof} Consider $K\in C^{0}([t_1,t_2],\bbanach)$ and $s\in[t_1,t_2]$. Then,  
\begin{eqnarray}
&&\left\|\int_{t_1}^{s}\mathfrak{P}^{(n,k)}_{\alpha,T}[K](v)\, \mathrm{d}v\right\|\nonumber\\
&&\qquad\qquad=\Biggl\|\sum_{l=0}^{k}T^{-l}\Biggl(\int_{t_1}^{s}\gamma^{\alpha}_{n,l+1} K(\rho^{\alpha}_{n,l+1}(v))-\gamma^{\alpha}_{n,l}K(\rho^{\alpha}_{n,l}(v))\,\mathrm{d}v\Biggr) T^{l}\Biggr\|\nonumber\\
&&\qquad\qquad=\Biggl\|\sum_{l=0}^{k}\int_{\rho^{\alpha}_{n,l}(s)}^{\rho^\alpha_{n,l+1}(s)}T^{-l}K(v)T^{l}\,\mathrm{d}v\Biggr\|\nonumber\\
&&\qquad\qquad\leq(s-t_1)\sum_{l=0}^{k}\alpha_{n,l}\|K\|_{C^{0}}\nonumber\\
&&\qquad\qquad\leq(s-t_1)\|K\|_{C^{0}}.
\end{eqnarray}
\end{proof}

\begin{Lem} \label{LEM2.3} Let $K\in C^{0}([t_1,t_2],\bbanach)$. Then, for all $0<k<n$,
\begin{eqnarray}
&&\left\|\int_{t_1}^{s}\mathfrak{P}^{(n,k)}_{\alpha,T}[K-T^{-1}KT](v)\,\mathrm{d}v\right\|\nonumber\\
&&\qquad \leq 2(s-t_1)\max_{0\leq l\leq k}{\alpha_{n,l}}\,\|K\|_{C^0}+(s-t_1)\sum_{l=0}^{k-1}|\alpha_{n,l+1}-\alpha_{n,l}|\nonumber\\
&&\qquad\qquad\times\Bigl(\|K\|_{C^0}+\sup_{v\in [\rho^{\alpha}_{n,l}(s),\rho^{\alpha}_{n,l+1}(s)]}\bigl\|K(v)-K(\rho^{\alpha}_{n,l}(s))\bigr\|\Bigr).
\end{eqnarray} 
\end{Lem}

\begin{proof}
Let $K\in C^{0}([t_1,t_2],\bbanach)$. For any $s\in [t_1,t_2]$ we have 
\begin{eqnarray}
&&\left\|\int_{t_1}^{s}\mathfrak{P}^{(n,k)}_{\alpha,T}[K-T^{-1}KT](v)\, \mathrm{d}v\right\|\nonumber\\
&&\qquad\qquad=\left\|\sum_{l=0}^{k}\int_{\rho^{\alpha}_{n,l}(s)}^{\rho^{\alpha}_{n,l+1}(s)}T^{-l}\bigl(K(v)-T^{-1}K(v)T^{\ }\bigr)T^{l}\,\mathrm{d}v\right\|\nonumber\\
&&\qquad\qquad=\Biggl\|\int_{\rho^\alpha_{n,0}(s)}^{\rho^{\alpha}_{n,1}(s)}K(v)\,\mathrm{d}v+\sum_{l=1}^{k}T^{-l}\left(\int_{\rho^{\alpha}_{n,l}(s)}^{\rho^{\alpha}_{n,l+1}(s)}K(v)\,\mathrm{d}v-\int_{\rho^{\alpha}_{n,l-1}(s)}^{\rho^{\alpha}_{n,l}(s)}K(v)\,\mathrm{d}v\right)T^{l}\nonumber\\
&&\qquad\qquad\qquad\qquad-\int_{\rho^{\alpha}_{n,k}(s)}^{\rho^{\alpha}_{n,k+1}(s)}T^{-k-1}K(v)T^{k+1}\,\mathrm{d}v\Biggr\|.
\label{EQ43}
\end{eqnarray}
Taking a closer look on the middle term, we can write 
\begin{eqnarray}
&&\Biggl\|T^{-l}\left(\int_{\rho^{\alpha}_{n,l}(s)}^{\rho^{\alpha}_{n,l+1}(s)}K(v)\,\mathrm{d}v-\int_{\rho^{\alpha}_{n,l-1}(s)}^{\rho^{\alpha}_{n,l}(s)}K(v)\,\mathrm{d}v\right)T^{l}\Biggr\|\nonumber\\
&&\qquad\qquad=\Biggl\|\int_{\rho^{\alpha}_{n,l}(s)}^{\rho^{\alpha}_{n,l+1}(s)}K(v)\,\mathrm{d}v-\int_{\rho^{\alpha}_{n,l-1}(s)}^{\rho^{\alpha}_{n,l}(s)}K(v)\,\mathrm{d}v\Biggr\|\nonumber\\
&&\qquad\qquad=\Biggl\|\bigl((\rho^{\alpha}_{n,l+1}(s)-\rho^{\alpha}_{n,l}(s))-(\rho^{\alpha}_{n,l}(s)-\rho^{\alpha}_{n,l-1}(s))\bigr)K(\rho^{\alpha}_{n,l}(s))\nonumber\\
&&\qquad\qquad\qquad\qquad+\int_{\rho^{\alpha}_{n,l}(s)}^{\rho^{\alpha}_{n,l+1}(s)}\bigl(K(v)-K(\rho^{\alpha}_{n,l}(s))\bigr)\,\mathrm{d}v\nonumber\\
&&\qquad\qquad\qquad\qquad-\int_{\rho^{\alpha}_{n,l-1}(s)}^{\rho^{\alpha}_{n,l}(s)}\bigl(K(v)-K(\rho^{\alpha}_{n,l}(s))\bigr)\,\mathrm{d}v\Biggr\|\nonumber\\
&&\qquad\qquad\leq \bigl|(\rho^{\alpha}_{n,l+1}(s)-\rho^{\alpha}_{n,l}(s))-(\rho^{\alpha}_{n,l}(s)-\rho^{\alpha}_{n,l-1}(s))\bigr|\,\bigl\|K(\rho^{\alpha}_{n,l}(s))\bigr\|\nonumber\\
&&\qquad\qquad\qquad\qquad+\int_{\rho^{\alpha}_{n,l-1}(s)}^{\rho^{\alpha}_{n,l+1}(s)}\bigl\|K(v)-K(\rho^{\alpha}_{n,l}(s))\bigr\|\,\mathrm{d}v\nonumber\\
&&\qquad\qquad\leq (s-t_1)|\alpha_{n,l}-\alpha_{n,l-1}|\,\Bigl(\bigl\|K(\rho^{\alpha}_{n,l}(s))\bigr\|\nonumber\\
&&\qquad\qquad\qquad\qquad+\sup_{v\in [\rho^{\alpha}_{n,l}(s),\rho^{\alpha}_{n,l+1}(s)]}\bigl\|K(v)-K(\rho^{\alpha}_{n,l}(s))\bigr\|\Bigr)\nonumber\\
&&\qquad\qquad\leq (s-t_1)|\alpha_{n,l}-\alpha_{n,l-1}|\,\Bigl(\bigl\|K\bigr\|_{C^0}\nonumber\\
&&\qquad\qquad\qquad\qquad+\sup_{v\in [\rho^{\alpha}_{n,l}(s),\rho^{\alpha}_{n,l+1}(s)]}\bigl\|K(v)-K(\rho^{\alpha}_{n,l}(s))\bigr\|\Bigr)
\end{eqnarray}
On the other hand,
\begin{eqnarray}
&&\left\|\int_{s_{n,0}}^{s_{n,1}} K(v)\,\mathrm{d}v-\int_{s_{n,k}}^{s_{n,k+1}} T^{-k-1}K(v)T^{k+1}\,\mathrm{d}v\right\|\nonumber\\
&&\qquad\qquad\qquad\qquad\leq 2(s-t_1)\max_{0\leq l\leq k}{\alpha_{n,l}}\,\|K\|_{C^0}.
\end{eqnarray} 
Therefore,
\begin{eqnarray*}
\left\|\int_{t_1}^{s}\mathfrak{P}^{(n,k)}_{\alpha,T}[K-T^{-1}KT]\,\mathrm{d}v\right\|\qquad\qquad\qquad\qquad\qquad\qquad\qquad
\end{eqnarray*}
\begin{eqnarray}
&&\qquad \leq 2(s-t_1)\max_{0\leq l\leq k}{\alpha_{n,l}}\,\|K\|_{C^0}+(s-t_1)\sum_{l=0}^{k-1}\Bigl|\alpha_{n,l+1}-\alpha_{n,l}\Bigr|\nonumber\\
&&\qquad\qquad\times\Bigl(\|K\|_{C^0}+\sup_{v\in [\rho^{\alpha}_{n,l}(s),\rho^{\alpha}_{n,l+1}(s)]}\bigl\|K(v)-K(\rho^{\alpha}_{n,l}(s))\bigr\|\Bigr).
\end{eqnarray} 
\end{proof}

\begin{Cor} \label{COR2.5} Let $K\in C^{0}([t_1,t_2],\bbanach)$. For any $s\in [t_1,t_2]$ and for all $0< k<n$,
\begin{equation}
\left\|\int_{t_1}^{s}\mathfrak{P}^{(n,k)}_{\alpha,T}[K-T^{-1}KT]\,\mathrm{d}v\right\|\leq D_{\alpha}(n)(s-t_1)\|K\|_{C_0}, 
\end{equation}
where 
\begin{equation}
D_\alpha(n)=2\max_{0\leq l<n}{\alpha_{n,l}}+3\sum_{l=0}^{n-2}|\alpha_{n,l+1}-\alpha_{n,l}|.
\end{equation}
\end{Cor}

\begin{Def}[] \label{DEF2.1} Let $\Fspace$ be a linear subspace of $C([t_1,t_2],\bbanach)$ which is complete in the norm topology induced by some $\|\cdot\|_{\Fspace}$. For a given $K\in\Fspace$, we say that its ergodic mean (with respect to the admissible $\alpha$ and the isometric isomorphism $T$) exists, if the sequence $\mathfrak{P}^{(n)}_{\alpha,T}[K]$ converges uniformly in the  $\|\cdot\|_\mathcal{F}$ norm. In this case, the limit  $\lim_{n\rightarrow\infty}\mathfrak{P}^{(n)}_{\alpha,T}[K]$ is denoted by $\mathfrak{P}_{\alpha,T}[K]$ and called the ergodic mean of $K$.
\end{Def}

\begin{Rk}
To motivate our definition, consider the isometric isomorphism $T\in\bbanach$, the continuous $K:[t_1,t_2]\rightarrow \bbanach$, the propagator $F(\cdot,\cdot)$ as the solution of the initial value problem (\ref{EQ1}), and an admissible $\alpha$. Then, for any $s\in[t_1,t_2]$, and the partition of $[t_1,s]$ defined by $t_1=\rho^\alpha_{n,0}(s)<\cdots < \rho^\alpha_{n,n}(s)=s$, $n\in\mathbb{N}$, we have 
\begin{eqnarray}
\prod_{l=0}^{n-1}TF(\rho^{\alpha}_{n,l+1}(s),\rho^{\alpha}_{n,l}(s))&=&T^{n}\prod_{l=0}^{n-1}T^{-l}F(\rho^{\alpha}_{n,l+1}(s),\rho^{\alpha}_{n,l}(s))T^{l}\nonumber\\
&=&T^{n}\left(\mathbbm{1}_{\banach}+\sum_{l=0}^{n-1}\int_{\rho^{\alpha}_{n,l}(s)}^{\rho^{\alpha}_{n,l+1}(s)}T^{-l}K(v)T^{l}\,\mathrm{d}v+\cdots\right)\nonumber\\
&=&T^{n}\left(\mathbbm{1}_{\banach}+\int_{t_1}^{s}\mathfrak{P}^{(n)}_{\alpha,T}[K](v)\,\mathrm{d}v+\cdots\right),
\end{eqnarray}
from which we anticipate that if the limit $\lim_{n\rightarrow\infty}\mathfrak{P}^{(n)}_{\alpha,T}[K]$ exists in the uniform topology of $\Fspace$, than the limit 
\begin{equation}
G(s,t_1)\vcentcolon =\lim_{n\rightarrow\infty}T^{-n}\prod_{l=0}^{n-1}T F(\rho^{\alpha}_{n,l+1}(s),\rho^{\alpha}_{n,l}(s))
\end{equation}
also exists, its convergence is uniform in the topology of $\Fspace$ and it should be equal to the solution of the initial value problem 
\begin{eqnarray}
\partial_1G(s,t_1)&=&\mathfrak{P}_{\alpha,T}[K](s)G(s,t_1)\qquad\qquad t_1\leq s\leq t_2,\nonumber\\
G(t_1,t_1)&=&\mathbbm{1}_{\banach}.
\end{eqnarray}
\par
\vspace{3mm}
Let $\mathrm{Dom}^r(\mathfrak{P}_{\alpha,T})\subseteq C^{r}([t_1,t_2],\bbanach)$, $r\in\mathbb{N}_0$ be the set which contains all $K\in  C^{r}([t_1,t_2],\bbanach)$ for which $\lim_{n\rightarrow\infty}\mathfrak{P}^{(n)}_{\alpha,T}[K]$ exists in the $C^{r}$ norm. Similarly, let $\mathrm{Dom}^{0,1}(\mathfrak{P}_{\alpha,T})\subseteq C^{0,1}([t_1,t_2]\bbanach)$ be the set which contains all $K\in C^{0,1}([t_1,t_2],$ $\bbanach)$ for which $\lim_{n\rightarrow\infty}\mathfrak{P}^{(n)}_{\alpha,T}[K]$ exists in the $C^{0,1}$ norm. Then, the notations for the kernels $\mathrm{Ker}^{r}(\mathfrak{P}_{\alpha,T})$ and $\mathrm{Ker}^{0,1}(\mathfrak{P}_{\alpha,T})$ are obvious. At last, let us denote with $P^{r}([t_1,t_2],\bbanach)$ the $\bbanach$ valued polynomials of order at most $r<\infty$, restriced to the domain $[t_1,t_2]$.
\end{Rk}

We list some elementary properties of $\mathfrak{P}_{\alpha,T}$.

\begin{Thm} \label{THM2.1} Let $\banach$ be a Banach space and $T\in\mathcal{B}(\banach)$ be an isometric isomorphism. Let $\alpha$ be an admissible sequence of sequences of positive numbers. Then,
\vspace{2mm}
\begin{enumerate}[label=\textit{\arabic*)}]
\item $\mathfrak{P}_{\alpha,T}$ maps $\mathrm{Dom}^r(\mathfrak{P}_{\alpha,T})$ into $C^{r}([t_1,t_2],\bbanach)$ and $\mathrm{Dom}^{0,1}(\mathfrak{P}_{\alpha,T})$ into\\ $C^{0,1}([t_1,t_2],\bbanach)$. 
\vspace{2mm}
\item Let $K_0\in\bbanach$ and $K\in C^0([t_1,t_2],\bbanach)$ be the constant function of value $K_0$. Then, $\mathfrak{P}_{\alpha,T}[K]$ exists if and only if $\lim_{n\rightarrow\infty}\sum_{l=0}^{n-1}\alpha_{n,l}T^{-l}KT^{l}$ exists in the uniform topology of $\bbanach$.
\vspace{2mm}
\item If $K\in\mathrm{Dom}^{0,1}(\mathfrak{P}_{\alpha,T})$, then $\|\mathfrak{P}_{\alpha,T}[K]\|_{C^0}\leq (1+t_2-t_1)\|K\|_{C^{0,1}}$.
\vspace{2mm}
\item If $K\in\mathrm{Dom}^{r+1}(\mathfrak{P}_{\alpha,T})$, then $\|\mathfrak{P}_{\alpha,T}[K]\|_{C^r}\leq (r+2+t_2-t_1)\|K\|_{C^{r+1}}$. Therefore, $\mathrm{Dom}^{r+1}(\mathfrak{P}_{\alpha,T})\cap P^{r}([t_1,t_2],\bbanach)$ is closed. 
\vspace{2mm}
\item For every $K\in C^{0}([t_1,t_2],\bbanach)$, if $K-T^{-1}KT\in \Dom^0(\mathfrak{P}_{\alpha,T})$, then\\
$\lim_{n\rightarrow \infty}\mathfrak{P}^{(n)}_{\alpha,T}[K-T^{-1}KT]=0$.
\vspace{2mm}
\item $\mathrm{Ker}^{0}(\mathfrak{P}_{\alpha,T})$ is a linear subspace of the closure of functions in $C^{0}([t_1,t_2],\mathcal{B}(\banach))$ of the form $K=L-T^{-1}LT^{\ }$, $K\in \Dom^{0}(\mathfrak{P}_{\alpha,T})$. More precisely, $\mathrm{Ker}^{0}(\mathfrak{P}_{\alpha,T})$ is equal to the linear set
\begin{equation}
\overline{\{L-T^{-1}LT:L\in C^{0}([t_1,t_2],\mathcal{B}(\banach))\}\cap \Dom^0(\mathfrak{P}_{\alpha,T})}^{C^0}\cap \Dom^0(\mathfrak{P}_{\alpha,T}).
\end{equation}
\vspace{2mm}
\item For every $K\in\mathrm{Dom}^0(\mathfrak{P}_{\alpha,T})$,  $\mathfrak{P}_{\alpha,T}[K]$ commutes with $T$.
\vspace{2mm}
\item $\mathfrak{P}_{\alpha,T}$ is an involution on $\mathrm{Dom}^{0}(\mathfrak{P}_{\alpha,T})$.
\end{enumerate}
\end{Thm}
\begin{proof}
$\ $\par
\noindent\textit{1)} Fix $n\in\mathbb{N}$ and $0\leq l<n$. Let $K\in\mathrm{Dom}^r(\mathfrak{P}_{\alpha,T})$. The functions $K\circ \rho^{\alpha}_{n,l}$ are members of $C^{r}([t_1,t_2],\bbanach)$, therefore $\mathfrak{P}^{(n)}_{\alpha,T}[K]\in C^{r}([t_1,t_2],\bbanach)$. Since  $\mathfrak{P}^{(n)}_{\alpha,T}[K]$ is convergent in the $C^r$ norm, the statement follows by uniform convergence. If $K\in\mathrm{Dom}^{0,1}(\mathfrak{P}_{\alpha,T})$, then $K\circ \rho^{\alpha}_{n,l}$ is a composition of Lipschitz continuous functions, therefore it is also Lipschitz continuous. Again, uniform convergence in the $C^{0,1}$ norm implies the statement.\par
\vspace{2mm}

\noindent\textit{2)}  Let $K_0\in\bbanach$, $s\in[t_1,t_2]$, then
\begin{eqnarray}
\mathfrak{P}^{(n)}_{\alpha,T}[K](s)=\sum_{l=0}^{n-1}\bigl(\gamma^{\alpha}_{n,l+1}-\gamma^{\alpha}_{n,l}\bigr)T^{-l}K_0T^{l}=\sum_{l=0}^{n-1}\alpha_{n,l}T^{-l}K_0T^{l}.
\end{eqnarray}
Hence, the limit exists in the $C^0$ norm, if and only if the ergodic mean 
\begin{equation}
\mathfrak{P}_{\alpha,T}[K_0]=\lim_{n\rightarrow\infty}\sum_{l=0}^{n-1}\alpha_{n,l}T^{-l}K_0T^{l}
\end{equation}
exists.
\par
\vspace{2mm}

\noindent\textit{3)} Let $K\in\mathrm{Dom}^{0,1}(\mathfrak{P}_{\alpha,T})$. Using Lemma \ref{LEM2.1}, we obtain
\begin{equation}
\|\mathfrak{P}_{\alpha,T}[K]\|_{C^0}=\lim_{n\rightarrow \infty}\|\mathfrak{P}^{(n)}_{\alpha,T}[K]\|_{C^0}\leq (1+t_2-t_1)\|K\|_{0,1}.
\end{equation}
\par
\vspace{2mm}

\noindent\textit{4)} Let $K\in\mathrm{Dom}^{r+1}(\mathfrak{P}_{\alpha,T})$. Using Lemma \ref{LEM2.1}, we obtain
\begin{eqnarray}
&&\|\mathfrak{P}_{\alpha,T}[K]\|_{C^r}=\lim_{n\rightarrow \infty}\|\mathfrak{P}^{(n)}_{\alpha,T}[K]\|_{C^r}=\lim_{n\rightarrow\infty}\sum_{p=0}^{r}\|\partial_t^p\mathfrak{P}^{(n)}_{\alpha,T}[K]\|_{C^0}\nonumber\\
&&\qquad\qquad\quad\qquad\qquad\qquad\qquad\qquad\leq (r+2+t_2-t_1)\|K\|_{C^{r+1}},
\end{eqnarray}
from which the first part of the statement follows. 
For any $K\in P^{r}([t_1,t_2],\bbanach)$, we have $\partial_t^{r+1}K=0$, therefore $\|K\|_{C^r}=\|K\|_{C^{r+1}}$ and so $\|\mathfrak{P}_{\alpha,T}[K]\|_{C^r}\leq (r+2+t_2-t_1)\|K\|_{C^{r}}$ whenever $\mathfrak{P}_{\alpha,T}[K]$ exists. From this estimate, and the fact that $P^{r}([t_1,t_2],\bbanach)$ is closed, the second part of the statement follows.
\par
\vspace{2mm}

\noindent\textit{5)} Since $K\in C^{0}([t_1,t_2],\bbanach)$, for any $\varepsilon>0$, there exists $\delta_{\varepsilon}>0$, such that for any two $x,x'\in [t_1,t_2]$ for which $|x-x'|\leq \delta_{\varepsilon}$ holds, $\|K(x)-K(x')\|\leq \varepsilon$ also holds. Since $\alpha$ is admissible, there exists $N_{\varepsilon,s}\in\mathbb{N}$ such that 
\begin{equation}
2(s-t_1)\max_{0\leq l<n}\alpha_{n,l}<\delta_{\varepsilon}\quad \text{and}\quad (s-t_1)\sum_{l=0}^{n-2}|\alpha_{n,l+1}-\alpha_{n,l}|<\delta_{\varepsilon}
\label{EQ42}
\end{equation}
hold whenever $n>N_{\varepsilon,s}$. Thus, if $n>N_{\varepsilon,s}$ holds, then for any $0< l <n-1$ and $x,x'\in [\rho^{\alpha}_{n,l-1}(s),\rho^{\alpha}_{n,l+1}(s)]$,
\begin{equation}
|x-x'|\leq \rho^{\alpha}_{n,l+1}(s)-\rho^{\alpha}_{n,l-1}(s)=(s-t_1)(\alpha_{n,l}+\alpha_{n,l-1})\leq 2(s-t_1)\max_{0\leq l<n}\alpha_{n,l}<\delta_{\varepsilon}
\end{equation}
is satisfied, therefore $\|K(x)-K(x')\|\leq \varepsilon$. Using Lemma \ref{LEM2.3},  this enables us to write 
\begin{equation}
\left\|\int_{t_1}^{s}\mathfrak{P}^{(n)}_{\alpha,T}[K-T^{-1}KT](v)\, \mathrm{d}v\right\|\leq\varepsilon\bigl\|K\bigr\|_{C^{0}} +\varepsilon\Bigl(\bigl\|K\bigr\|_{C^{0}}+\varepsilon\Bigr)\quad n>N_{\varepsilon,s}.
\label{EQ44}
\end{equation}
That is 
\begin{equation}
\lim_{n\rightarrow\infty}\left\|\int_{t_1}^{s}\mathfrak{P}^{(n)}_{\alpha,T}[K-T^{-1}KT](v)\, \mathrm{d}v\right\|=0,
\end{equation}
but $\mathfrak{P}^{(n)}_{\alpha,T}[K]$ converges in the $C^0$ norm, so taking the limit and integration are exchangeable, thereso
\begin{equation}
\left\|\int_{t_1}^{s}\mathfrak{P}_{\alpha,T}[K-T^{-1}KT](v)\, \mathrm{d}v\right\|=0.
\end{equation}
Since $\mathfrak{P}_{\alpha,T}[K-T^{-1}KT]\in C^0([t_1,t_2],\bbanach)$ and the equality holds for all $s\in[t_1,t_2]$, we have to have $\mathfrak{P}_{\alpha,T}[K-T^{-1}KT]=0$.
\par
\vspace{2mm}

\noindent\textit{6)} Define the spaces
\begin{eqnarray}
\mathcal{M}&\vcentcolon=&\{L-T^{-1}LT^{\ }:L\in C^{0}([t_1,t_2],\bbanach)\}\cap \Dom^0(\mathfrak{P}_{\alpha,T}),\nonumber\\
\overline{\mathcal{M}}&\vcentcolon=&\overline{\{L-T^{-1}LT^{\ }:L\in C^{0}([t_1,t_2],\bbanach)\}\cap \Dom^0(\mathfrak{P}_{\alpha,T})}^{C^0}.
\end{eqnarray}
Note that \textit{5)} implies $\mathfrak{P}_{\alpha,T}[\mathcal{M}]=0$. Let $K$ be a member of $\overline{\mathcal{M}}\cap \Dom^0(\mathfrak{P}_{\alpha,T})$, let $K_\varepsilon\in\mathcal{M}$ such that $\|K-K_\varepsilon\|_{C^0}\leq \varepsilon$. Then, the application of Lemma \ref{LEM2.2} gives
\begin{eqnarray}
\left\|\int_{t_1}^{s}\bigl(\mathfrak{P}^{(n)}_{\alpha,T}[K]-\mathfrak{P}^{(n)}_{\alpha,T}[K_\varepsilon]\bigr)(v)\, \mathrm{d}v\right\|&\leq& (t_2-t_1)\|K-K_{\varepsilon}\|_{C^0}\nonumber\\
&\leq& (t_2-t_1)\varepsilon.
\end{eqnarray}
Uniform convergence and \textit{5)} ensures that
\begin{eqnarray}
\lim_{n\rightarrow\infty}\left\|\int_{t_1}^{s}\bigl(\mathfrak{P}^{(n)}_{\alpha,T}[K]-\mathfrak{P}^{(n)}_{\alpha,T}[K_\varepsilon]\bigr)(v)\, \mathrm{d}v\right\|&=&\left\|\int_{t_1}^{s}\bigl(\mathfrak{P}_{\alpha,T}[K]-\mathfrak{P}_{\alpha,T}[K_\varepsilon]\bigr)(v)\, \mathrm{d}v\right\|\nonumber\\
&=&\left\|\int_{t_1}^{s}\mathfrak{P}_{\alpha,T}[K](v)\, \mathrm{d}v\right\|,
\end{eqnarray}
thus 
\begin{equation}
\left\|\int_{t_1}^{s}\mathfrak{P}_{\alpha,T}[K](v)\, \mathrm{d}v\right\|\leq (t_2-t_1)\varepsilon,
\label{EQ54}
\end{equation}
so $\lim_{n\rightarrow\infty}\mathfrak{P}^{(n)}_{\alpha,T}[K]=0$ in the $C^{0}$ norm. This calculation shows that any member $K$ of $\overline{\mathcal{M}}$ has either vanishing ergodic mean or $\mathfrak{P}_{\alpha,T}[K]$ does not exists. Let 
\begin{equation}
\mathcal{M}'\vcentcolon= \Ker^{0}(\mathfrak{P}_{\alpha,T})\setminus \bigl(\overline{\mathcal{M}}\cap \Dom^0(\mathfrak{P}_{\alpha,T})\bigr).
\end{equation}
We prove that $\Ker^{0}(\mathfrak{P}_{\alpha,T})\subseteq \overline{\mathcal{M}}$, that is $\mathcal{M}'$ is empty. Let $K\in \mathcal{M'}$. By definition of $\mathcal{M}'$, $K\neq 0$ must hold. Provided by the Hahn-Banach separation theorem, there exists $\varphi\in C^0([t_1,t_2],\bbanach)^*$, for which $\varphi(K)=1$  and $\varphi$ vanishes on all members of $\overline{\mathcal{M}}$. Then, 
\begin{eqnarray}
\varphi(K)&=&\varphi(K)+\varphi(T^{-l}KT^{l})-\varphi(T^{-l}KT^{l})\nonumber\\
&=&\varphi(K-T^{-l}KT^{l})+\varphi(T^{-l}KT^{l})\nonumber\\
&=&\sum_{k=0}^{l-1}\varphi\bigl(T^{-k}KT^{k}-T^{-1}(T^{-k}KT^{k})T\bigr)+\varphi(T^{-l}KT^{l})\nonumber\\
&=&\varphi(T^{-l}KT^{l})
\label{EQ48}
\end{eqnarray}
for all $l\in\mathbb{N}_0$, provided that whenever $\mathfrak{P}_{\alpha,T}[K]$ exists, then $\mathfrak{P}_{\alpha,T}[T^{-k}KT^k]$ also exists, therefore $T^{-k}KT^{k}-T^{-1}(T^{-k}KT^{k})T\in \mathcal{M}$ for all $k\in\mathbb{N}_0$. Note that
\begin{eqnarray}
\mathfrak{P}^{(n)}_{\alpha,T}[K]&=&-\gamma^{\alpha}_{n,0}K\circ\rho^{\alpha}_{n,0}\nonumber\\
&&\qquad\qquad+\sum_{l=1}^{n-1}\gamma_{n,l}^{\alpha}\Bigl(T^{-(l-1)}\bigl(K\circ\rho^{\alpha}_{n,l}\bigr)T^{l-1}-T^{-l}\bigl(K\circ\rho^{\alpha}_{n,l}\bigr)T^{l}\Bigr)\nonumber\\
&&\qquad\qquad+\gamma^{\alpha}_{n,n}T^{-n}\bigl(K\circ\rho^{\alpha}_{n,n}\bigl)T^{n}.
\label{EQ49}
\end{eqnarray}
The middle term in the sum of the right hand side of (\ref{EQ49}) belongs to $\{L-T^{-1}LT^{\ }:L\in C^{0}([t_1,t_2],\bbanach)\}$. To show that it belongs to $\mathcal{M}$, fix $n\in\mathbb{N}$ and $0<l<n$. The crucial observation is that for any $m\in\mathbb{N}$ and $0\leq k<m$, $\rho^{\alpha}_{m,k}\circ\rho^{\alpha}_{n,l}=\rho^{\alpha}_{n,l}\circ\rho^{\alpha}_{m,k}$ holds: For any $s\in [t_1,t_2]$,
\begin{eqnarray}
\bigl(\rho^{\alpha}_{m,k}\circ\rho^{\alpha}_{n,l}\bigr)(s)&=&t_1+(\rho^{\alpha}_{n,l}(s)-t_1)\gamma^{\alpha}_{m,k}\nonumber\\
&=&t_1+\bigl(t_1+(s-t_1)\gamma^{\alpha}_{n,l}-t_1\bigr)\gamma^{\alpha}_{m,k}\nonumber\\
&=&t_1+(s-t_1)\gamma^{\alpha}_{n,l}\gamma^{\alpha}_{m,k}\nonumber\\
&=&t_1+\bigl(t_1+(s-t_1)\gamma^{\alpha}_{m,k}-t_1)\gamma^{\alpha}_{n,l}\nonumber\\
&=&t_1+(\rho^{\alpha}_{m,k}(s)-t_1)\gamma^{\alpha}_{n,l}\nonumber\\
&=&\bigr(\rho^{\alpha}_{n,l}\circ\rho^{\alpha}_{m,k}\bigr)(s).
\end{eqnarray}
Therefore, for any $M\in C([t_1,t_2],\bbanach)$,
\begin{eqnarray}
\mathfrak{P}^{(m)}_{\alpha,T}[M]\circ \rho^{\alpha}_{n,l}&=&\sum_{k=0}^{m-1}T^{-k}\Bigl(\gamma^{\alpha}_{m,k+1}\,K\circ \rho^{\alpha}_{m,k+1}\circ \rho^{\alpha}_{n,l}-\gamma^{\alpha}_{m,k}\,K\circ \rho^{\alpha}_{m,k}\circ \rho^{\alpha}_{n,l}\Bigr)T^k\nonumber\\
&=&\sum_{k=0}^{m-1}T^{-k}\Bigl(\gamma^{\alpha}_{m,k+1}\,K\circ \rho^{\alpha}_{n,l}\circ \rho^{\alpha}_{m,k+1}-\gamma^{\alpha}_{m,k}\,K\circ \rho^{\alpha}_{n,l}\circ \rho^{\alpha}_{m,k}\Bigr)T^k\nonumber\\
&=&\mathfrak{P}^{(m)}_{\alpha,T}[M\circ\rho^{\alpha}_{n,l}] .
\end{eqnarray}
Assume that $M\in \Dom^0(\mathfrak{P}_{\alpha,T})$. Let $\varepsilon>0$ be arbitrary. Then, there exists $N_\varepsilon$ such that for all $v\in [t_1,t_2]$, and $n>N_{\varepsilon}$,
\begin{equation}
\|\mathfrak{P}_{\alpha,T}[M](v)-\mathfrak{P}^{(n)}_{\alpha,T}[M](v)\|\leq \varepsilon.
\end{equation}
Particularly, if $v=\rho^{\alpha}_{n,l}(s)$ for some $s\in [t_1,t_2]$, 
\begin{equation}
\|\bigl(\mathfrak{P}_{\alpha,T}[M]\bigr)(\rho^{\alpha}_{n,l}(s))-\bigl(\mathfrak{P}^{(n)}_{\alpha,T}[M]\bigr)(\rho^{\alpha}_{n,l}(s))\|\leq \varepsilon.
\end{equation}
But $\mathfrak{P}^{(n)}_{\alpha,T}[M]\circ\rho^{\alpha}_{n,l}=\mathfrak{P}^{(n)}_{\alpha,T}[M\circ\rho^{\alpha}_{n,l}]$, so
\begin{equation}
\|\mathfrak{P}_{\alpha,T}[M]\circ\rho^{\alpha}_{n,l}-\mathfrak{P}^{(n)}_{\alpha,T}[M\circ\rho^{\alpha}_{n,l}]\|_{C^0}\leq \varepsilon.
\end{equation}
That is, $\mathfrak{P}_{\alpha,T}[M\circ\rho^{\alpha}_{n,l}]$ exists and is equal to $\mathfrak{P}_{\alpha,T}[M]\circ\rho^{\alpha}_{n,l}$. Substituting $T^{-(l-1)}KT^{l-1}-T^{-l}KT^{l}K$ in the place of $M$, we see that $T^{-l}(K\circ\rho^{\alpha}_{n,l})T^{l}-T^{-(l+1)}(K\circ\rho^{\alpha}_{n,l})T^{l+1}$ is a member of $\mathcal{M}$, so
\begin{equation}
\varphi\bigl(T^{-l}(K\circ\rho^{\alpha}_{n,l})T^{l}-T^{-l-1}(K\circ\rho^{\alpha}_{n,l})T^{l+1}\bigr)=0.
\end{equation}
Provided that $\gamma_{n,0}^{\alpha}=0$, $\gamma_{n,n}^{\alpha}=1$ and $\rho^{\alpha}_{n,n}=\mathrm{id}_{[t_1,t_2]}$,
\begin{equation}
\varphi\left(\mathfrak{P}^{(n)}_{\alpha,T}[K]\right)=\varphi(T^{-n}KT^n)=\varphi(K),
\end{equation}
where (\ref{EQ48}) has been used. Therefore,
\begin{equation}
0=\varphi(0)=\lim_{n\rightarrow\infty}\varphi(\mathfrak{P}^{(n)}_{\alpha,T}[K])=\varphi(K)=1,
\end{equation}
which is impossible. \par
\vspace{2mm}
\noindent\textit{7)} Let $K\in\mathrm{Dom}^{0}(\mathfrak{P}_{\alpha,T})$. Then, $T^{-1}KT\in\mathrm{Dom}^{0}(\mathfrak{P}_{\alpha,T})$ also holds, so $K-T^{-1}KT\in \mathrm{Dom}^{0}(\mathfrak{P}_{\alpha,T})$. Using $T^{-1}\mathfrak{P}^{(n)}_{\alpha,T}[K]T=\mathfrak{P}^{(n)}_{\alpha,T}[T^{-1}KT]$, we can write\\
$T^{-1}\mathfrak{P}_{\alpha,T}[K]T=\mathfrak{P}_{\alpha,T}[T^{-1}KT]$ to obtain
\begin{eqnarray}
\bigl[T,\mathfrak{P}_{\alpha,T}[K]\bigr]&=&T\bigl(\mathfrak{P}_{\alpha,T}[K]-T^{-1}\mathfrak{P}_{\alpha,T}[K]T\bigr)\nonumber\\
&=&T\bigl(\mathfrak{P}_{\alpha,T}[K]-\mathfrak{P}_{\alpha,T}[T^{-1}KT]\bigr)\nonumber\\
&=&T\mathfrak{P}_{\alpha,T}[K-T^{-1}KT]\nonumber\\
&=&0,
\end{eqnarray}
due to \textit{5)}.
\par
\vspace{2mm}
\noindent\textit{8)} Consequence of \textit{7)}: Assume that $K\in\mathrm{Dom}^0(\mathfrak{P}_{\alpha,T})$, then
\begin{eqnarray}
\mathfrak{P}_{\alpha,T}\bigl[\mathfrak{P}_{\alpha,T}[K]\bigr]&=&\lim_{n\rightarrow\infty}\Biggl(\sum_{l=0}^{n-1}\gamma^{\alpha}_{n,l+1}T^{-l}\bigl(\mathfrak{P}_{\alpha,T}[K]\circ \rho^{\alpha}_{n,l+1}\bigr)T^{l}\nonumber\\
&&\qquad\qquad\qquad-\sum_{l=0}^{n-1}T^{-l}\gamma^{\alpha}_{n,l}\bigl(\mathfrak{P}_{\alpha,T}[K]\circ \rho^{\alpha}_{n,l}\bigr)T^{l}\Biggr)
\end{eqnarray}
\begin{eqnarray}
&&\qquad\qquad\qquad\quad=\lim_{n\rightarrow\infty}\sum_{l=0}^{n-1}\gamma^{\alpha}_{n,l+1}\mathfrak{P}_{\alpha,T}[K]\circ \rho^{\alpha}_{n,l+1}-\gamma^{\alpha}_{n,l}\mathfrak{P}_{\alpha,T}[K]\circ \rho^{\alpha}_{n,l}\nonumber\\
&&\qquad\qquad\qquad\quad=\lim_{n\rightarrow\infty}\left(\gamma^{\alpha}_{n,n}\mathfrak{P}_{\alpha,T}[K]\circ\rho^{\alpha}_{n,n}-\gamma^{\alpha}_{n,0}\mathfrak{P}_{\alpha,T}[K]\circ\rho^{\alpha}_{n,0}\right)\nonumber\\
&&\qquad\qquad\qquad\quad=\mathfrak{P}_{\alpha,T}[K],
\end{eqnarray}
due to $\gamma^{\alpha}_{n,0}=0$, $\gamma^{\alpha}_{n,n}=1$ and $\rho^{\alpha}_{n,n}=\mathrm{id}_{[t_1,t_2]}$.
\end{proof}

\subsection{Two existence results}
Theorem \ref{THM2.1} shows that there are significant properties of the usual ergodic mean which are lost when the more general framework of Definition \ref{DEF2.1} is introduced. The two main difficulties are: The domain of $\mathfrak{P}_{\alpha,T}$ is not in general closed and the kernel of $\mathfrak{P}_{\alpha,T}$ is not in general the closure of functions of the form $K-T^{-1}KT$, $K\in C^0([t_1,t_2],\bbanach)$, nor even necessarily closed. One can raise the question whether the domain of $\mathfrak{P}_{\alpha,T}$ is non-trivial, that is whether it contains functions beyond those members of $C^0([t_1,t_2],\bbanach)$ which commute with $T$ at every point of their argument. It is hard to answer such a question in the greatest generality, but we can give explicit examples which demonstrate the non-triviality of the domain of $\mathfrak{P}_{\alpha,T}$. These examples already demonstrate that the bounds of \textit{3)} and \textit{4)} of Theorem \ref{THM2.1} cannot be in general improved. But before moving to the concrete examples, we need two technical statements.
 
\begin{Lem}
Let $\omega\in\mathbb{C}$ of modulus one and $k\in\mathbb{N}$, $k>1$. Then,  
\begin{equation}
\lim_{n\rightarrow \infty}\sum_{l=0}^{n-1}\omega^l\frac{(l+1)^k-l^k}{n^{k}}
\end{equation}
exists and it is equal to one if $\omega=1$ and vanishes in any other case.
\end{Lem}

\begin{proof}
Let $\omega\in\mathbb{C}$ and $k\in\mathbb{N}$ as the Lemma stated. If $\omega=1$, we have a telescopic sequence after the sum, thus 
\begin{equation}
\sum_{l=0}^{n-1}\frac{(l+1)^k-l^k}{n^{k}}=\frac{n^{k}}{n^k}=1.
\end{equation}
Assume that $\omega\neq 1$. Then, 
\begin{eqnarray}
\sum_{l=0}^{n-1}\omega^l\frac{(l+1)^k-l^k}{n^{k}}=\frac{1}{n^k}\sum_{l=0}^{n-1}\sum_{m=0}^{k-1}\binom{k}{m}\omega^ll^m
\end{eqnarray}
Let us introduce the moment generating function $f(\xi)=\omega\exp(\xi)$, then
\begin{eqnarray}
\sum_{l=0}^{n-1}\omega^l\frac{(l+1)^k-l^k}{n^{k}}&=&\frac{1}{n^k}\sum_{l=0}^{n-1}\sum_{m=0}^{k-1}\binom{k}{m}\bigl[\partial^{m}_{\xi}f^l\bigr](0)\nonumber\\
&=&\frac{1}{n^k}\sum_{m=0}^{k-1}\binom{k}{m}\left[\partial^{m}_{\xi}\left(\sum_{l=0}^{n-1}f^l\right)\right](0)\nonumber\\
&=&\frac{1}{n^k}\sum_{m=0}^{k-1}\binom{k}{m}\left[\partial^{m}_{\xi}\left(\frac{1-f^n}{1-f}\right)\right](0)\nonumber\\
&=&\frac{1}{n^k}\sum_{m=0}^{k-1}\binom{k}{m}\left[\sum_{p=0}^{m}\binom{m}{p}\partial^{m-p}_{\xi}(1-f^n)\cdot \partial_\xi^{p}(1-f)^{-1}\right](0).
\label{EQ57}
\end{eqnarray}
Changing the variable of the first derivative to $\zeta=\omega\exp(\xi)$, we can rewrite it as
\begin{eqnarray}
\bigl[\partial^{m-p}_{\xi}(1-f^n)\bigr](\xi)&=&\bigl[(\zeta\partial_{\zeta})^{m-p}(1-\zeta^n)\bigr]\bigl(\zeta(\xi)\bigr)\nonumber\\
&=&-n^{m-p}\zeta^n(\xi)\nonumber\\
&=&-n^{m-p}\omega^n\exp(n\xi),
\end{eqnarray} 
for $p<m$ and, since $|(1-f^n)|(0)\leq 2$,
\begin{equation}
|\partial^{m-p}_{\xi}(1-f^n)|(0)=2n^{m-p}.
\end{equation}
for all $0\leq p\leq m$. Changing the variable of the second  derivative in the last line of (\ref{EQ57}) to $\rho=1-\omega\exp(\xi)$ we can rewrite it as
\begin{eqnarray}
\bigl[\partial^{p}_{\xi}(1-f)\bigr](\xi)&=&\left[\bigl((\rho-1)\partial_{\rho}\bigr)^p\frac{1}{\rho}\right](\rho(\xi))\nonumber\\
&=&\left[(-1)^p\frac{1}{\rho}+\sum_{q=2}^{p}\frac{q-1}{q+1}\frac{1-(-q)^{p-q+2}}{\rho^q}\right](\rho(\xi)),
\label{EQ56}
\end{eqnarray}
which can be verified using induction and the formula
\begin{equation}
\bigl((\rho-1)\partial_{\rho}\bigr)^p\frac{1}{\rho^s}=\frac{s}{\rho^{s+1}}-\frac{s}{\rho^{s}}.
\end{equation}
We have the following estimates of the parts of the sum in (\ref{EQ56}):
\begin{eqnarray}
\left|(-1)^p\frac{1}{\rho}\right|&\leq&\frac{1}{|\rho|}\nonumber\\
\left|\sum_{q=2}^{p}\frac{q-1}{q+1}\frac{1}{\rho^q}\right|&\leq&\frac{|\rho|^p-1}{|\rho|-1}-\frac{1}{|\rho|}-1<\frac{|\rho|^p-1}{|\rho|-1}-\frac{1}{|\rho|}\nonumber\\
\left|\sum_{q=2}^{p}\frac{q-1}{q+1}\frac{-(-q)^{p-q+2}}{\rho^q}\right|&\leq&\sum_{q=2}^{p}\frac{q^p}{q^q|\rho|^q}<\sum_{q=2}^{p}\frac{q^p}{|\rho|^q}<\frac{\sum_{q=0}^{p}q^p}{\min\{1,|\rho|^p\}}.
\end{eqnarray}
Setting $\xi=0$ and therefore $\rho=1-\omega$, we obtain the following upper bound:
\begin{equation}
\bigl|\partial^{p}_{\xi}(1-f)^{-1}\bigr|(0)|<\frac{|1-\omega|^p-1}{|1-\omega|-1}+\frac{1}{\Delta(\omega,p)}\sum_{q=0}^pq^p,
\end{equation}
where $\Delta(\omega,p)=\min\{1,|1-\omega|^p\}$. The sum $\sum_{q=0}^pq^p$ can be estimated using Bernoulli's formula:
\begin{eqnarray}
\sum_{q=0}^pq^p&=&\frac{1}{p+1}\sum_{j=0}^{p}\binom{p+1}{j}B^+_jp^{p+1-j}\nonumber\\
&=&\frac{p^{p+1}}{p+1}\sum_{j=0}^{p}\binom{p+1}{j}B^+_j\frac{1}{p^j}\nonumber\\
&<&\frac{p^{p+1}}{p+1}\sum_{j=0}^{p}\binom{p+1}{j}B^+_j\nonumber\\
&=&p^{p+1},
\end{eqnarray}
where the $B^+_j$'s are Bernoulli numbers of the second kind satisfying
\begin{equation}
\sum_{j=0}^{p}\binom{p+1}{j}B^+_j=p+1.
\end{equation}
These results can be used to estimate $|n^{-k}\sum_{l=0}^{n-1}\omega^l\bigl((l+1)^k-l^k\bigr)|$. Namely, we can write
\begin{eqnarray*}
\sum_{m=0}^{k-1}\sum_{p=0}^{m}\binom{k}{m}\binom{m}{p}2n^{m-p}\frac{|1-\omega|^p}{|1-\omega|-1}=\frac{2\sum_{s=0}^1(-1)^{s+1}(|1-\omega|+n+s)^k}{|1-\omega|-1},
\end{eqnarray*}
\begin{eqnarray}
\sum_{m=0}^{k-1}\sum_{p=0}^{m}\binom{k}{m}\binom{m}{p}2n^{m-p}\frac{1}{|1-\omega|-1}=\frac{2\sum_{s=0}^{1}(-1)^{s+1}(n+1+s)^k}{|1-\omega|-1}
\end{eqnarray}
and, if $n>k\geq p$ holds, 
\begin{eqnarray}
\sum_{m=0}^{k-1}\sum_{p=0}^{m}\binom{k}{m}\binom{m}{p}2n^{m-p}\frac{p^{p+1}}{\Delta(\omega,p)}&<&2\sum_{m=0}^{k-1}\sum_{p=0}^{m}\binom{k}{m}\binom{m}{p}n^{m}\frac{p}{\Delta(\omega,p)}.
\label{EQ58}
\end{eqnarray}
Note that for all $p\in\mathbb{N}$, either $\Delta(\omega,p)=1$ or $\Delta(\omega,p)=|1-\omega|^{p}$ holds, therefore $\Delta(\omega,p)=\Delta^p(\omega,1)=\vcentcolon \Delta^p(\omega)$. Continuing (\ref{EQ58}), we can write  
\begin{eqnarray}
2\sum_{m=0}^{k-1}\sum_{p=0}^{m}\binom{k}{m}\binom{m}{p}n^{m}\frac{p}{\Delta(\omega)}&=&2\sum_{m=0}^{k-1}\binom{k}{m}mn^{m}(\Delta^{-1}(\omega)+1)^{m-1}\nonumber\\
&=&\sum_{m=0}^{k-1}\binom{k}{m}m((1+\Delta^{-1}(\omega))n)^{m}\nonumber\\
&=&\frac{2k}{1+\Delta^{-1}(\omega)}\Bigl(\bigl((1+\Delta^{-1}(\omega))n+1\bigr)^k\nonumber\\
&&\quad-\bigl((1+\Delta^{-1}(\omega)\bigr)n+1)^{k-1}-\bigl((1+\Delta^{-1}(\omega))n\bigr)^k\Bigr).
\end{eqnarray}
Since
\begin{eqnarray}
&&\lim_{n\rightarrow\infty}\frac{1}{n^k}\frac{(|1-\omega|+n+1)^k-(|1-\omega|+n)^k}{|1-\omega|-1}=0,\nonumber\\
&&\lim_{n\rightarrow\infty}\frac{1}{n^k}\frac{(n+2)^k-(n+1)^k}{|1-\omega|-1}=0,\nonumber\\
&&\lim_{n\rightarrow\infty}\frac{k}{n^k}\frac{k}{1+\Delta^{-1}(\omega)}\Bigl(\bigl((1+\Delta^{-1}(\omega))n+1\bigr)^k-\bigl((1+\Delta^{-1}(\omega)\bigr)n+1)^{k-1}\nonumber\\
&&\qquad\qquad\qquad\qquad\qquad\qquad-\bigl((1+\Delta^{-1}(\omega))n\big)^k\Bigr)=0,
\end{eqnarray}
the proof of the statement of the Lemma is completed.
\end{proof}

\begin{Lem} \label{LEM2.8} Let $\omega\in\mathbb{C}$ of modulus one and $z\in\mathbb{C}$, $z\neq 0$. Then,  
\begin{equation}
\lim_{n\rightarrow \infty}\sum_{l=0}^{n-1}\omega^l
\frac{
(l+1)\exp\left(z\frac{l+1}{n}\right)-l\exp\left(z\frac{l}{n}\right)}{n}
\label{EQ59}
\end{equation}
exists and it is equal to $\exp(z)$ if $\omega=1$ and vanishes in any other case.
\end{Lem}

\begin{proof}
Let $\omega\in\mathbb{C}$ and $z\in\mathbb{N}$ as the Lemma stated. If $\omega=1$, we again have a telescopic sequence after the sum, thus 
\begin{equation}
\sum_{l=0}^{n-1}
\frac{
(l+1)\exp\left(z\frac{l+1}{n}\right)-l\exp\left(z\frac{l}{n}\right)}{n}=\exp(z).
\end{equation}
Assume that $\omega\neq 1$ and $\omega\exp(z/n)\neq 1$. We can rewrite (\ref{EQ59}) as 
\begin{eqnarray}
\qquad\sum_{l=0}^{n-1}\omega^l
\frac{
(l+1)\exp\left(z\frac{l+1}{n}\right)-l\exp\left(z\frac{l}{n}\right)}{n}&=&\frac{\mathrm{e}^{z/n}-1}{n}\sum_{l=0}^{n-1}l\omega^l\mathrm{e}^{lz/n}+\frac{1}{n}\frac{1-\omega^n\mathrm{e}^z}{1-\omega\mathrm{e}^{z/n}}.
\label{EQ60}
\end{eqnarray}
The second term tends zero uniformly on every bounded subset of $\mathbb{C}$. The first term is
\begin{eqnarray}
\frac{\mathrm{e}^{z/n}-1}{n}\sum_{l=0}^{n-1}l\omega^l\mathrm{e}^{lz/n}=(1-\mathrm{e}^{z/n})\frac{\omega^n\mathrm{e}^{z(n-1)/n}}{1-\omega\mathrm{e}^{z/n}}+\frac{\omega}{n}\frac{\omega\mathrm{e}^{z}-1}{(1-\omega\mathrm{e}^{z/n})^2},
\label{EQ61}
\end{eqnarray}
which tends to zero as $n\rightarrow\infty$ uniformly on every bounded subset of $\mathbb{C}$. 
\end{proof}

\begin{Prop} \label{PROP2.9} Let $\mathcal{H}$ be a Hilbert space, $\Fspace\subseteq\bhilbert$ be a two sided ideal of $\bhilbert$ wich is complete in the unitarily invariant norm $\|\cdot\|_\Fspace$. Define $\alpha$ through $\alpha_{n,l}=l/n$ and consider the isometry with the left multiplication with some unitary $U\in\bhilbert$ whose spectrum $\sigma_U$ contains only finitely many points. Then,
\begin{enumerate}[label=\textit{\arabic*)}]
\item For every $K\in P^r([0,t],\bFspace)$, 
\begin{equation}
\mathfrak{P}_{\alpha,U}[K]=\sum_{\lambda\in\sigma_{U}}P_\lambda K P_\lambda. 
\end{equation}
\item For every periodic $K\in C^{0}([0,2\pi],\bFspace)$, which is of the form $K(s)=\sum_{k\in\mathbb{Z}}K_k\exp(2\pi\mathrm{i} ks)$, $s\in[0,2\pi]$, where $K_k\in\bFspace$ are equal to zero but finitely many, we have   
\begin{equation}
\mathfrak{P}_{\alpha,U}[K]=\sum_{\lambda\in\sigma_{U}}P_\lambda KP_\lambda. 
\end{equation}
\end{enumerate}
\end{Prop}

\begin{proof}
$\ $\par
\noindent \textit{1)} Let $m\in\mathbb{N}_0$ be the smallest integer such that $\partial^{m+1}_{s}K=0$, $m\leq r$, that is $K(s)=\sum_{k=0}^{m}K_{k}s^k$, where $K_k\in\bFspace$. Then, 
\begin{equation}
P_\lambda \mathfrak{P}^{(n)}_{\alpha,U}[K](s)P_\mu=\sum_{k=0}^m s^k P_\lambda K_k P_\mu\sum_{l=0}^{n-1}(\overline{\lambda}\mu)^l\frac{(l+1)^{k+1}-l^{k+1}}{n^{k+1}}.
\end{equation}
for all $s\in [0,t]$ and $\lambda,\mu\in\sigma_{U}$. If $\lambda=\mu$, then $P_\lambda \mathfrak{P}^{(n)}_{\alpha,U}[K]P_\lambda=P_\lambda KP_\lambda$. Otherwise, since $K$ is a finite degree polynomial and $\sigma_U$ is finite, the limits of (\ref{EQ59}) are uniform in the degree $k$ and $\omega_{\lambda,\mu}=\overline{\lambda}\mu$. Therefore,   
\begin{equation}
\lim_{n\rightarrow\infty}\mathfrak{P}^{(n)}_{\alpha,U}[K]=\sum_{\lambda\in\sigma_{U}}P_\lambda KP_\lambda. 
\end{equation}
\par
\vspace{2mm}
\noindent\textit{2)} Assume the $K_k$'s are supported on the set of integers $\{-m,\dots,0,\dots,m\}$, $m\geq 0$. Let $n\in\mathbb{N}$ greater than the maximum of the set of integers
\begin{eqnarray}
\Lambda&=&\{l\in\mathbb{N}:\overline{\lambda}\mu=\exp(2\pi\mathrm{i} tk/l),\ \lambda\neq \mu,\ \lambda,\mu\in\sigma_U,\nonumber\\
&&\qquad\qquad t\in [0,2\pi],\ k\in\mathbb{Z},\ |k|\leq m\},
\end{eqnarray}
which is finite, since $\sigma_U$ is finite. Then, an upper bound on (\ref{EQ60}) with the substitution $\omega=\overline{\lambda}\mu$, $\lambda\neq \mu$ and $z=2\pi\mathrm{i}kt$ can be given using 
\begin{equation}
\left|1-\mathrm{e}^{2\pi\mathrm{i} tk/n}\right|=\left|2\pi\mathrm{i}k\int_0^{t/n}\mathrm{e}^{2\pi\mathrm{i}v}\ \mathrm{d}v\right|\leq \frac{2\pi |k|t}{n},
\end{equation}
and  the exact equation (\ref{EQ61}):
\begin{eqnarray}
\left|\sum_{l=0}^{n-1}(\overline{\lambda}\mu)^l
\frac{
(l+1)\exp\left(2\pi\mathrm{i} kt\frac{l+1}{n}\right)-l\exp\left(2\pi\mathrm{i} kt\frac{l}{n}\right)}{n}\right|&\leq&\frac{1}{n}\left(\frac{2\pi |k|t+2}{\Delta}+\frac{2}{\Delta^2}\right)\nonumber\\
&\leq& \frac{1}{n}\left(\frac{(2\pi)^2 m+2}{\Delta}+\frac{2}{\Delta^2}\right),
\end{eqnarray}
where 
\begin{equation}
\Delta^2=\min_{\substack{\lambda,\mu\in\sigma\\ \lambda\neq\mu}}\min_{\substack{k\in\mathbb{Z}\\|k|\leq m}} \inf_{0<t\leq 2\pi}\inf_{\substack{n\in\mathbb{N}\\ n>\max\Lambda}}2\left(1-\Re\left(\overline{\lambda}\mu\mathrm{e}^{2\pi \mathrm{i}kt/n}\right)\right)>0.
\end{equation}
This allows us to write
\begin{eqnarray}
\|P_\lambda \mathfrak{P}^{(n)}_{\alpha,U}[K](t)P_\mu\|_{\bFspace}&\leq& \|P_\lambda\|_{\bFspace}\,\|P_\mu\|_{\bFspace}\Biggl(\|K_0\|_{\bFspace}\frac{2}{n}\frac{1}{\Delta}+ \nonumber\\
&&\qquad\frac{1}{n}\left(\frac{(2\pi)^2 m+2}{\Delta}+\frac{2}{\Delta^2}\right)\sum_{\substack{k=-m\\k\neq 0}}^{m}\|K_k\|_\bFspace\Biggr),
\label{EQ63}
\end{eqnarray}
whenever $\lambda\neq\mu$. The projections $P_\lambda\in\bFspace$ are understood as left multiplications with the corresponding projection $P_\lambda\in\bhilbert$. Therefore, $P_\lambda \mathfrak{P}^{(n)}_{\alpha,U}[K]P_\mu\rightarrow 0$ in $C^0([0,2\pi],\bFspace)$ uniformly over the set $\sigma_U$, whenever $\lambda\ne\mu$. If $\lambda=\mu$, Lemma \ref{LEM2.8} shows that $P_\lambda \mathfrak{P}^{(n)}_{\alpha,U}[K]P_\lambda=P_\lambda K P_\lambda$. Therefore, the statement of the Proposition follows.
\end{proof}

\begin{Rk}
The proof of statement \textit{2)} of Proposition \ref{PROP2.9} illustrates why in general the domain and the kernel of $\mathfrak{P}_{\alpha,T}$ cannot be closed in the $C^0$ topology: the bound (\ref{EQ63}) is uniform in the $t$ variable only if $\Lambda$ is a finite set, which holds if $K_k=0$ for all $|k|>m$ with some $m\geq 0$. In this case, $\Delta$ is finite as well as the r.h.s~of (\ref{EQ63}). However, if $\{K_k\}_{k\in\mathbb{Z}}$ contains infinitely many non-vanishing members of $\bFspace$, the set $\Lambda$ is infinite. However, if one would ease the uniform convergence to point-wise convergence, it could be proved that $\lim_{n\rightarrow\infty}\mathfrak{P}^{(n)}_{\alpha,U}[K](t)$ exists in the uniform topology of $\bFspace$ if $K$ satisfies the Sobolev property
\begin{equation}
\sum_{k\in\mathbb{Z}}(1+|k|)\|K_k\|_{\bFspace}<\infty.
\end{equation}
These and other issues will be addressed in a later publication.
 \end{Rk}

\section{Weak growth bounds}
\label{SEC3}
If $\mathcal{X}$ is a Banach space, then every uniformly continuous semi-group $(F(t))_{t\geq 0}\subset\bbanach$ is of the form $F(t)=\exp(Kt)$ with some $K\in\bbanach$ \cite{engel1999one}. Note that if $F_1,F_2:\mathbb{R}_{0}^+\rightarrow\bbanach$ are uniformly continuous semi-groups generated by $K_1,K_2\in \bbanach$, then the Trotter product formula (see \cite{engel1999one}) applies without any restriction on $K_1$ and $K_2$, that is
\begin{equation}
\lim_{n\rightarrow\infty}\left\|\exp\bigl((K_1+K_2)t\bigr)-\bigl[F_1(t/n)F_2(t/n)\bigr]^n\right\|_{\bbanach}=0\qquad\qquad t\in\mathbb{R}_0^+.
\label{EQ37}
\end{equation}
We again warn the reader, that, for the sake of brevity, in what follows, the subscript of the uniform norm $\|\cdot\|_{\bbanach}$ is omitted in the proofs of the statements.

Let $\banach$ be a Banach space. Every uniformly continuous semi-group $(F(t))_{t\geq 0}\subset\bbanach$ is quasi-contractive, i.e. {$\|F(t)\|_{\bbanach}\leq \exp(w t)$ } is satisfied for all $t\in\mathbb{R}^+_0$ with some $w\in\mathbb{R}$ for which $w \leq \|K\|_\bbanach$ holds, if $K\in\bbanach$ generates $(F(t))_{t\geq 0}$. Let $\omega$ be defined by 
\begin{equation}
\omega\vcentcolon=\inf\{w\in\mathbb{R}:\exists\ 1 \leq M_{w}<\infty\ \text{st.}\ \|F(t)\|_\bbanach\leq M_{w}\mathrm{e}^{wt}\ \text{for}\ \text{all}\ t\in\mathbb{R}_0^+\},
\end{equation}
then $\omega$ is the growth bound of $(F(t))_{t\geq 0}$. It turns out that the following definition is useful to obtain error bounds with the Trotter product formula that performs better than the bound $\|F(t)\|_\bbanach\leq \exp(\|K\|_\bbanach t)$.\par
\vspace{3mm}
\begin{Def} Let $\banach$ be a Banach space and let $(F(t))_{t\geq 0}\subset\bbanach$ be a uniformly continuous semi-group generated by $K\in\bbanach$. Then, $\omega_K\in\mathbb{R}$ is the \textit{weak growth bound} of $(F(t))_{t\geq 0}$, if 
\begin{equation}
\omega_K=\inf\{w\in\mathbb{R}: \|F(t)\|_\bbanach\leq  \mathrm{e}^{wt}\ \text{for}\ \text{all}\ t\in\mathbb{R}_0^+\}.
\end{equation}
\end{Def}

\begin{Prop} \label{lem:weak_gb1}
The weak growth bound admits the following properties.
\begin{enumerate}[label=\textit{\arabic*)}]
\item For any $K\in\bbanach$,
\begin{equation}
\omega_K=\sup_{v\in\mathbb{R}^+}\frac{\ln\|\exp(Kv)\|_\bbanach}{v}.
\end{equation}
\item If $K\in C([t_1,t_2],\bbanach)$, $\interval$, then $\omega: [t_1,t_2]\rightarrow\mathbb{R}$, $s\mapsto \omega_{K(s)}$ is continuous.
\end{enumerate}
\end{Prop}

\begin{proof}
$\ $\par
\noindent\textit{1)} It is straightforward that $\sup_{v\in\mathbb{R}^+}\,v^{-1}\ln\|\exp(Kv)\|\leq \omega_K$. Assume that 
\begin{equation}
\sup_{v\in\mathbb{R}^+}\, v^{-1}\ln\|\exp(Kv)\|< \omega_K,
\end{equation}
then, there exists $0<\varepsilon$, such that 
\begin{equation}
\sup_{v\in\mathbb{R}^+} v^{-1}\ln\|\exp(Kv)\|<\omega_K-\varepsilon,
\end{equation}
so for all $v\in\mathbb{R}^+$, the inequality
\begin{equation}
\|\exp(Kv)\|<\mathrm{e}^{(\omega_K-\varepsilon )v}<\mathrm{e}^{\omega_Kv}
\end{equation}
holds and extends to $\mathbb{R}_0^+$ by continuity. But this contradicts to the fact that $\omega_K$ is the weak growth bound of $\bigl(\exp(Kv)\bigr)_{v\geq 0}$. 
\par
\vspace{2mm}
\noindent\textit{2)} Let $s,s'\in T\vcentcolon =[t_1,t_2]$, $F_s(v)\vcentcolon=\exp\bigl(K(s)v\bigr)$, $v\in\mathbb{R}_0^+$ and $\Delta\vcentcolon=K(s')-K(s)\in\bbanach$. Then, the application of the Trotter product formula (\ref{EQ37}) yields
\begin{eqnarray}
\|F_s(v)\! \!  \!  &{-}&\! \!  \!  F_{s'}(v)\| = \! \! \lim_{m\rightarrow \infty}\left\|\left[F_s\left(\frac{v}{m}\right)\right]^m-\left[F_s\left(\frac{v}{m}\right)\exp\left(\frac{\Delta v}{m}\right)\right]^m\right\| \nonumber\\
&=& \! \!  \! \lim_{m\rightarrow\infty}\left\|\sum_{l=0}^{m-1}F_s\left(\frac{v(m{-}l)}{m}\right)\left(\mathbbm{1}_{\banach}{-}\exp\left(\frac{\Delta v}{m}\right)\right)\left[F_s\left(\frac{v}{m}\right)\exp\left(\frac{\Delta v}{m}\right)\right]^l\right\|\nonumber\\
&\leq& \lim_{m\rightarrow\infty}\sum_{l=0}^{m-1}\mathrm{e}^{\omega_{K(s)}v}\exp\left(\|\Delta\|\frac{l}{m}\right)\left\|\mathbbm{1}_\banach-\exp\left(\Delta\frac{v}{m}\right)\right\|\nonumber\\
&\leq& \lim_{m\rightarrow\infty}\sum_{l=0}^{m-1}\mathrm{e}^{\omega_{K(s)}v}\exp\left(\|\Delta\|\frac{l}{m}\right)\exp\left(\|\Delta\|\frac{v}{m}\right)
\frac{\|\Delta\| v}{m}\nonumber\\
&=&\lim_{m\rightarrow\infty}\mathrm{e}^{\omega_{K(s)}v}\exp\left(\|\Delta\|\frac{v}{m}\right)
\frac{\|\Delta\| v}{m}\frac{1-\exp(\|\Delta\| v)}{1-\exp(\|\Delta\| v/m)}\nonumber\\[8pt]
&=&\mathrm{e}^{\omega_{K(s)}v}\bigl|1-\mathrm{e}^{\|\Delta\| v}\bigr|.
\label{EQ2}
\end{eqnarray}
Changing the roles of $F_s(\cdot)$ and $F_{s'}(\cdot)$ in the calculation above, we find
\begin{equation}
\left\|F_s(v)-F_{s'}(v)\right\|=\mathrm{e}^{\omega_{K(s')}v}\bigl|1-\mathrm{e}^{\|\Delta\| v}\bigr|.
\label{EQ3}
\end{equation}
Define the functions $f_{w,\delta}:\mathbb{R}_0^+\rightarrow \mathbb{R}_0^+$, $w\in T$, $\delta\geq 0$ as
\begin{equation}
f_{w,\delta}(v)=
\begin{dcases}
\delta & \text{if}\ v=0,\\
\frac{\exp(\omega_{K(w)}v)}{v\|F_w(v)\|}\bigl|1-\mathrm{e}^{\delta v}\bigr| & \text{if}\ 0<v,
\end{dcases}
\end{equation}
which are continuous on $\mathbb{R}_0^+$. Using the concavity of the logarithm, we get
\begin{eqnarray}
\frac{\ln\|F_{s}(v)\|}{v}&\leq&\frac{1}{v}\ln\Bigl(\|F_{s'}(v)\|+\|F_{s'}(v)-F_{s}(v)\|\Bigr)\nonumber\\
&\leq&\frac{\ln\|F_{s'}(v)\|}{v}+\frac{\|F_{s'}(v)-F_{s}(v)\|}{v\|F_{s'}(v)\|}\nonumber\\
&=&\frac{\ln\|F_{s'}(v)\|}{v}+f_{s',\|\Delta\|}(v),
\label{EQ38}
\end{eqnarray}
where (\ref{EQ2}) has been used. Similarly, with the help of (\ref{EQ3}) we can write
\begin{equation}
\frac{\ln\|F_{s'}(v)\|}{v}\leq \frac{\ln\|F_{s}(v)\|}{v} +f_{s,\|\Delta\|}(v).
\label{EQ39}
\end{equation}
Combining the inequalities (\ref{EQ38}) and (\ref{EQ39}), we obtain 
\begin{equation}
\frac{\ln\|F_{s'}(v)\|}{v}-f_{s,\|\Delta\|}(v)\leq \frac{\ln\|F_{s}(v)\|}{v}\leq\frac{\ln\|F_{s'}(v)\|}{v}+f_{s',\|\Delta\|}(v) .
\label{EQ6}
\end{equation}
Let $C$ be a compact subset of $\mathbb{R}^+_0$. Then, the function $g:C\times T \rightarrow \mathbb{R}^+_0$, $g(v,s)\vcentcolon=\|F_s(v)\|$ is continuous and its image is compact, therefore closed and bounded. Additionally, 
\begin{equation}
0<\inf_{(v,s)\in C\times T}\,g(v,s)
\end{equation}
holds, because if $0=\inf_{(v,s)\in C\times T}\,g(v,s)$ would be true, then, since the domain and the range of $g$ is closed, a pair $(v_0,s_0)\in C\times T$ could be found such that $0=g(v_0,s_0)=\|F_{s_0}(v_0)\|$, which contradicts to basic properties of uniformly continuous semi-groups. Therefore, $f_w$ $(w=s,s')$ has the upper bound
\begin{eqnarray}
\frac{\exp(\omega_{K(w)}v)}{v\|F_w(v)\|}\bigl|1-\mathrm{e}^{\|\Delta\| v}\bigr|&\leq& \frac{\exp\Bigl(\sup_{s\in T}\,\|K(s)\|\cdot \sup_{v\in C}\,v\Bigr)}{\inf_{(v,s)\in C\times T}\,g(v,s)}\frac{\bigl|1{-}\mathrm{e}^{\|\Delta\| v}\bigr|}{v}\nonumber\\
&=\vcentcolon& f_{C,s,s'}(v),   \  v\in C.
\end{eqnarray}
Using (\ref{EQ6}), we obtain
\begin{equation}
\frac{\ln\|F_{s'}(v)\|}{v}-f_{C,s,s'}(v)\leq \frac{\ln\|F_{s}(v)\|}{v}\leq\frac{\ln\|F_{s'}(v)\|}{v}+f_{C,s,s'}(v)
\end{equation}
for all $v\in C$. Taking the supremum in the variable $v$ over $C$, we can write
\begin{equation}
\sup_{v\in C}\frac{\ln\|F_{s'}(v)\|}{v}-\sup_{v\in C}\,f_{C,s,s'}(v)\leq \sup_{v\in C}\frac{\ln\|F_{s}(v)\|}{v}\leq\,\sup_{v\in C}\frac{\ln\|F_{s'}(v)\|}{v}+\sup_{v\in C}\,f_{C,s,s'}(v).
\end{equation}
Since $\sup_{v\in C}f_{C,s,s'}(v)\rightarrow 0$ if $s\rightarrow s'$, then 
\begin{equation}
\lim_{s'\rightarrow s} \sup_{v\in C}\frac{\ln\|F_{s'}(v)\|}{v}=\sup_{v\in C}\frac{\ln\|F_{s}(v)\|}{v}.
\end{equation}
Since $C$ is an arbitrary compact subset of $\mathbb{R}_0^+$, we can write
\begin{equation}
\lim_{s'\rightarrow s} \sup_{v\in \mathbb{R}_0^+}\frac{\ln\|F_{s'}(v)\|}{v}=\sup_{v\in \mathbb{R}_0^+}\frac{\ln\|F_{s}(v)\|}{v},
\end{equation}
which proves the continuity of $s\mapsto \omega_{K(s)}$.
\end{proof}

The next lemma constrains from above the growth bounds of uniformly continuous semi-groups whose generators are weighted sums of generators of semi-groups.

\begin{Lem} \label{lem:weak_gb2}
Let $\banach$ be a Banach space and let $\Gamma\vcentcolon=\{K_{n,l}:n\in\mathbb{N},\,0\leq l<n\}$ be a subset of $\bbanach$. Let $\alpha=(\alpha_1,\alpha_2,\dots)$ be a sequence of sequences $(\alpha_{n,l})_{0\leq l<n}$ $n\in\mathbb{N}$ of non-negative reals. Let $\omega_{n,l} \in\mathbb{R}$ be the weak growth bounds of the semi-groups $(F_{n,l}(t))_{t\geq 0}$, each of them generated by the corresponding $K_{n,l}$. Define
\begin{equation}
\Sigma_{n,l}(\Gamma,\alpha)\vcentcolon=\sum_{p=0}^{l}\alpha_{n,p}K_{n,p},\quad \Omega_{n,l}(\Gamma,\alpha)\vcentcolon=\sum_{p=0}^{l}\alpha_{n,p}\omega_{n,p},\quad 0\leq l< n,
\end{equation}
and $\Sigma_{n}(\Gamma,\alpha)\vcentcolon=\Sigma_{n,{n-1}}(\Gamma,\alpha)$, $\Omega_{n}(\Gamma,\alpha)\vcentcolon=\Omega_{n,{n-1}}(\Gamma,\alpha)$, particularly. Then, for all $n\in\mathbb{N}$ and $0\leq l<n$, the semi-groups generated by each $\Sigma_{n,l}(\Gamma,\alpha)$ satisfy 
\begin{equation}
\left\|\exp\bigl(\Sigma_{n,l}(\Gamma,\alpha)t\bigr)\right\|_\bbanach\leq \left\|\exp\left(\Omega_{n,l}(\Gamma,\alpha) t\right)\right\|_\bbanach
\label{st1L1}
\end{equation}
for all $t\in\mathbb{R}^+_0$. Furthermore, if the uniform limit 
\begin{equation}
\mathfrak{P}[\Gamma,\alpha]\vcentcolon=\lim_{n\rightarrow\infty}\Sigma_{n}(\Gamma,\alpha)
\label{a2L1}
\end{equation}
exists in $\bbanach$ and the limit
\begin{equation}
\Omega[\Gamma,\alpha]\vcentcolon=\lim_{n\rightarrow\infty}\Omega_{n}(\Gamma,\alpha)
\label{a3L1}
\end{equation}
also exists, then the semi-group generated by $\mathfrak{P}[\Gamma,\alpha]$ satisfies $\|\exp(\mathfrak{P}[\Gamma,\alpha]t)\|_\bbanach\leq \exp(\Omega[\Gamma,\alpha] t)$ for all $t\in\mathbb{R}^+_0$. 
\end{Lem}
 
\begin{proof} The statement is a straightforward consequence of the Trotter product formula. For a fixed $n\in\mathbb{N}$, we have $\Sigma_{n,0}(\Gamma,\alpha)=\alpha_{n,0}K_0$, thus $\|\exp\bigl(\Sigma_{n,0}(\Gamma,\alpha)t\bigr)\|\leq \exp(\alpha_{n,0}\omega_{n,0} t)$ provided that $\alpha$ contains only non-negative numbers. Assume that (\ref{st1L1}) holds for some $0<l<n-1$. Then, the product formula gives
\begin{eqnarray}
\left\|\exp\bigl(\Sigma_{n,l+1}(\Gamma,\alpha)t\bigr)\right\|
&=&\lim_{m\rightarrow \infty}\left\|\left[\exp\Biggl(K_{n,l}\frac{\alpha_{n,l}t}{m}\Biggr)\exp\Biggl(\Sigma_{n,l}[\Gamma,\alpha]\frac{t}{m}\Biggr)\right]^m\right\|\nonumber\\
&\leq&\lim_{m\rightarrow \infty}\left\|\exp\Biggl(K_{n,l}\frac{\alpha_{n,l}t}{m}\Biggr)\right\|^m\, \left\|\exp\Biggl(\Sigma_{n,l}[\Gamma,\alpha]\frac{t}{m}\Biggr)\right\|^m\nonumber\\
&\leq&\exp\Bigl(\alpha_{n,l}\omega_{n,l} t\Bigr)\exp\Biggl(\Omega_{n,l}[\Gamma,\alpha] t\Biggr)\nonumber\\
&=&\exp\Biggl(\Omega_{n,l+1}[\Gamma,\alpha] t\Biggr)
\end{eqnarray}
for all $t\in\mathbb{R}_0^+$. If the limits (\ref{a2L1}) and (\ref{a3L1}) exist, then the continuity of the exponential function and the norm gives
\begin{eqnarray}
\left\|\exp\bigl(\mathfrak{P}[\Gamma,\alpha]t\bigr)\right\|&=&\lim_{n\rightarrow \infty}\left\|\exp\bigl(\Sigma_{n}(\Gamma,\alpha)t\bigr)\right\| \nonumber \\
&\le& \lim_{n\rightarrow\infty}\exp(\Omega_{n}[\Gamma,\alpha] t)\nonumber\\
&=&\exp(\Omega[\Gamma,\alpha] t)
\end{eqnarray}
for all $t\in\mathbb{R}_0^+$. 
\end{proof}

\begin{Cor} Let $\hilbert$ be a Hilbert space, $U\in\bhilbert$ be unitary and $\alpha=(\alpha_1,\alpha_2,\dots)$ be a set of sequences of non-negative reals as described above. Let $\banach\subseteq \bhilbert$ be a two sided ideal endowed with the unitarily invariant norm $\|\cdot\|_\banach$ such that $(\banach,\|\cdot\|_\banach)$ forms a Banach space. Let $K\in\bbanach$ be a generator of a  uniformly continuous semi-group of weak growth bound $\omega_K$. Define ${U^*}^{p}KU^{p}$ as either to be equal to $L^p_{U^*}KL^p_{U}$ or $R^p_{U^*}KR^p_{U}$  for any $p\in\mathbb{N}_0$, where $L_{\bullet}$, $R_{\bullet}$ denote the multiplication with the argument from the left and from the right, respectively.  Then,
\begin{equation}
\left\|\exp\left(\sum_{p=0}^{l}\alpha_{n,p}U^{*p}KU^{p}\, t\right)\right\|_{\bbanach}\leq \exp\left(\sum_{p=0}^{l}\alpha_{n,p}\,\omega_K\, t\right)\qquad\qquad t\in\mathbb{R}_0^+.
\label{co1st1}
\end{equation}
Furthermore, if $\sum_{l=0}^{n-1}\alpha_{n,l}=1$ holds for all $n\in\mathbb{N}$ and the ergodic mean
\begin{equation}
\mathfrak{P}_{\alpha,U}[K]\vcentcolon=\lim_{n\rightarrow \infty}\sum_{l=0}^{n-1}\alpha_{n,l}U^{*l}KU^{l}
\label{co1a1}
\end{equation}
exists in $\bbanach$, then $\mathfrak{P}_{\alpha,U}[K]$ generates a semi-group which satisfy 
\begin{equation}
\left\|\exp\bigl(\mathfrak{P}_{\alpha,U}[K]t\bigr)\right\|_\bbanach\leq\exp(\omega_Kt)
\label{co1st2}
\end{equation} 
for all $t\in\mathbb{R}_0^+$. 
\end{Cor}

\begin{proof}
Let $\Gamma$ be defined through its members $K_{n,l}=U^{*l}KU^{l}$. Since the norm $\|\cdot\|_\mathcal{X}$ is unitarily invariant, $K\mapsto U^{*l}KU^l$ is an isometry of $\bbanach$ for every $l\in\mathbb{N}_0$. Thus,
\begin{equation}
\|\exp\bigl(K_{n,l}t\bigr)\|=\|U^{*l}\exp\bigl(Kt\bigr)U^l\|=\|\exp\bigl(Kt\bigr)\|\leq \exp(\omega_K t)
\end{equation}
holds for all $t\in\mathbb{R}_0^+$. Concerning the second statement, we can apply Lemma~\ref{lem:weak_gb2} with $\omega_{n,l}\vcentcolon=\omega_K$ to conclude (\ref{co1st1}). Furthermore, provided that $\sum_{l=0}^{n-1}\alpha_{n,l}=1$ and (\ref{co1a1}) exists, we also have (\ref{co1st2}).
\end{proof}
 
Tthe following statement will be used in section \ref{SEC6}.\par
\vspace{3mm}
\begin{Cor}\label{COR3.5}
Let $\banach$ be a Banach space and let $K\in C([t_1,t_2],\bbanach)$, $-\infty<t_1<t_2$ $<\infty$. If the weak growth bound of $(\exp(K(t)v))_{v\geq 0}$ is denoted by $\omega_{K(t)}$, then
\begin{equation}
\left\|\exp\left(v\int_{t_1}^{t_2}K(t)\,\mathrm{d}t\right)\right\|\leq \exp\left(v\int_{t_1}^{t_2}\omega_{K(t)}\,\mathrm{d}t\right)
\label{co2st1}
\end{equation}
holds for all $v\in\mathbb{R}_0^+$.

\end{Cor}

\begin{proof} 
 Let $\alpha$ be given as in Lemma~\ref{lem:weak_gb2} with the additional properties of the members
\begin{equation}
\sum_{l=0}^{n-1}\alpha_{n,l}=1,\qquad\qquad\lim_{n\rightarrow \infty}\max_{0\leq k<n}\alpha_{n,k}=0.
\label{EQ5}
\end{equation}
Define $\Gamma$ as the set containing the elements $K_{n,l}\vcentcolon=K(t_{n,l})$, where 
\begin{equation}
t_{n,l}\vcentcolon=t_1+(t_2-t_1)\sum_{p=0}^{l-1}\alpha_{n,p},\qquad \qquad 0\leq l< n.
\end{equation}
Provided that (\ref{EQ5}) holds and $K$ is continuous, the uniform limit 
\begin{equation}
\lim_{n\rightarrow \infty}\sum_{l=0}^{n-1}\alpha_{n,l}K_{n,l}=\lim_{n\rightarrow \infty}\sum_{l=0}^{n-1}\alpha_{n,l}K(t_{n,l})=\frac{1}{t_2-t_1}\int_{t_1}^{t_2}K(t)\,\mathrm{d}t
\end{equation}
exists in $\bbanach$. Using Proposition \ref{lem:weak_gb1}, $[t_1,t_2]\ni t\mapsto \omega_{K(t)}$ is continuous, therefore integrable on $[t_1,t_2]$. In the notation of Lemma~\ref{lem:weak_gb2}, we write
\begin{equation}
\lim_{n\rightarrow\infty}\Omega_n(\Gamma,\alpha)=\lim_{n\rightarrow\infty}\sum_{l=0}^{n-1}\alpha_{n,l}\omega_{n,l}=\frac{1}{t_2-t_1}\int_{t_1}^{t_2}\omega_{K(t)}\,\mathrm{d}t.
\end{equation}
Using Lemma~\ref{lem:weak_gb2}, we obtain
\begin{equation}
\left\|\exp\left(\frac{u}{t_2-t_1}\int_{t_1}^{t_2}K(t)\mathrm{d}t\right)\right\|\leq \exp\left(\frac{u}{t_2-t_1}\int_{t_1}^{t_2}\omega_{K(t)}\,\mathrm{d}t\right)
\end{equation}
for all $u\in\mathbb{R}_0^+$. Substitution of $u=(t_2-t_1)v$, where $v\in\mathbb{R}_0^+$ completes the proof. 
\end{proof}

Proposition~\ref{prop:growth_prop} concerns the growth properties of solutions $F(\cdot,\cdot)$ of the initial value problem
\begin{eqnarray}
\partial_1F(u,v)&=&K(u)F(u,v),\qquad\qquad t_1\leq v\leq u\leq t_2, \nonumber\\
F(v,v)&=&\mathbbm{1}_{\banach},
\label{EQ27}
\end{eqnarray}
where $K\in C([t_1,t_2],\bbanach)$, $\interval$.

\begin{Prop}  \label{prop:growth_prop}
\textit{Let $\banach$ be a Banach space and let $K\in C([t_1,t_2],\bbanach)$. Then, the solution $F(\cdot,\cdot)$ of (\ref{EQ27}) satisfies 
\begin{equation}
\|F(u,v)\|_\bbanach\leq \exp\left(\int_{v}^{u}\omega_{K(s)}\,\mathrm{d}s\right),\qquad \qquad t_1\leq v\leq u\leq t_2,
\end{equation}
where $\omega_{K(s)}$ is the weak growth bound of the uniformly continuous semi-group generated by $K(s)$ if $s\in [t_1,t_2]$.}
\end{Prop}
\begin{proof}
Since $[t_1,t_2]$ is closed, $K$ is uniformly continuous. Let $\varepsilon> 0$ be given, then there exists $\delta_{\varepsilon} >0$ such that for all $s,s'\in [t_1,t_2]$, if $|s-s'|<\delta_{\varepsilon}$ holds, $\|K(s')-K(s)\|<\varepsilon$ also holds. Let $s_0,\dots,s_{m_{\varepsilon}}$ defined through
\begin{equation}
s_k\vcentcolon=v+\frac{k\delta_{\varepsilon}}{m_{\varepsilon}}\quad \mathrm{if}\ 0\leq k < m_\varepsilon,\qquad s_{m_\varepsilon}\vcentcolon=u,\qquad\qquad m_{\varepsilon}\vcentcolon=\left\lceil{\frac{u-v}{\delta_{\varepsilon}}}\right\rceil.
\label{EQ29}
\end{equation}
Then, using the characteristic functions of the intervals $[s_l,s_{l+1}]$, $0\leq l< m$, we have
\begin{equation}
\sup_{s\in [s_0,s_{m_{\varepsilon}}]}\left\|K(s)-\sum_{l=0}^{m-1}K(s_l)\chi_{[s_l,s_{l+1}]}(s)\right\|<\varepsilon.
\end{equation}
We apply Peano's perturbation theorem for linear equations of which proof we write out in detail for the sake of completeness. For any $s\in [s_l,s_{l+1}]$ we have 
\begin{eqnarray}
\|F(s,s_l){-}\exp\bigl(K(s_l)s\bigr)\| \! \! \!&=& \! \! \! \left\|\int_{s_l}^{s}\bigr(K(t)F(t,s_l)-K(s_l)\exp\bigl(K(s_l)t\bigr)\,\mathrm{d}t\right\|\nonumber\\
\! \!&\leq & \! \! \! \int_{s_l}^{s}\left\|K(t)F(t,s_l)-K(t)\exp\bigl(K(s_l)t\bigr)\right\|\mathrm{d}t \nonumber \\
 &&  \! \! \! \qquad+\int_{s_l}^{s}\left\|\bigl(K(t)-K(s_l)\bigr)\exp\bigl(K(s_l)t\bigr)\right\|\mathrm{d}t\nonumber\\
\! \! \! &<& \! \! \! \|K\|_{C^0}\int_{s_l}^{s}\left\|F(t{,}s_l){-}\exp\bigl(K(s_l)t\bigr)\right\| \! \mathrm{d}t{+}\varepsilon\,\mathrm{e}^{t_2\|K\|_{C^0}}(s{-}s_l).
\end{eqnarray}
Applying Gr\"onwall's inequality, we can write
\begin{eqnarray}
\|F(s,s_l)-\exp\bigl(K(s_l)s\bigr)\|&<& \varepsilon \, \mathrm{e}^{t_2\|K\|_{C^0}}\frac{\mathrm{e}^{(s-s_l)\|K\|_{C^0}}-1}{\|K\|_{C^0}}
\label{EQ7}
\end{eqnarray}
Since $F(\cdot,\cdot)$ satisfies the integral equation
\begin{equation}
F(s',s)=\mathbbm{1}_{\banach}+\int_{s}^{s'}K(w)F(w,s)\,\mathrm{d}w, \qquad \qquad v\leq s\leq s'\leq u,
\end{equation}
its norm satisfies the integral inequality
\begin{equation}
\|F(s',s)\|\leq 1+\int_{s}^{s'}\|K(w)\|\,\|F(w,s)\|\,\mathrm{d}w\leq 1+\|K\|_{C^0}\int_{s}^{s'}\|F(w,s)\|\,\mathrm{d}w.
\label{EQ73}
\end{equation}
A further application of Gr\"onwall's inequality yields 
\begin{equation}
\|F(s_{m_{\varepsilon}},s_{l+1})\| {\le} \exp\left((s_{m_{\varepsilon}}-s_{l+1})\|K\|_{C^0}\right).
\end{equation}
Using (\ref{F1}), (\ref{EQ7}) and (\ref{EQ73}), we obtain
\begin{eqnarray}
&& \left\|F(u,v){-}\prod_{l=0}^{m_{\varepsilon}-1}\exp\bigl(K(s_l)(s_{l+1}-s_{l})\bigr) \right\|\nonumber\\
&&\qquad\qquad\qquad\leq \varepsilon\,\mathrm{e}^{\|K\|_{C^0}(u-v)}\mathrm{e}^{t_2\|K\|_{C^0}}\left\lceil{\frac{u-v}{\delta_\varepsilon}}\right\rceil \frac{\mathrm{e}^{\delta_\varepsilon\|K\|_{C^0}}-1}{\|K\|_{C^0}}\nonumber\\
&&\qquad\qquad\qquad\leq \varepsilon\,(u-v)\mathrm{e}^{\|K\|_{C^0}(u-v)}\mathrm{e}^{t_2\|K\|_{C^0}}\frac{\mathrm{e}^{\delta_\varepsilon\|K\|_{C^0}}-1}{\delta_{\varepsilon}\|K\|_{C^0}}
\label{EQ30}
\end{eqnarray}
Therefore,
\begin{eqnarray}
\|F(u,v)\|&\leq&\left\|F(u,v)-\prod_{l=0}^{m_{\varepsilon}-1} \exp\bigl(K(s_l)(s_{l+1}-s_{l})\right\|+\left\|\prod_{l=0}^{m_{\varepsilon}-1}\ \exp\bigl(K(s_l)(s_{l+1}-s_{l}) \right\|\nonumber\\
&\leq& \varepsilon\,(u-v)\mathrm{e}^{\|K\|_{C^0}(u-v)}\mathrm{e}^{t_2\|K\|_{C^0}}\frac{\mathrm{e}^{\delta_\varepsilon\|K\|_{C^0}}-1}{\delta_{\varepsilon}\|K\|_{C^0}}+\mathrm{e}^{\sum_{l=0}^{m_{\varepsilon}-1}\omega_{K(s_l)}(s_{l+1}{-}s_{l})}.
\end{eqnarray}
As $\varepsilon\rightarrow 0$ we can assume $\delta_{\varepsilon}\rightarrow 0$, so 
\begin{equation}
\lim_{\varepsilon \rightarrow 0}  \varepsilon\,(u-v)\mathrm{e}^{\|K\|_{C^0}(u-v)}\mathrm{e}^{t_2\|K\|_{C^0}}\frac{\mathrm{e}^{\delta_\varepsilon\|K\|_{C^0}}-1}{\delta_{\varepsilon}\|K\|_{C^0}}=0.
\label{EQ32}
\end{equation}
Furthermore, if $\delta_\varepsilon\rightarrow 0$, then $m_{\varepsilon}\rightarrow \infty$. Thus, Proposition \ref{lem:weak_gb1} allows us to write
\begin{equation}
\lim_{\varepsilon \rightarrow 0}\sum_{l=0}^{m_{\varepsilon}-1}\omega_{K(s_l)}(s_{l+1}-s_{l})=\int_{v}^{u}\omega_{K(s)}\,\mathrm{d}s.
\label{EQ33}
\end{equation}
Hence,
\begin{equation}
\|F(u,v)\|\leq\exp\left(\int_{v}^{u}\omega_{K(s)}\,\mathrm{d}s\right).
\end{equation}
\end{proof}

Proposition \ref{PROP3.7} is a perturbation result concerning the solutions $G(\cdot,\cdot)$, $H(\cdot,\cdot)$ of the initial value problems 
\begin{eqnarray}
&& \partial_1G(u,v)=K(u)G(u,v),\quad \partial_1H(u,v)=L(u)H(u,v),\; \; \; t_1\leq v\leq u\leq t_2,\nonumber\\
&& \quad G(v,v)=\mathbbm{1}_{\banach},\quad \qquad \; \; \; \; \quad\ \ \, H(v,v)=\mathbbm{1}_{\banach},
\label{EQ28}
\end{eqnarray}
where $K,L\in C([t_1,t_2],\bbanach)$, $\interval$.

\begin{Prop}\label{PROP3.7}
Let $\banach$ be a Banach space and let $K,L\in C([t_1,t_2],\bbanach)$, $-\infty<$ $t_1<t_2<\infty$ and $G(\cdot,\cdot)$, $H(\cdot,\cdot)$ be solutions of the initial value problems (\ref{EQ28}). Let $\omega:[t_1,t_2]\rightarrow \mathbb{R}$ be integrable on $[t_{1},t_{2}]$, which dominates the weak growth bounds $\omega_{K(\cdot)}$ and $\omega_{L(\cdot)}$ on $[t_1,t_2]$. Then, 
\begin{equation}
\|G(u,v)-H(u,v)\|_\bbanach\leq\exp\left(\int_{v}^{u}\omega(s)\,\mathrm{d}s\right)\, \int_{v}^{u}\|K(s)-L(s)\|_\bbanach\,\mathrm{d}s
\end{equation}
holds for all $t_1\leq v\leq u\leq t_2$.
\end{Prop}

\begin{proof} Let $K,L\in C([t_1,t_2],\bbanach)$. Since $[t_1,t_2]$ is closed, $K$ and $L$ are uniformly continuous on $[t_1,t_2]$, therefore for every $\varepsilon>0$, there exists $\delta_{\varepsilon}$ such that if $s,s'\in[t_1,t_2]$ and $|s-s'|\leq \delta_{\varepsilon}$ hold, then both $\|K(s)-K(s')\|\leq \varepsilon$ and $\|L(s)-L(s')\|\leq \varepsilon$ hold. Define $s_0,\dots,s_{m_\varepsilon}$ as in (\ref{EQ29}), then using the estimate (\ref{EQ30}) and the triangle inequality, we can write
\begin{eqnarray}
\left\| G(u,v)-H(u,v)\right\|&\leq&\left\|\prod_{l=0}^{m_{\varepsilon}-1}\exp\bigl(K(s_l)(s_{l+1}-s_{l})\bigr)-\prod_{l=0}^{m_{\varepsilon}-1}\exp\bigl(L(s_l)(s_{l+1}-s_{l})\bigr)\right\|\nonumber\\
&&\qquad + f(\varepsilon),
\end{eqnarray} 
where 
\begin{equation}
f(\varepsilon)=\varepsilon\,(u-v)\sum_{X\in \{K,L\}}\mathrm{e}^{\|X\|_{C^0}(u-v)}\mathrm{e}^{t_2\|X\|_{C^0}}\frac{\mathrm{e}^{\delta_\varepsilon\|X\|_{C^0}}-1}{\delta_{\varepsilon}\|X\|_{C^0}}
\end{equation}
Using (\ref{F1}), we obtain 
\begin{eqnarray}
&&\Biggl\|\prod_{l=0}^{m_{\varepsilon}-1}\exp\bigl(K(s_l)(s_{l+1}-s_{l})\bigr)-\prod_{l=0}^{m_{\varepsilon}-1}\exp\bigl(L(s_l)(s_{l+1}-s_{l})\bigr)\Biggr\|\nonumber\\
&&\qquad\qquad\qquad\qquad\leq
\exp\left(\sum_{p=0}^{m_\varepsilon{-}1}\omega(s)(s_{p+1}{-}s_p)\right)
\sum_{l=0}^{m_\varepsilon{-}1}\mathrm{e}^{-\omega(s_l)(s_{l+1}{-}s_l)}\nonumber\\
&&\qquad\qquad\qquad\qquad\qquad\times\Bigl\|\exp\bigl(K(s_l)(s_{l+1}{-}s_{l}){-}\exp\bigl(L(s_l)(s_{l+1}{-}s_{l})\bigr)\Bigl\|.
\label{EQ31}
\end{eqnarray}
We have to estimate the differences of the last line of the previous equation. For any $0\leq l <m_\varepsilon$, introduce for the sake of brevity $\Delta_{l}\vcentcolon =L(s_l)-K(s_l)\in\bbanach$ and $\sigma_{l}\vcentcolon=s_{l+1}-s_l$. Then, we can apply the Trotter product formula:
\begin{eqnarray}
\bigl{\|} \exp\left(K(s_l)\sigma_l\right) \! \! \! \! &-& \!  \!  \! \!  \exp\left(L(s_l)\sigma_l \right) \bigr{\|} \nonumber\\
&=&\! \! \lim_{n\rightarrow \infty}\Bigl\|\left[\exp\left(K(s_l)\frac{\sigma_l}{n}\right)\right]^n-\left[\exp\left(K(s_l)\frac{\sigma_l}{n}\right)\exp\left(\Delta_l\frac{\sigma_l}{n}\right)\right]^n\Bigr\|\nonumber\\
&\leq& \! \! \lim_{n\rightarrow\infty}\sum_{p=0}^{n-1}\Bigl\|\exp\left(K(s_l)\frac{\sigma_l}{n}\right)\Bigr\|^{n{-}p}\,\Bigl\|\mathbbm{1}_\banach{-}\exp\left(\Delta_l\frac{\sigma_l}{n}\right)\Bigr\|\nonumber\\
&&\qquad\qquad\times\Bigl\|\exp\left(K(s_l)\frac{\sigma_l}{n}\right)\exp\left(\Delta_l\frac{\sigma_l}{n}\right)\Bigr\|^{p}\nonumber\\
&\leq&\lim_{n\rightarrow\infty}\Bigl\|\exp\left(K(s_l)\frac{\sigma_l}{n}\right)\Bigr\|^{n}\Bigl\|\mathbbm{1}_\banach-\exp\left(\Delta_l\frac{\sigma_l}{n}\right)\Bigr\|\sum_{p=0}^{n-1}\Bigl\|\exp\left(\Delta_l\frac{\sigma_l}{n}\right)\Bigr\|^p\nonumber\\
&\leq &\lim_{n\rightarrow\infty}\mathrm{e}^{\omega(s_l)\sigma_l}\mathrm{e}^{\|\Delta_l\|\sigma_l/n}\frac{\|\Delta_l\|\sigma_l}{n}\frac{1-\exp(\|\Delta_l\|\sigma_l)}{1-\exp(\|\Delta_l\|\sigma_l/n)}\nonumber\\[6pt]
&=&\mathrm{e}^{\omega(s_l)\sigma_l}|1-\exp(\|\Delta_l\|\sigma_l)|\nonumber\\[6pt]
&\leq&\sigma_l\|\Delta_l\|\mathrm{e}^{\omega(s_l)\sigma_l}\mathrm{e}^{\|\Delta_l\|\sigma_l}.
\end{eqnarray}
This result can be substituted into (\ref{EQ31}) to obtain
\begin{eqnarray*}
&&\Biggl\|\prod_{l=0}^{m_{\varepsilon}{-}1}\exp\bigl(K(s_l)(s_{l+1}{-}s_{l})\bigr)
{-}\prod_{l=0}^{m_{\varepsilon}{-}1}\exp\bigl(L(s_l)(s_{l+1}{-}s_{l})\bigr)\Biggr\|\qquad\qquad\qquad\qquad\qquad\qquad\nonumber\\
&& \quad \leq\exp\left(\sum_{p=0}^{m_\varepsilon-1}\omega(s)(s_{p+1}-s_p)\right)\sum_{l=0}^{m_\varepsilon-1}(s_{l+1}-s_l)\|\Delta_l\|\mathrm{e}^{\|\Delta_l\|(s_{l+1}-s_{l})}
\end{eqnarray*}
\begin{eqnarray}
&& \quad \leq \exp(\delta_{\varepsilon}\max_{0{\leq} l<n}\|\Delta_l\|)\exp\left(\sum_{p=0}^{m_\varepsilon{-}1}\omega(s)(s_{p+1}{-}s_p)\right)
\sum_{l=0}^{m_\varepsilon{-}1}(s_{l+1}{-}s_l)\|K(s_{l}){-}L(s_l)\|.
\end{eqnarray}
In taking the $\varepsilon\rightarrow 0$ limit, we can assume  $\delta_{\varepsilon}\rightarrow 0$, therefore we can use (\ref{EQ32}), (\ref{EQ33}) and the fact that $s\mapsto\|K(s)-L(s)\|$ is continuous, so 
\begin{equation}
\lim_{\varepsilon\rightarrow 0}\sum_{l=0}^{m_\varepsilon-1}(s_{l+1}-s_l)\|K(s_{l})-L(s_l)\|=\int_{v}^{u}\|K(s)-L(s)\|\,\mathrm{d}s,
\end{equation}
hence
\begin{equation}
\|G(u,v)-H(u,v)\|\leq\exp\left(\int_{v}^{u}\omega(s)\,\mathrm{d}s\right)\, \int_{v}^{u}\|K(s)-L(s)\|\,\mathrm{d}s,
\end{equation}
which is the statement of the Proposition. 
\end{proof}

\section{Some technical results}
\label{SEC4}
Let $P\in\bbanach$ be a projection to a closed subspace of the Banach space $\banach$ satisfying $\|P\|_\bbanach=1$, let $K_0,\dots,K_{n-1}\in C([t_{1},t_{2}],\bbanach)$, $\interval$. Let $\mathfrak{T}_n$ be a partition of $[t_{\mathrm{1}},t_{\mathrm{2}}]$, defined by $t_{\mathrm{1}}=t_{n,0}<t_{n,1}<\dots <t_{n,n}=t_2$. Define $F_{l}(\cdot,\cdot)$, $0\leq l< n$ as the solutions of the initial value problems
\begin{eqnarray}
\partial_1F_l(u,v)&=&K_l(u)F_l(u,v)\qquad\qquad t_{n,l}\leq v\leq u\leq t_{n,l+1},\nonumber\\
F_l(v,v)&=&\mathbbm{1}_{\banach}.
\label{EQ40}
\end{eqnarray}
Then, the product $\prod_{l=0}^{n-1}PF_l(t_{n,l+1},t_{n,l})$ can be interpreted from the physical point of view as the time evolution of a system disturbed by the projective measurement $P$ at time instances $t_{n,1},\dots ,t_{n,n-1}$ such that the time evolution between the consecutive measurements is given by the operators $F_l(\cdot,t_{l})$, $0\leq l<n$. Lemma \ref{LEM4.1} shows, that if the measurements are frequent within the time interval $[t_{\mathrm{1}},t_{\mathrm{2}}]$ and $\|P\|_{\bbanach}=1$, then only the functions $PK_lP$ count in the time evolution, that is if $G_l(\cdot,\cdot)$, $0\leq l<n$ are the solutions of the initial value problems
\begin{eqnarray}
\partial_1G_l(u,v)&=&(PK_lP)(u)G_l(u,v)\qquad\qquad t_{n,l}\leq v\leq u\leq t_{n,l+1},\nonumber\\
G_l(v,v)&=&P,
\label{EQ41}
\end{eqnarray}
then the time evolution can be well approximated by $\prod_{l=0}^{n-1}G_l(t_{n,l+1},t_{n,l})$.\par 
\begin{Lem}\label{LEM4.1}
Let $\mathfrak{T}_n$ be a partition of the interval $[t_{\mathrm{1}},t_{\mathrm{2}}]$, $\interval$ given by $t_{\mathrm{1}}=t_{n,0}<\cdots <t_{n,n}=t_{2}$ with norm $|\mathfrak{T}_n|\vcentcolon=\max_{0\leq l<n}(t_{n,l+1}-t_{n,l})$. Let $P\in\bbanach$ be a projection satisfying $\|P\|_\bbanach=1$. Let $K_0,\dots,K_{n-1}$ be continuous functions mapping the intervals $[t_{n,l},t_{n,l+1}]$ to $\bbanach$. Define $F_l(\cdot,\cdot)$ and $G_l(\cdot,\cdot)$, $0\leq l<n$ as the solutions of (\ref{EQ40}) and (\ref{EQ41}), respectively. Let $\omega_l:[t_{n,l},t_{n,l+1}]\rightarrow \mathbb{R}$, $0\leq l<n$ be integrable on $[t_{n,l+1},t_{n,l}]$ such that they dominate both $\omega_{(PK_lP)(\cdot)}$ and $\omega_{K_l(\cdot)}$ on the intervals $[t_{n,l},t_{n,l+1}]$. Then, for all $x\in\banach$,
\begin{eqnarray}
&&\left\|\prod_{l=0}^{n-1}PF_l(t_{n,l+1},t_{n,l})\,x-\prod_{l=0}^{n-1}G_l(t_{n,l+1},t_{n,l})\,x\right\|_\banach\nonumber\\
&&\qquad\qquad\leq n\,\Omega\,\,|\mathfrak{T}_n|^2\,f[\mathfrak{T}_n]\|Px\|_\banach+\,\Omega\,|\mathfrak{T}_n|\,g[\mathfrak{T}_n]\|(\mathbbm{1}_{\banach}-P)x\|_\banach,
\end{eqnarray}
where 
\begin{eqnarray}
&&\Omega \vcentcolon=\int_{t_1}^{t_2}\omega(t)\,\mathrm{d}t\qquad\omega(\cdot)\vcentcolon =\sum_{l=0}^{n-1}\omega_{l}(\cdot)\chi_{[t_{n,l},t_{n,l+1}]}(\cdot)\nonumber\\
&&f[\mathfrak{T}_n]\vcentcolon=\Lambda[\mathfrak{T}_n]\Bigl(k_K^2\exp\bigl(k_K|\mathfrak{T}_n|\bigr)+k_{P}^2\exp\bigl(k_{P}|\mathfrak{T}_n|\bigr)\Bigr)=k_K^2+k^2_P+\mathcal{O}\bigl(|\mathfrak{T}_n|\bigr) \nonumber\\[4pt]
&&\Lambda[\mathfrak{T}_n]\vcentcolon=\max_{0\leq l<n}\exp\left(\|\omega_{l}\|_{L^1}(t_{n,l+1}-t_{n,l})\right)=1+\mathcal{O}\bigl(|\mathfrak{T}_n|\bigr)\nonumber\\[4pt]
&&g[\mathfrak{T}_n]\vcentcolon=\|\mathbbm{1}_\banach-P\|k_{0}\exp\left(\int_{t_{n,0}}^{t_{n,1}}\Bigl(k_0-\omega_0(t)\Bigr)\,\mathrm{d}t\right)=k_0+\mathcal{O}\bigl(|\mathfrak{T}_n|\bigr),
\end{eqnarray}
with the constants $k_K$, $k_{P}$ and $k_0$:
\begin{eqnarray}
k_K\vcentcolon=\max_{0\leq l<n}\|K_l\|_{C^0},\  k_{P}\vcentcolon=\max_{0\leq l<n}\|PK_lP\|_{C^0},\  k_0\vcentcolon=\|PK_0(\mathbbm{1}_{\banach}-P)\|_{C^0}.
\end{eqnarray}
\end{Lem}

\begin{proof} Let $F_l$ and $G_l$, $0\leq l <n$ be given as in the statement of the Lemma. The inequalities 
\begin{eqnarray}
\|PF_l(t_{n,l+1},t_{n,l})P\|\leq \|F_l(t_{n,l+1},t_{n,l})\|&\leq& \exp\left(\int_{t_{n,l}}^{t_{n,l+1}}\omega_{K_l(t)}\,\mathrm{d}t\right)\nonumber\\ 
&\leq& \exp\left(\int_{t_{n,l}}^{t_{n,l+1}}\omega(t)\,\mathrm{d}t\right)
\end{eqnarray}
and 
\begin{equation}
\|G_l(t_{n,l+1},t_{n,l})\|\leq \exp\left(\int_{t_{n,l}}^{t_{n,l+1}}\omega_{(PK_lP)(t)}\,\mathrm{d}t\right)\leq\exp\left(\int_{t_{n,l}}^{t_{n,l+1}}\omega(t)\,\mathrm{d}t\right)
\end{equation}
are simple consequences of the properties of the norm, $\|P\|=1$ and Proposition \ref{prop:growth_prop}. Using these upper bounds and (\ref{F1}), we get
\begin{eqnarray}
&&\left\|\left(\prod_{l=0}^{n-1}PF_l(t_{n,l+1},t_{n,l})\right)P-\prod_{l=0}^{n-1}G_l(t_{n,l+1},t_{n,l})\right\|\nonumber\\
&&\qquad\qquad\leq \Omega\,\Lambda[\mathfrak{T}_n]\sum_{l=0}^{n-1}\|PF_l(t_{n,l+1},t_{n,l})P-G_l(t_{n,l+1},t_{n,l})\|.
\label{EQ22}
\end{eqnarray}
Since the first two members of the Picard series of $PF_l(t_{n,l+1},t_{n,l})P$ and $G_l(t_{n,l+1},t_{n,l})$ agree for all $0\leq l<n$, we can use (\ref{F2}) and (\ref{F3}) to obtain
\begin{eqnarray}
\|PF_l(t_{n,l+1},t_{n,l})P-G_l(t_{n,l+1},t_{n,l})\|
&\leq&|\mathfrak{T}_n|^2k_K^2\mathrm{e}^{k_K|\mathfrak{T}_n|}+|\mathfrak{T}_n|^2k^2_{P}\mathrm{e}^{k_{P}|\mathfrak{T}_n|}\nonumber\\[8pt]
&=&|\mathfrak{T}_n|^{2}\frac{f[\mathfrak{T}_n]}{\Lambda[\mathfrak{T}_n]}.
\end{eqnarray}
The substitution of this result into (\ref{EQ22}) leads to 
\begin{equation}
\left\|\left(\prod_{l=0}^{n-1}PF_l(t_{n,l+1},t_{n,l})\right)P-\prod_{l=0}^{n-1}G_l(t_{n,l+1},t_{n,l})\right\|\leq n\,\Omega\,|\mathfrak{T}_n|^{2}\,f[\mathfrak{T}_n].
\end{equation}
Let us denote $\overline{P}\vcentcolon=\mathbbm{1}_{\banach}-P$. Since the initial projection of $\prod_{l=0}^{n-1}G_l(t_{n,l+1},t_{n,l})$ is $P$ and $\|P\|=1$, we can write
\begin{eqnarray*}
&&\left\|\left(\prod_{l=0}^{n-1}PF_l(t_{n,l+1},t_{n,l})-\prod_{l=0}^{n-1}G_l(t_{n,l+1},t_{n,l})\right)\overline{P}\right\|\nonumber\\
&&\qquad\qquad\qquad=\left\|\left(\prod_{l=0}^{n-1}PF_l(t_{n,l+1},t_{n,l})\right)\overline{P}\right\|\nonumber\\[8pt]
&&\qquad\qquad\qquad=\left\|\left(\prod_{l=1}^{n-1}PF_l(t_{n,l+1},t_{n,l})P\right)PF_l(t_{n,1},t_{n,0})\overline{P}\right\|\nonumber\\
&&\qquad\qquad\qquad\leq \exp\left(\int_{t_1}^{t_{2}}\omega(t)\, \mathrm{d}t\right)\exp\left(-\int_{t_{n,0}}^{t_{n,1}}\omega(t)\,\mathrm{d}t\right)\|\overline{P}\|\sum_{p=1}^{\infty}\frac{\|\bigl\{K_0\bigr\}_{p}(t_{n,1},t_{n,0})\|}{p!}
\end{eqnarray*}
\begin{eqnarray}
&\leq& \Omega\,\Lambda[\mathfrak{T}_n]\|\overline{P}\|\sum_{p=1}^{\infty}\frac{(t_{n,1}-t_{n,0})^pk^p_{0}}{p!}\nonumber\qquad\qquad\qquad\qquad\\
&\leq& \Omega\,\Lambda[\mathfrak{T}_n]\|\overline{P}\|(t_{n,1}-t_{n,0})k_{0}\mathrm{e}^{k_0(t_{n,1}-t_{n,0})}\nonumber\\[6pt]
&\leq&\Omega\,|\mathfrak{T}_n|g[\mathfrak{T}_n].
\end{eqnarray}
Since for any $x\in\banach$, the equation $x=Px+\overline{P}x$ holds, the statement of the lemma follows.
\end{proof}

Let $\banach$ be a Banach space with a decomposition $\banach=\banach_{\mathrm{I}}\oplus\banach_{\mathrm{C}}$ such that both $\banach_{\mathrm{I}}$ and $\banach_{\mathrm{C}}$ are closed subspaces of $\banach$ and the continuous projection $P_{\mathrm{I}}$ corresponding to $\banach_{\mathrm{I}}$ and its complementary projection $P_{\mathrm{C}}$ are norm one. Assume that $T\in\bbanach$ is of the form $T=T_{\mathrm{I}}\oplus T_{\mathrm{C}}$, where $T_\mathrm{I}\in\mathcal{B}(\banach_{\mathrm{I}})$ is an isometric isomorphism, while $T_\mathrm{C}\in\mathcal{B}(\banach_{\mathrm{C}})$ is a strict contraction, that is $\|T_{\mathrm{C}}\|_\bbanach=\|T_{\mathrm{C}}\|_{\mathcal{B}(\banach_{\mathrm{C}})}<1$. Lemma \ref{LEM4.2} shows, that the results of Lemma \ref{LEM4.1} apply also to products where the projection is replaced by $T$.\par 

\begin{Lem}\label{LEM4.2} Let $\banach$ be a Banach space with a decomposition $\banach=\banach_{\mathrm{I}}\oplus\banach_{\mathrm{C}}$ and $T=T_{\mathrm{I}}\oplus T_{\mathrm{C}}\in\bbanach$ obeying the conditions described above. Let $\mathfrak{T}_n$ be a partition of the interval $[t_{1},t_{2}]$, $\interval$ given by $t_{1}=t_{n,0}<\cdots <t_{n,n}=t_{2}$ with norm $|\mathfrak{T}_n|\vcentcolon=\max_{0\leq l<n}(t_{n,l+1}-t_{n,l})$. Let $K_0,\dots,K_{n-1}$ be continuous functions mapping the intervals $[t_{n,l},t_{n,l+1}]$, $0\leq l<n$ to $\bbanach$. Let $\omega_l:[t_{n,l},t_{n,l+1}]\rightarrow \mathbb{R}$, $0\leq l<n$ be integrable on $[t_{n,l},t_{n,l+1}]$ such that they dominate both $\omega_{(PK_lP)(\cdot)}$ and $\omega_{K_l(\cdot)}$ on the intervals $[t_{n,l},t_{n,l+1}]$. Define $F_l(\cdot,\cdot)$ and $G_l(\cdot,\cdot)$, $0\leq l<n$ as the solutions of (\ref{EQ40}) and (\ref{EQ41}), respectively. Then,
\begin{eqnarray}
&&\left\|\prod_{l=0}^{n-1}TF_l(t_{n,l+1},t_{n,l})-T_\mathrm{I}^n\,\prod_{l=0}^{n-1}T^{-l}_\mathrm{I}G_l(t_{n,l+1},t_{n,l})T^{l}_\mathrm{I}\right\|\nonumber\\
&&\qquad\qquad\qquad\leq n\Omega\,|\mathfrak{T}_n|^2h[\mathfrak{T}_n]+\Omega\,|\mathfrak{T}_n|h_1[\mathfrak{T}_n]+\Omega\,\|T_\mathrm{C}\|^n
\end{eqnarray}
for all $t\in\mathbb{R}^+_0$, where $\Omega$ and the error functions $h$ and $h_1$ are defined through
\begin{eqnarray*}
&&\Omega \vcentcolon=\int_{t_\mathrm{1}}^{t_\mathrm{2}}\omega(t)\,\mathrm{d}t\qquad\qquad \omega(\cdot)\vcentcolon =\sum_{l=0}^{n-1}\omega_{l}(\cdot)\chi_{[t_{n,l},t_{n,l+1}]}(\cdot)\nonumber\\
&&h[\mathfrak{T}_n]\vcentcolon=\Lambda[\mathfrak{T}_n]\bigl(k_K^2\mathrm{e}^{k_K|\mathfrak{T}_n|}(1+\tau_\mathrm{C}\mathrm{e}^{k_K|\mathfrak{T}_n|})+k_P^2\mathrm{e}^{k_P|\mathfrak{T}_n|}\bigr)\nonumber\\
&&\qquad\quad =(1+\tau_\mathrm{C})k_K^2+k^2_P+\mathcal{O}(|\mathfrak{T}_n|)\nonumber\\
&&h_1[\mathfrak{T}_n]\vcentcolon=\tau_\mathrm{C}\Lambda[\mathfrak{T}_n]\mathrm{e}^{k_P|\mathfrak{T}_n|}=\tau_\mathrm{C}(1+\mathcal{O}\bigl(|\mathfrak{T}_n|\bigr)),
\end{eqnarray*}
\begin{eqnarray}
&&\Lambda[\mathfrak{T}_n]\vcentcolon=\max_{0\leq l<n}\exp\left(\|\omega_{l}\|_{L^1}(t_{n,l+1}-t_{n,l})\right)=1+\mathcal{O}\bigl(|\mathfrak{T}_n|\bigr)
\end{eqnarray}
with constants
\begin{eqnarray}
k_K\vcentcolon=\max_{0\leq l<n}\|K_l\|_{C^0},\quad k_{P}\vcentcolon=\max_{0\leq l<n}\|PK_lP\|_{C^0},\quad\tau_{\mathrm{C}}\vcentcolon=\frac{1+\|T_{\mathrm{C}}\|_\bbanach}{1-\|T_{\mathrm{C}}\|_\bbanach}.
\end{eqnarray}
\end{Lem}

\begin{proof}Define $\hat{T}_\mathrm{I}\vcentcolon=T_\mathrm{I}\oplus P_\mathrm{C}\in \bbanach$, $\hat{T}_\mathrm{C}\vcentcolon=P_\mathrm{I}\oplus T_\mathrm{C}\in\bbanach$ and $\hat{T}^{-1}_\mathrm{I}\vcentcolon=T^{-1}_\mathrm{I}\oplus P_\mathrm{I}\in \bbanach$. Note that the equations
\begin{equation}
T=\hat{T}_\mathrm{I}\hat{T}_\mathrm{C}=\hat{T}_\mathrm{C}\hat{T}_\mathrm{I},\qquad \hat{T}^{-1}_\mathrm{I}\hat{T}^{\ }_\mathrm{I}=\hat{T}^{\ }_\mathrm{I}\hat{T}^{-1}_\mathrm{I}=\mathbbm{1}_{\banach}
\end{equation} 
hold and $\|\hat{T}_\mathrm{C}\|=\|\hat{T}_\mathrm{I}\|=1$ by the assumptions posed on the norms of $P_{\mathrm{C}}$ and $P_{\mathrm{I}}$. We apply the notation $F_l\vcentcolon=\hat{T}_\mathrm{I}^{-l}F_l(t_{n,l+1},t_{n,l})\hat{T}_\mathrm{I}^{l}$, so that the equation
\begin{equation}
\prod_{l=0}^{n-1}TF_l(t_{n,l+1},t_{n,l})=\hat{T}_\mathrm{I}^n\prod_{l=0}^{n-1}\hat{T}_\mathrm{C}F_l
\end{equation}
holds. Since $\hat{T}_{\mathrm{I}}$ is an isometric isomorphism in $\bbanach$, we can write
\begin{equation}
\left\|\prod_{l=0}^{n-1}TF(t_{n,l+1},t_{n,l})-\hat{T}_\mathrm{I}^n\prod_{l=0}^{n-1}P_\mathrm{I}F_lP_\mathrm{I}\right\|=\left\|\prod_{l=0}^{n-1}\hat{T}_\mathrm{C}F_l-\prod_{l=0}^{n-1}P_\mathrm{I}F_lP_\mathrm{I}\right\|.
\label{EQ9}
\end{equation}
We introduce the following members of $\bbanach$
\begin{eqnarray}
A_k&\vcentcolon=&\left(\prod_{l=0}^{k}P_\mathrm{I}F_l\right)P_\mathrm{I},\qquad B_k\vcentcolon=\sum_{l=0}^k\left(\prod_{p=l}^k P_\mathrm{I}F_p\right)\left(\prod_{q=0}^{l-1} T_\mathrm{C}F_q\right)P_\mathrm{C},\nonumber\\
C_k&\vcentcolon=&\sum_{l=0}^k\left(\prod_{p=l}^k T_\mathrm{C}F_p\right)\left(\prod_{q=0}^{l-1} P_\mathrm{I}F_q\right)P_\mathrm{I},\qquad D_k\vcentcolon=\left(\prod_{l=0}^k T_\mathrm{C}F_l\right)P_\mathrm{C},
\end{eqnarray}
where $0\leq k<n$. They satisfy the relations
\begin{eqnarray}
A_{k+1}&=&P_\mathrm{I}F_{k+1}\,A_k \qquad\qquad\quad\quad\ \  B_{k+1}=P_\mathrm{I}F_{k+1}\,\bigl(B_k+D_k\bigr)\nonumber\\
C_{k+1}&=&P_\mathrm{C}F_{k+1}\,\bigl(A_k+C_k\bigr)\qquad\quad D_{k+1}=T_\mathrm{C}F_{k+1}\,D_k.\label{EQ10}
\end{eqnarray}
for all $0\leq k< n$. Moreover, let us introduce their sums as 
\begin{equation}
\Sigma_k\vcentcolon=A_k+B_k+C_k+D_k,\qquad 0\leq k<n,
\end{equation}
and the differences
\begin{equation}
\Delta_k\vcentcolon=\Sigma_k-\prod_{l=0}^{k}\hat{T}_\mathrm{C}F_k,\qquad 0\leq k<n.
\end{equation}
We prove the following assertion: if $m$ is a non-negative integer, then 
\begin{eqnarray}
\left\|\Delta_m\right\|&\leq&\Lambda[\mathfrak{T}_n]|\mathfrak{T}_n|^2 k_K^2\mathrm{e}^{2k_K|\mathfrak{T}_n|}\exp\left(\int_{t_{n,0}}^{t_{n,m+1}}\omega(s)\,\mathrm{d}s\right)\sum_{l=0}^{m}\delta_{T,l},\nonumber\\
\delta_{T,l}&=&\frac{(1+\|T_\mathrm{C}\|)\,(1-\|T_\mathrm{C}\|^{l})}{1-\|T_\mathrm{C}\|},\qquad 0\leq l<n .
\label{EQ11}
\end{eqnarray}
We prove by induction. For $m=0$, the equation $\hat{T}_\mathrm{C}F_0=\Sigma_0$ holds, thus $\|\hat{T}_\mathrm{C}F_0-\Sigma_0\|=0$, which is the result of the evaluation of the inequality (\ref{EQ11}) at $m=0$. Assume that the statement holds for some integral $m>0$. Then, provided that $\|\hat{T}_\mathrm{C}\|= 1$,
\begin{eqnarray}
\left\|\Delta_{m+1}\right\|&=&\left\|\hat{T}_\mathrm{C}F_{m+1}\Delta_m+\hat{T}_\mathrm{C}F_{m+1}\Sigma_m-\Sigma_{m+1}\right\|\nonumber\\
&\leq&\left\|\hat{T}_\mathrm{C}F_{m+1}\Delta_m\right\|+\left\|\hat{T}_\mathrm{C}F_{m+1}\Sigma_m-\Sigma_{m+1}\right\|\nonumber\\
&\leq& \exp\left(\int_{t_{n,m+1}}^{t_{n,m+2}}\omega(s)\,\mathrm{d}s\right)\|\Delta_m\|+\left\|\hat{T}_\mathrm{C}F_{m+1}\Sigma_m-\Sigma_{m+1}\right\|
\end{eqnarray}
where the inequality 
\begin{eqnarray}
\|\hat{T}_\mathrm{C}F_{m+1}\|&\leq& \|\hat{T}_\mathrm{C}\|\,\left\|\hat{T}_{\mathrm{I}}^{-(m+1)}F_{m+1}(t_{n,m+2},t_{n,m+1})\hat{T}_{\mathrm{I}}^{m+1}\right\|\nonumber\\
&\leq& \exp\left(\int_{t_{n,m+1}}^{t_{n,m+2}}\omega(s)\,\mathrm{d}s\right),
\end{eqnarray}
provided by Proposition \ref{prop:growth_prop} has been used. We have to show that 
\begin{eqnarray}
&&\left\|\hat{T}_\mathrm{C}F_{m+1}\Sigma_m-\Sigma_{m+1}\right\|\leq\Lambda[\mathfrak{T}_n]|\mathfrak{T}_n|^2 k_K^2\mathrm{e}^{2k_K|\mathfrak{T}_n|}\exp\left(\int_{t_{n,0}}^{t_{n,m+1}}\omega(s)\,\mathrm{d}s\right)\nonumber\\
&&\qquad\qquad\qquad\qquad\qquad\qquad\qquad\times \frac{(1-\|T_\mathrm{C}\|)(1-\|T_\mathrm{C}\|^{m+1})}{1-\|T_\mathrm{C}\|}.
\label{EQ12}
\end{eqnarray}
Since both the initial and the final projections of $T_\mathrm{C}\in\bbanach$ are equal to $P_\mathrm{C}$, a short calculation with the help of equations of (\ref{EQ10}) gives
\begin{eqnarray}
P_\mathrm{I}\left(\hat{T}_\mathrm{C}F_{m+1}\Sigma_m-\Sigma_{m+1}\right)P_\mathrm{I}
&=&P_\mathrm{I}F_{m+1}\, C_m,\nonumber\\
P_\mathrm{I}\left(\hat{T}_\mathrm{C}F_{m+1}\Sigma_m-\Sigma_{m+1}\right)P_\mathrm{C}&=&0,\nonumber\\
P_\mathrm{C}\left(\hat{T}_\mathrm{C}F_{m+1}\Sigma_m-\Sigma_{m+1}\right)P_\mathrm{I},
&=&0,\nonumber\\
P_\mathrm{C}\left(\hat{T}_\mathrm{C}F_{m+1}\Sigma_m-\Sigma_{m+1}\right)P_\mathrm{C}
&=&P_\mathrm{C}F_{m+1}\,B_{m},
\end{eqnarray}
which results in
\begin{equation}
\left\|\hat{T}_\mathrm{C}F_{m+1}\Sigma_m-\Sigma_{m+1}\right\|\leq\left\|P_\mathrm{I}F_{m+1}\, B_m\right\|+\left\|P_\mathrm{C}F_{m+1}\,C_{m}\right\|.
\label{Offdiag1}
\end{equation}
We can obtain an upper bound of both terms of the sum as follows. Provided that $\|P_\mathrm{C}\|=\|P_\mathrm{I}\|=1$ and $\hat{T}_\mathrm{I}$ is an isometric isomorphism, we can use (\ref{F2}) and (\ref{F3}) to obtain
\begin{eqnarray}
P_{\mathrm{x}}F_{m}P_{\mathrm{y}}&\leq&k_K|\mathfrak{T}_n|\mathrm{e}^{k_K|\mathfrak{T}_n|}
\end{eqnarray}
holds for all $0\leq m<n$, whenever $(\mathrm{x},\mathrm{y})$ is equal to either $(\mathrm{C},\mathrm{I})$ or $(\mathrm{I},\mathrm{C})$. Thus, we can write 
\begin{equation}
\left\|P_\mathrm{I}F_{m+1}\, B_m\right\|\leq \left\|B_m\right\|  |\mathfrak{T}_n| k_K  \mathrm{e}^{k_K|\mathfrak{T}_n|},\quad\left\|P_\mathrm{I}F_{m+1}\, C_m\right\|\leq \left\|C_m\right\|  |\mathfrak{T}_n| k_K  \mathrm{e}^{k_K|\mathfrak{T}_n|}.
\label{EQ14}
\end{equation}
Furthermore,
\begin{eqnarray}
&&\left\|C_m\right\|=\left\|\sum_{k=0}^{m}\left(\prod_{p=k+1}^{m}T_\mathrm{C}F_p\right)\cdot T_\mathrm{C}F_kP_\mathrm{I}\cdot\left(\prod_{q=0}^{k-1}P_\mathrm{I}F_q\right)\cdot P_\mathrm{I}\right\|\nonumber\\
&&\quad\leq\sum_{k=0}^{m}\|T_\mathrm{C}\|^{m-k+1}\,\|P_\mathrm{C}F_{k}P_\mathrm{I}\|\cdot\prod_{\substack{p=0\\p\neq k}}^{m}\left\|F_p\right\|\nonumber\\
&&\quad\leq|\mathfrak{T}_n| k_K   \mathrm{e}^{k_K|\mathfrak{T}_n| }\exp\left(\int_{t_{n,0}}^{t_{n,m+1}}\omega(s)\,\mathrm{d}s\right)\sum_{k=0}^{m}\|T_\mathrm{C}\|^{m-k+1}\,\exp\left(-\int_{t_{n,k}}^{t_{n,k+1}}\omega(s)\,\mathrm{d}s\right)\nonumber\\
&&\quad\leq \Lambda[\mathfrak{T}_n]|\mathfrak{T}_n| k_K\mathrm{e}^{k_K|\mathfrak{T}_n|}\exp\left(\int_{t_{n,0}}^{t_{n,m+1}}\omega(s)\,\mathrm{d}s\right) \frac{\|T_\mathrm{C}\|\,(1-\|T_\mathrm{C}\|^{m+1})}{1-\|T_\mathrm{C}\|}.
\label{EQ15}
\end{eqnarray}
Substitution of this result into the inequality (\ref{EQ14}) gives 
\begin{eqnarray}
&&\left\|P_\mathrm{I}F_{m+1} C_m\right\|\nonumber\\
&&\qquad\leq \Lambda[\mathfrak{T}_n]|\mathfrak{T}_n|^2 k_K^2 \mathrm{e}^{2k_K|\mathfrak{T}_n|}\exp\left(\int_{t_{n,0}}^{t_{n,m+1}}\omega(s)\,\mathrm{d}s\right)\frac{\,\|T_\mathrm{C}\|\,(1-\|T_\mathrm{C}\|^{m+1})}{1-\|T_\mathrm{C}\|}.
\end{eqnarray}
Regarding $\|B_m\|$, a similar calculation gives
\begin{eqnarray}
\|B_m\|&\leq&\Lambda[\mathfrak{T}_n]|\mathfrak{T}_n| k_K\mathrm{e}^{k_K|\mathfrak{T}_n|}\exp\left(\int_{t_{n,0}}^{t_{n,m+1}}\omega(s)\,\mathrm{d}s\right) \frac{1-\|T_\mathrm{C}\|^{m+1}}{1-\|T_\mathrm{C}\|},
\label{EQ16}
\end{eqnarray}
and
\begin{equation}
\left\|P_\mathrm{I}F_{m+1}\, B_m\right\|\leq\Lambda[\mathfrak{T}_n]|\mathfrak{T}_n|^2 k_K^2\mathrm{e}^{2k_K|\mathfrak{T}_n|}\exp\left(\int_{t_{n,0}}^{t_{n,m+1}}\omega(s)\,\mathrm{d}s\right) \frac{1-\|T_\mathrm{C}\|^{m+1}}{1-\|T_\mathrm{C}\|},
\end{equation}
which imply (\ref{EQ12}).

Turning our attention back to equation (\ref{EQ9}), we observe that
\begin{eqnarray}
\left\|\prod_{k=0}^{n-1}\hat{T}_\mathrm{C}F_k-\prod_{k=0}^{n-1}P_\mathrm{I}F_kP_\mathrm{I}\right\|&\leq &\left\|\prod_{k=0}^{n-1}\hat{T}_\mathrm{C}F_k-\Sigma_{n-1}\right\|+\Biggl\|\Sigma_{n-1}-A_{n-1}\Biggr\|\nonumber\\[8pt]
&\leq&\|\Delta_{n-1}\|+\|B_{n-1}\|+\|C_{n-1}\|+\|D_{n-1}\|.
\label{EQ17}
\end{eqnarray}
We use $\|T_\mathrm{C}\|<1$ to obtain $\sum_{l=0}^{n-1}\delta_{T,l}\leq n\tau_{\mathrm{C}}$ and so
\begin{eqnarray}
\|\Delta_{n-1}\|&\leq&n\Omega\,\Lambda[\mathfrak{T}_n]|\mathfrak{T}_n|^2 k_K^2\tau_{\mathrm{C}}\mathrm{e}^{2k_K|\mathfrak{T}_n|}.
\end{eqnarray}
Then, using the upper bounds (\ref{EQ15}), (\ref{EQ16}) and $\|T_\mathrm{C}\|<1$ again allows us to write
\begin{equation}
\|B_{n-1}\|+\|C_{n-1}\|\leq\Omega\,\Lambda[\mathfrak{T}_n]|\mathfrak{T}_n|\tau_{\mathrm{C}} k_K\mathrm{e}^{k_K|\mathfrak{T}_n|},
\end{equation}
and 
\begin{equation}
\left\|D_{n-1}\right\|= \left\|\left(\prod_{l=0}^{n-1}T_\mathrm{C}F_l\right)P_\mathrm{C}\right\|\leq \|T_\mathrm{C}\|^{n} \prod_{l=0}^{n-1}\left\|F_l\right\|\leq \Omega\,\left\|T_\mathrm{C}\right\|^{n}.
\end{equation}
Thus, the l.h.s.~of (\ref{EQ17}) has the upper bound 
\begin{equation}
\left\|\prod_{l=0}^{n-1}\hat{T}_\mathrm{C}F_l-\prod_{l=0}^{n-1}P_\mathrm{I}F_lP_\mathrm{I}\right\|\leq n\Omega\,|\mathfrak{T}_n|^2 \tau_{\mathrm{C}}h_2[\mathfrak{T}_n]+\Omega\,|\mathfrak{T}_n|\tau_\mathrm{C}h_1[\mathfrak{T}_n]+\Omega\,\|T_{\mathrm{C}}\|^{n},
\label{EQ18}
\end{equation}
where $h_1$ and $h_2$ are defined through
\begin{equation}
h_{1}[\mathfrak{T}_n]=\Lambda[\mathfrak{T}_n]k_K\mathrm{e}^{k|\mathfrak{T}_n|},\qquad h_{2}[\mathfrak{T}_n]=\Lambda[\mathfrak{T}_n]k_K^2\mathrm{e}^{2k_K|\mathfrak{T}_n|}.
\end{equation}
We apply Lemma \ref{LEM4.1} to the projection $P_\mathrm{I}$ and the sequences $\hat{T}_{\mathrm{I}}^{-l}F_l(\cdot,\cdot)\hat{T}_{\mathrm{I}}^{l}$ and $\hat{T}_{\mathrm{I}}^{-l}G_l(\cdot,\cdot)\hat{T}_{\mathrm{I}}^{l}$, $0\leq l<n$: Provided that the initial projection of $\prod_{k=0}^{n-1}P_\mathrm{I}F_kP_\mathrm{I}$ is $P_\mathrm{I}$, we can write
\begin{equation}
\left\|\prod_{l=0}^{n-1}P_\mathrm{I}F_lP_\mathrm{I}-\prod_{l=0}^{n-1}T_{\mathrm{I}}^{-l}G_l(t_{l+1},t_l)T_{\mathrm{I}}^{l}\right\|\leq n\Omega\,|\mathfrak{T}_n|^2\bigr(k_Kh_1[\mathfrak{T}_n]+h_{P}[\mathfrak{T}_n]\bigr),
\label{EQ19}
\end{equation}
where 
\begin{equation}
h_P[\mathfrak{T}_n]=\Lambda[\mathfrak{T}_n]k_P^2\mathrm{e}^{k_P|\mathfrak{T}_n|}.
\end{equation}
Combining (\ref{EQ18}) and (\ref{EQ19}) with the triangle inequality leads to
\begin{eqnarray}
&&\left\|\prod_{l=0}^{n-1}TF_l(t_{n,l+1},t_{n,l})-T^n_{\mathrm{I}}\prod_{l=0}^{n-1}T_{\mathrm{I}}^{-l}G_l(t_{n,l+1},t_{n,l})T_{\mathrm{I}}^{l}\right\|\nonumber\\
&&\qquad=\left\|\prod_{l=0}^{n-1}TF_l(t_{n,l+1},t_{n,l})-\hat{T}^n_{\mathrm{I}}\prod_{l=0}^{n-1}T_{\mathrm{I}}^{-l}G_l(t_{n,l+1},t_{n,l})T_{\mathrm{I}}^{l}\right\|\nonumber\\
&&\qquad=\left\|\prod_{l=0}^{n-1}\hat{T}_\mathrm{C}F_l-\prod_{l=0}^{n-1}T_{\mathrm{I}}^{-l}G_l(t_{l+1},t_l)T_{\mathrm{I}}^{-l}\right\|\nonumber\\
&&\qquad\leq\left\|\prod_{l=0}^{n-1}\hat{T}_\mathrm{C}F_l-\prod_{l=0}^{n-1}P_\mathrm{I}F_lP_\mathrm{I}\right\|+\left\|\prod_{l=0}^{n-1}P_\mathrm{I}F_lP_\mathrm{I}-\prod_{l=0}^{n-1}T_{\mathrm{I}}^{-l}G_l(t_{n,l+1},t_{n,l})T_{\mathrm{I}}^{-l}\right\|\nonumber\\[6pt]
&&\qquad\leq n\Omega\,|\mathfrak{T}_n|^2\bigl(k_Kh_1[\mathfrak{T}_n]+\tau_{\mathrm{C}}h_2[\mathfrak{T}_n]+h_P[\mathfrak{T}_n]\bigr)+\tau_{\mathrm{C}}|\mathfrak{T}_n|h_1[\mathfrak{T}_n]+\Omega\,\|T_\mathrm{C}\|^n
\end{eqnarray}
which is the statement of the Lemma.
\end{proof}

\begin{Cor} \label{COR4.3} Let $\banach$ be a Banach space with a decomposition $\banach=\banach_{\mathrm{I}}\oplus\banach_{\mathrm{C}}$ such that both $\banach_{\mathrm{I}}$ and $\banach_{\mathrm{C}}$ are closed subspaces of $\banach$ and the continuous projection $P_{\mathrm{I}}$ corresponding to $\banach_{\mathrm{I}}$ and its complementary projection $P_{\mathrm{C}}$ are norm one. Assume that $T\in\bbanach$ is of the form $T=P_{\mathrm{I}}\oplus T_{\mathrm{C}}$, where $T_\mathrm{C}\in\mathcal{B}(\banach_{\mathrm{C}})$ is a strict contraction. Let $K\in C([t_1,t_2],\bbanach)$ and define the propagator $F(\cdot,\cdot)$ as the solution of the initial value problem
\begin{eqnarray}
\partial_1F(u,v)&=&K(u)F(u,v)\qquad\qquad t_{1}\leq v\leq u\leq t_{2}\nonumber\\
F(v,v)&=&\mathbbm{1}_{\banach}.
\end{eqnarray}
Let $G(\cdot,\cdot)$ be the solution of the initial value problem 
\begin{eqnarray}
\partial_1G(u,v)&=&\bigl(P_\mathrm{I}KP_\mathrm{I}\bigr)(u)G(u,v)\qquad\qquad t_{1}\leq v\leq u\leq t_{2},\nonumber\\
G(v,v)&=&P_{\mathrm{I}}.
\end{eqnarray}
Let $\alpha=(\alpha_1,\alpha_2,\dots)$ be a sequence of sequences $(\alpha_{n,l})_{0\leq l<n}$ of positive numbers satisfying
\begin{equation}
\sum_{l=0}^{n-1}\alpha_{n,l}=1\qquad \lim_{n\rightarrow\infty}\max_{0\leq l<n}\alpha_{n,l}=0 \qquad \lim_{n\rightarrow\infty}\sqrt{n}\max_{0\leq l<n}\alpha_{n,l}=0.
\label{EQ86}
\end{equation}
Then,
\begin{equation}
\lim_{n\rightarrow\infty}\prod_{l=0}^{n-1}T\bigl(F\circ(\rho^{\alpha}_{n,l+1}\times \rho^{\alpha}_{n,l})\bigr)=G(\cdot,t_1)
\end{equation}
in the $C^0$ norm.
\end{Cor}
\begin{proof}
Let $\omega:[t_1,t_2]\rightarrow \mathbb{R}$ be integrable such that $\omega$ dominates both $\omega_{K}$ and $\omega_{PKP}$ on the interval $[t_1,t_2]$. Let $s\in [t_1,t_2]$. Set $\mathfrak{T}_{n}(s)$ to be the partition of the interval $[t_1,s]$ given by $t_1=\rho^{\alpha}_{n,0}(s)<\cdots <\rho^{\alpha}_{n,n}(s)=s$. Define $K_0,\dots K_{n-1}$ as the restrictions of $K$ to the intervals $[\rho^{\alpha}_{n,l}(s),\rho^{\alpha}_{n,l+1}(s)]$, $0\leq l<n$. Then, the application of Lemma \ref{LEM4.2} gives
\begin{eqnarray}
&&\left\|\prod_{l=0}^{n-1}TF(\rho^{\alpha}_{n,l+1}(s),\rho^{\alpha}_{n,l}(s))-P_\mathrm{I}\,\prod_{l=0}^{n-1}P_\mathrm{I}G(\rho^{\alpha}_{n,l+1}(s),\rho^{\alpha}_{n,l}(s))P_\mathrm{I}\right\|\nonumber\\
&&\qquad=\left\|\prod_{l=0}^{n-1}TF(\rho^{\alpha}_{n,l+1}(s),\rho^{\alpha}_{n,l}(s))-\prod_{l=0}^{n-1}G(\rho^{\alpha}_{n,l+1}(s),\rho^{\alpha}_{n,l}(s))\right\|\nonumber\\
&&\qquad=\left\|\prod_{l=0}^{n-1}TF(\rho^{\alpha}_{n,l+1}(s),\rho^{\alpha}_{n,l}(s))-G(s,t_1)\right\|\nonumber\\
&&\qquad\leq n\Omega\,|\mathfrak{T}_n(s)|^2g(s-t_1)+\Omega\,|\mathfrak{T}_n(s)|g_1(s-t_1)+\Omega\,\|T_\mathrm{C}\|^n,
\label{EQ85}
\end{eqnarray}
where $\Omega$ and the error functions $g$ and $g_1$ are defined through
\begin{eqnarray}
\Omega &=&\int_{t_\mathrm{1}}^{s}\omega(t)\,\mathrm{d}t\nonumber\\
g(s-t_1)&=&\Lambda_n(s-t_1)\Bigl(\|K\|^2_{C^0}\mathrm{e}^{\|K\|_{C^0}|\mathfrak{T}_n(s)|}+\tau_\mathrm{C}\|K\|^2_{C^0}\mathrm{e}^{2\|K\|_{C^0}|\mathfrak{T}_n(s)|}\nonumber\\
&&\qquad\qquad+\|P_\mathrm{I}KP_{\mathrm{I}}\|^2_{C^0}\mathrm{e}^{\|P_\mathrm{I}KP_{\mathrm{I}}\|_{C^0}|\mathfrak{T}_n(s)|}\Bigr)\nonumber\\
g_1(s-t_1)&=&\tau_\mathrm{C}\Lambda_n(s-t_1)\mathrm{e}^{\|P_\mathrm{I}KP_{\mathrm{I}}\|_{C^0}|\mathfrak{T}_n(s)|}\nonumber\\
\Lambda_n(s-t_1)&=&\exp\bigl(\|\omega\|_{L^1}(s-t_1)\max_{0\leq l<n}\alpha_{n,l}\bigr)
\end{eqnarray}
with 
\begin{eqnarray}
\tau_{\mathrm{C}}\vcentcolon=\frac{1+\|T_{\mathrm{C}}\|_\bbanach}{1-\|T_{\mathrm{C}}\|_\bbanach}.
\end{eqnarray}
Since $|\mathfrak{T}_n(s)|=(s-t_1)\max_{0\leq l<n}\alpha_{n,l}$, the upper bound appearing in (\ref{EQ85}) is uniformly bounded on $[t_1,t_2]$ for every $n\in\mathbb{N}$. Furthermore, since $\|T_\mathrm{C}\|^n\rightarrow 0$ and (\ref{EQ86}) holds, this upper bound turns to zero as $n$ tends to infinity. 
\end{proof}
\begin{Rk} Corollary \ref{COR4.3} is a generalization and, concerning the speed of con\-ver\-gen\-ce, is an improved form of (\ref{Lidar}) studied in  \cite{dominy2013analysis} and \cite{paz2012zeno}.
\end{Rk}

\section{Generalized quantum Zeno dynamics}
\label{SEC5}
Consider an arbitrary isometric isomorphism $T:\banach\rightarrow\banach$ and the admissible $\alpha$ defined through $\alpha_{n,l}=l/n$. For any $K_0\in\bbanach$, the constant valued function $[t_1,t_2]\ni s\rightarrow K(s)=K_0-T^{-1}K_0T$ is in the kernel of $\mathfrak{P}_{\alpha,T}$:
\begin{equation}
\|\mathfrak{P}^{(n,k)}_{\alpha,T}[K]\|_{C^0}\leq \frac{2\|K_0\|_\bbanach}{n},
\label{EQ78}
\end{equation}
therefore $\mathfrak{P}^{(n)}_{\alpha,T}[K]\rightarrow 0$. It was first pointed out in \cite{bernad2017dynamical}, that the  bound (\ref{EQ78}) enables one to prove the  convergence of  (\ref{Q1a}) with the choice $C_n=T^{-n}$ if $\mathfrak{P}_{\alpha,T}[K]$ exists. In \cite{bernad2017dynamical}, the following strategy was used: If $K\in C([t_1,t_2],\bbanach)$ is such that, $\partial_tK=0$ and $\mathfrak{P}_{\alpha,T}[K]$ exists, then one decomposes it as $K=\mathfrak{P}_{\alpha,T}[K]+\overline{K}$. Since $\mathfrak{P}_{\alpha,T}[\overline{K}]=0$, then for every $\varepsilon>0$, there exists $\overline{K}_\varepsilon=L_\varepsilon-T^{-1}L_\varepsilon T$, $L_\varepsilon$ being constant valued functions, such that $\|\overline{K}_\varepsilon-\overline{K}\|_{C^0}\leq \varepsilon$ \cite{yosida1938mean}. One first proves
\begin{equation}
\lim_{n\rightarrow\infty}T^{-n}\Bigl(T\exp\bigl(\mathfrak{P}_{\alpha,T}[K]t/n+\overline{K}_\varepsilon t/n\bigr)\Bigr)^n=\exp\bigl(\mathfrak{P}_{\alpha,T}[K]t\bigr)
\end{equation}
using (\ref{EQ78}). Then, after proving that 
\begin{equation}
\left\|\Bigl(T\exp\bigl(\mathfrak{P}_{\alpha,T}[K]t/n+\overline{K}_\varepsilon t/n\bigr)\Bigr)^n-\Bigl(T\exp\bigl(\mathfrak{P}_{\alpha,T}[K]t/n+\overline{K} t/n\bigr)\Bigr)^n\right\|_{\bbanach}=\mathcal{O}(\varepsilon)
\end{equation}
holds, one concludes that
\begin{equation}
\lim_{n\rightarrow\infty}T^{-n}\Bigl(T\exp\bigl(Kt/n\bigr)\Bigr)^n=\exp\bigl(\mathfrak{P}_{\alpha,T}[K]t\bigr).
\end{equation}

Here, the same strategy will be used. After proving a result similar to (\ref{EQ78}), we follow that strategy in the proof of the first main Theorem of the paper.

\begin{Lem} \label{LEM5.1} Given and admissible $\alpha$, an isometric isomorphism $T\in\bbanach$, $K\in C^{0,1}([t_1,t_2],\bbanach)$ such that $K=L-T^{-1}LT$, $L\in C^{0,1}([t_1,t_2],\bbanach)$ we have 
\begin{equation}
\bigl\|\mathfrak{P}^{(n,k)}_{\alpha,T}[K](s)\bigr\|_\bbanach\leq\psi(s-t_1)S_\alpha(n)\|L\|_{C^{0,1}},
\label{EQ84}
\end{equation}
where $\psi(x)=2\max\{2,3x\}$ and 
\begin{equation}
S_\alpha(n)=\max_{0\leq l<n}\alpha_{n,l}+\sum_{l=0}^{n-2}|\alpha_{n,l+1}-\alpha_{n,l}|.
\end{equation}
\end{Lem}
\begin{proof}
Let $\alpha$ be admissible, $K=L-T^{-1}LT$, $L\in C^{0,1}([t_1,t_2],\bbanach)$. Then, 
\begin{eqnarray}
\mathfrak{P}^{(n,k)}_{\alpha,T}[K]&=&\sum_{l=0}^{k}T^{-l}\Bigl(\gamma^{\alpha}_{n,l+1}\bigl(K\circ\rho^{\alpha}_{n,l+1}\bigr)-\gamma^{\alpha}_{n,l}\bigl(K\circ\rho^{\alpha}_{n,l}\bigr)\Bigr)T^{l}\nonumber\\
&=&\gamma^{\alpha}_{n,1}\bigl(L\circ\rho^{\alpha}_{n,1}\bigr)-\gamma^{\alpha}_{n,0}\bigl(L\circ\rho^{\alpha}_{n,0}\bigr)\nonumber\\
&&\quad+\sum_{l=1}^{k}T^{-l}\Bigl(\gamma^{\alpha}_{n,l+1}\bigl(L\circ\rho^{\alpha}_{n,l+1}\bigr)-2\gamma^{\alpha}_{n,l}\bigl(L\circ\rho^{\alpha}_{n,l}\bigr)+\gamma^{\alpha}_{n,l-1}\bigl(L\circ\rho^{\alpha}_{n,l-1}\bigr)\Bigr)T^{l}\nonumber\\
&&\quad-T^{-(k+1)}\Bigl(\gamma^{\alpha}_{n,k+1}\bigl(L\circ\rho^{\alpha}_{n,k+1}\bigr)-\gamma^{\alpha}_{n,k}\bigl(L\circ\rho^{\alpha}_{n,k}\bigr)\Bigr)T^{k+1}.
\label{EQ79}
\end{eqnarray}
In what follows, we give upper bounds on the three terms arising in (\ref{EQ79}). Using the identity 
\begin{eqnarray}
&&\gamma^{\alpha}_{n,k+1}\bigl(L\circ\rho^{\alpha}_{n,k+1}\bigr)-\gamma^{\alpha}_{n,k}\bigl(L\circ\rho^{\alpha}_{n,k}\bigr)\nonumber\\
&&\qquad\qquad\qquad=\gamma^{\alpha}_{n,k+1}\Bigl(L\circ\rho^{\alpha}_{n,k+1}-L\circ\rho^{\alpha}_{n,k}\Bigr)-\alpha_{n,k}\bigl(L\circ\rho^{\alpha}_{n,k}\bigr),
\end{eqnarray}
and the fact $\gamma^{\alpha}_{n,0}=0$, $\gamma^{\alpha}_{n,1}=\alpha_{n,1}$, we can write
\begin{eqnarray}
\|\gamma^{\alpha}_{n,1}L(\rho^{\alpha}_{n,1}(s))-\gamma^{\alpha}_{n,0}L(\rho^{\alpha}_{n,0}(s))\|&\leq& \alpha_{n,1}\|L(\rho^{\alpha}_{n,1}(s))\|\nonumber\\
&\leq&\alpha_{n,1}\|L\|_{C^0}
\label{EQ80}
\end{eqnarray}
for any $s\in [t_1,t_2]$. Similarly,  
\begin{eqnarray}
\|\gamma^{\alpha}_{n,k+1}L(\rho^{\alpha}_{n,k+1}(s))-\gamma^{\alpha}_{n,k}L(\rho^{\alpha}_{n,k}(s))\|&\leq& \gamma^{\alpha}_{n,k+1}\|L(\rho^{\alpha}_{n,k+1}(s))-L(\rho^{\alpha}_{n,k}(s))\|\nonumber\\
&&\qquad\qquad+\alpha_{n,k}\|L(\rho^{\alpha}_{n,k}(s))\|\nonumber\\
&\leq&\gamma^{\alpha}_{n,k+1}\alpha_{n,k}(s-t_1)[L]_1+\alpha_{n,k}\|L\|_{C^0}\nonumber\\
&\leq&\alpha_{n,k}(s-t_1)[L]_1+\alpha_{n,k}\|L\|_{C^0}
\label{EQ81}
\end{eqnarray}
for any $s\in [t_1,t_2]$. To give an upper bound on the term $\sum_{l=1}^{k}\cdots$ of (\ref{EQ79}), we first do some manipulation on its summands using the identity $\gamma^{\alpha}_{n,l+1}-\gamma^{\alpha}_{n,l}=\alpha_{n,l}$:
\begin{eqnarray}
&&\gamma^{\alpha}_{n,l+1}\bigl(L\circ\rho^{\alpha}_{n,l+1}\bigr)-2\gamma^{\alpha}_{n,l}\bigl(L\circ\rho^{\alpha}_{n,l}\bigr)+\gamma^{\alpha}_{n,l-1}\bigl(L\circ\rho^{\alpha}_{n,l-1}\bigr)\nonumber\\
&&\qquad\qquad\qquad=\gamma^{\alpha}_{n,l+1}\bigl(L\circ\rho^{\alpha}_{n,l+1}-L\circ\rho^{\alpha}_{n,l}\bigr)-\gamma^{\alpha}_{n,l}\bigl(L\circ\rho^{\alpha}_{n,l}-L\circ\rho^{\alpha}_{n,l-1}\bigr)\nonumber\\
&&\qquad\qquad\qquad\qquad\qquad\qquad+\alpha_{n,l}\bigl(L\circ\rho^{\alpha}_{n,l}\bigr)-\alpha_{n,l-1}\bigl(L\circ\rho^{\alpha}_{n,l-1}\bigr)\nonumber\\
&&\qquad\qquad\qquad =\gamma^{\alpha}_{n,l+1}\Bigl(\bigl(L\circ\rho^{\alpha}_{n,l+1}-L\circ\rho^{\alpha}_{n,l}\bigr)-\bigl(L\circ\rho^{\alpha}_{n,l}-L\circ\rho^{\alpha}_{n,l-1}\bigr)\Bigr)\nonumber\\
&&\qquad\qquad\qquad\qquad\qquad\qquad+2\alpha_{n,l}\bigl(L\circ\rho^{\alpha}_{n,l}-L\circ\rho^{\alpha}_{n,l-1}\bigr)\nonumber\\
&&\qquad\qquad\qquad\qquad\qquad\qquad+(\alpha_{n,l}-\alpha_{n,l-1})\bigl(L\circ\rho^{\alpha}_{n,l-1}\bigr).
\end{eqnarray}
For any $s\in [t_1,t_2]$, the following upper bounds are nearly straightforward. First,
\begin{eqnarray}
&&\|\bigl(L(\rho^{\alpha}_{n,l+1}(s))-L(\rho^{\alpha}_{n,l}(s))\bigr)-\bigl(L(\rho^{\alpha}_{n,l}(s))-L(\rho^{\alpha}_{n,l-1}(s))\bigr\|\nonumber\\
&&\qquad\qquad\qquad\leq\|L(\rho^{\alpha}_{n,l+1}(s))-L(\rho^{\alpha}_{n,l}(s))\|+\|L(\rho^{\alpha}_{n,l}(s))-L(\rho^{\alpha}_{n,l-1}(s))\|\nonumber \\[4pt]
&&\qquad\qquad\qquad\leq \sup_{s\in [\rho^{\alpha}_{n,l-1}(s),\rho^{\alpha}_{n,l+1}(s)]}2\|L(\rho^{\alpha}_{n,l+1}(s))-L(\rho^{\alpha}_{n,l-1}(s))\|\nonumber\\[4pt]
&&\qquad\qquad\qquad\leq 2(s-t_1)|\alpha_{n,l}-\alpha_{n,l-1}|[L]_1.
\end{eqnarray}
Second,
\begin{eqnarray}
\|L(\rho^{\alpha}_{n,l}(s))-L(\rho^{\alpha}_{n,l-1}(s))\|\leq \alpha_{n,l-1}(s-t_1)[L]_1.
\end{eqnarray}
Finally,
\begin{equation}
\|L(\rho^{\alpha}_{n,l-1}(s))\|\leq \|L\|_{C^0}.
\end{equation}
Therefore,
\begin{eqnarray}
&&\left\|\sum_{l=1}^{k}T^{-l}\bigl(\gamma^{\alpha}_{n,l+1}L\circ\rho^{\alpha}_{n,l+1}-2\gamma^{\alpha}_{n,l}L\circ\rho^{\alpha}_{n,l}+\gamma^{\alpha}_{n,l-1}L\circ\rho^{\alpha}_{n,l-1}\bigr)T^l\right\|\nonumber\\
&&\qquad\leq (s-t_1)\sum_{l=1}^{k}\bigl(2\gamma^{\alpha}_{n,l+1}|\alpha_{n,l}-\alpha_{n,l-1}|+2\alpha_{n,l+1}\alpha_{n,l-1}\bigr)[L]_1\nonumber\\
&&\qquad\qquad\qquad+\sum_{l=1}^{k}|\alpha_{n,l}-\alpha_{n,l-1}|\|L\|_{C^0}
\label{EQ82}.
\end{eqnarray}
Combining (\ref{EQ80}), (\ref{EQ81}), (\ref{EQ82}) with the help of triangle inequality, we arrive to
\begin{eqnarray}
&&\|\mathfrak{P}^{(n,k)}[K](s)\|\leq \alpha_{n,1}\|L\|_{C^0} +\alpha_{n,k}(s-t_1)[L]_1+\alpha_{n,k}\|L\|_{C^0}\nonumber\\
&&\qquad\qquad\qquad\qquad+(s-t_1)\sum_{l=1}^{k}\bigl(2\gamma^{\alpha}_{n,l+1}|\alpha_{n,l}-\alpha_{n,l-1}|+2\alpha_{n,l+1}\alpha_{n,l-1}\bigr)[L]_1\nonumber\\
&&\qquad\qquad\qquad\qquad+\sum_{l=1}^{k}|\alpha_{n,l}-\alpha_{n,l-1}|\|L\|_{C^0}.
\label{EQ83}
\end{eqnarray}
Using $\alpha_{n,l},\gamma^{\alpha}_{n,l}\leq 1$, we get 
\begin{eqnarray*}
\|\mathfrak{P}^{(n,k)}[K](s)\|&\leq& 2\max_{0\leq l<n}\alpha_{n,l}\|L\|_{C^0} +(s-t_1)\max_{0\leq l<n}\alpha_{n,l}[L]_1\qquad\qquad\qquad\qquad\quad\\
&&\qquad+2(s-t_1)[L]_1\sum_{l=0}^{k-1}|\alpha_{n,l+1}-\alpha_{n,l}|\\
&&\qquad+2(s-t_1)[L]_1\max_{0\leq l<n}\alpha_{n,l}\sum_{l=0}^{k-1}\alpha_{n,l}\\
&&\qquad+\|L\|_{C^0}\sum_{l=0}^{k-1}|\alpha_{n,l+1}-\alpha_{n,l}|
\end{eqnarray*}
\begin{eqnarray}
&&\qquad\leq  \Bigl(2\|L\|_{C^0} +3(s-t_1)[L]_1\Bigr)\max_{0\leq l<n}\alpha_{n,l}\nonumber\\
&&\qquad\qquad\qquad+\Bigl(2(s-t_1)[L]_1+\|L\|_{C^0}\Bigr)\sum_{l=0}^{n-2}|\alpha_{n,l+1}-\alpha_{n,l}|\nonumber\\
&&\qquad\leq \Bigl(2\|L\|_{C^0} +3(s-t_1)[L]_1\Bigr)\Bigl(\max_{0\leq l<n}\alpha_{n,l}+\sum_{l=0}^{n-1}|\alpha_{n,l+1}-\alpha_{n,l}|\Bigr)\nonumber\\
&&\qquad\leq \psi(s-t_1)\Bigl(\max_{0\leq l<n}\alpha_{n,l}+\sum_{l=0}^{n-1}|\alpha_{n,l+1}-\alpha_{n,l}|\Bigr)\|L\|_{C^{0,1}}.
\end{eqnarray}
\end{proof}

\begin{Rk}
When $L\in C^{0,1}([t_1,t_2],\bbanach)$ is constant valued, than $[L]_1=0$ and with the choice $\alpha_{n,l}=l/n$, (\ref{EQ83}) reduces to (\ref{EQ78}). Nearly the same holds, if one considers (\ref{EQ84}): If $[L]_1=0$, then $\|L\|_{C^{0,1}}=\|L\|_{\bbanach}$ and (\ref{EQ84}) gives
\begin{equation}
\|\mathfrak{P}_{\alpha,T}[K](s)\|_\bbanach\leq \psi(s-t_1)\frac{\|L\|_\bbanach}{n}.
\end{equation}
Noting that the l.h.s of this equation does not depend on $s$, and $\min_{s\in [t_1,t_2]}\psi(s-t_1)=4$, we get back (\ref{EQ78}) up to a factor of two.
\end{Rk}

\begin{Cor}\label{COR5.3}
Given and admissible $\alpha$, an isometric isomorphism $T\in\bbanach$, $K\in C^{0,1}([t_1,t_2],\bbanach)$ such that $K=L-T^{-1}LT$, $L\in C^{0,1}([t_1,t_2],\bbanach)$, $p>0$ we have 
\begin{equation}
\bigr\|\bigl\{\mathfrak{P}^{(n,k)}_{\alpha,T}[K]\bigr\}_p(s,t_1)\bigl\|_{\bbanach}\leq \frac{\bigl(2+3(s-t_1)\bigr)^{2p}}{3^pp!}S^p_\alpha(n)\|L\|_{C^{0,1}}^p \, .
\end{equation}
\end{Cor}

\begin{proof}
\begin{eqnarray*}
\bigr\|\bigl\{\mathfrak{P}^{(n,k)}_{\alpha,T}[K]\bigr\}_p(s,t_1)\bigl\|&\leq& \int_{t_1}^{s}\mathrm{d}v_p\cdots \int_{t_1}^{v_2}\mathrm{d}v_1\,S^p_\alpha(n)\|L\|^p_{C^{0,1}}\prod_{k=1}^p\psi(v_k-t_1)\nonumber\\
&\leq& 2^pS^p_\alpha(n)\|L\|^p_{C^{0,1}}\int_{t_1}^{s}\mathrm{d}v_p\cdots \int_{t_1}^{v_2}\mathrm{d}v_1\,\prod_{k=1}^p\bigl(2+3(v_k-t_1)\bigr)\nonumber\\
&\leq& S^p_\alpha(n)\|L\|^p_{C^{0,1}}\left(\frac{2}{3}\right)^p\int_{2}^{2+3(s-t_1)}\mathrm{d}u_p\cdots \int_{2}^{u_2}\mathrm{d}u_1\,\prod_{k=1}^p u_k
\end{eqnarray*}
\begin{eqnarray}
&&\qquad\qquad\leq S^p_\alpha(n)\|L\|^p_{C^{0,1}}\left(\frac{2}{3}\right)^p\int_{0}^{2+3(s-t_1)}\mathrm{d}u_p\cdots \int_{0}^{u_2}\mathrm{d}u_1\,\prod_{k=1}^p u_k\nonumber\\
&&\qquad\qquad\leq\left(\frac{2}{3}\right)^p\frac{\bigl(2+3(s-t_1)\bigr)^{2p}}{2^pp!}S^p_\alpha(n)\|L\|_{C^{0,1}}^p.
\end{eqnarray}
\end{proof}

\begin{Thm}[Generalized Zeno dynamics of contractions]
Let $\banach$ be a Banach space with a decomposition $\banach=\banach_{\mathrm{I}}\oplus\banach_{\mathrm{C}}$ such that both $\banach_{\mathrm{I}}$ and $\banach_{\mathrm{C}}$ are closed subspaces of $\banach$ and the continuous projection $P_{\mathrm{I}}$ corresponding to $\banach_{\mathrm{I}}$ and its complementary projection $P_{\mathrm{C}}$ are norm one. Assume that $T\in\bbanach$ is of the form $T=T_{\mathrm{I}}\oplus T_{\mathrm{C}}$, where $T_\mathrm{I}\in\mathcal{B}(\banach_{\mathrm{I}})$ is an isometric isomorphism, while $T_\mathrm{C}\in\mathcal{B}(\banach_{\mathrm{C}})$ is a strict contraction.  Let $K\in C^0([t_1,t_2],\bbanach)$ and define the propagator $F(\cdot,\cdot)$ as the solution of the initial value problem
\begin{eqnarray}
\partial_1F(u,v)&=&K(u)F(u,v)\qquad\qquad t_{1}\leq v\leq u\leq t_{2}\nonumber\\
F(v,v)&=&\mathbbm{1}_{\banach}.
\end{eqnarray}
Let $G(\cdot,\cdot)$ be the solution of the initial value problem 
\begin{eqnarray}
\partial_1G(u,v)&=&\mathfrak{P}_{\alpha,T}[P_\mathrm{I}KP_\mathrm{I}](u)G(u,v)\qquad\qquad t_{1}\leq v\leq u\leq t_{2}\nonumber\\
G(v,v)&=&P_{\mathrm{I}},
\end{eqnarray}
Let $\alpha$ be an admissible set of sequences of positive numbers satisfying the additional condition
\begin{equation}
\lim_{n\rightarrow\infty}\sqrt{n}\sum_{l=0}^{n-1}|\alpha_{n,l+1}-\alpha_{n,l}|=0
\label{EQ71}
\end{equation}
Assume  that
\begin{enumerate}
\item the uniform ergodic mean $\mathfrak{P}_{\alpha,T_\mathrm{I}}[P_\mathrm{I}KP_\mathrm{I}]$ exists, and
\item as a member of the kernel of $\mathfrak{P}_{\alpha,T_\mathrm{I}}$, the difference $P_\mathrm{I}KP_\mathrm{I}-\mathfrak{P}_{\alpha,T_\mathrm{I}}[P_\mathrm{I}KP_\mathrm{I}]$ can be approximated by members of $\Ker^{0,1}(\mathfrak{P}_{\alpha,T_\mathrm{I}})$ of the form $L-T_\mathrm{I}^{-1}LT_{\mathrm{I}}$, $L\in C^0([t_1,t_2],\bbanach)$ with arbitrary precision in the $C^{0}$ norm.
\end{enumerate}
Then, 
\begin{equation}
\lim_{n\rightarrow\infty}T^{-n}_\mathrm{I} \prod_{l=0}^{n-1}TF\circ(\rho^{\alpha}_{n,l+1}\times\rho^{\alpha}_{n,l})=G(\cdot,t_1)
\end{equation}
in the $C^0$ norm.
\end{Thm}

\begin{Rk}\label{RK5.2}
If (\ref{EQ71}) holds, then 
\begin{equation}
\lim_{n\rightarrow\infty}\sqrt{n}\max_{1\leq l<n}\alpha_{n,l}=0
\label{EQ72}
\end{equation}
also holds. To see this, note that $\min_{0\leq l<n}\alpha_{n,l}<n^{-1}$ holds obviuosly, therefore $\lim_{n\rightarrow\infty}\min_{0\leq l<n}\alpha_{n,l}=0$. Assume that $\alpha_{n,p}=\max_{0\leq l<n}\alpha_{n,l}$ and  $\alpha_{n,q}=\min_{0\leq l<n}\alpha_{n,l}$. Then,
\begin{equation}
|\sqrt{n}\alpha_{n,p}-\sqrt{n}\alpha_{n,q}|=\sqrt{n}\left|\sum_{p\leq l<q}\alpha_{n,l+1}-\alpha_{n,l}\,\right|\leq\sqrt{n}\sum_{l=0}^{n-1}|\alpha_{n,l+1}-\alpha_{n,l}|,
\end{equation}
which implies (\ref{EQ72}).
\end{Rk}

\begin{proof} For the sake of brevity, let $K_\mathrm{I}\vcentcolon= P_\mathrm{I}KP_\mathrm{I}$ and $\overline{K}\vcentcolon=K_\mathrm{I}-\mathfrak{P}_{\alpha,T_{\mathrm{I}}}[K_\mathrm{I}]$. Note that $\mathfrak{P}_{\alpha,T_\mathrm{I}}[\overline{K}]=0$. Let $\overline{K}_{\varepsilon}=L_\varepsilon-T_\mathrm{I}^{-1}L_\varepsilon T_\mathrm{I}\in\Ker^{0,1}(\mathfrak{P}_{\alpha,T_\mathrm{I}})$, $L_\varepsilon\in C^0([t_1,t_2],\bbanach)$ such that $\|\overline{K}-\overline{K}_\varepsilon\|_{C^0}\leq\varepsilon$ for a given $\varepsilon>0$. We extend this notation to the case $\varepsilon=0$, such that $\overline{K}_0\equiv \overline{K}$. We define the propagators $F_{n,l,\varepsilon}(\cdot,\cdot)$, $n\in\mathbb{N}$, $0\leq l<n$; $G(\cdot,\cdot)$ and $H_{n,l,\varepsilon}(\cdot,\cdot)$, $n\in\mathbb{N}$, $0\leq l<n$ as the solutions of the following initial value problems:

\begin{eqnarray}
\partial_1F_{n,l,\varepsilon}(u,v)&=&\bigl(\mathfrak{P}_{\alpha,T_\mathrm{I}}[K_\mathrm{I}]+T_\mathrm{I}^{-l}\overline{K}_\varepsilon T_\mathrm{I}^{l}\bigr)(u)F_{n,l,\varepsilon}(u,v)\qquad t_{1}\leq v\leq u\leq t_{2},\nonumber\\
F_{n,l,\varepsilon}(v,v)&=&P_{\mathrm{I}},\nonumber\\
&&\nonumber\\
\partial_1G(u,v)&=&\bigl(\mathfrak{P}_{\alpha,T_\mathrm{I}}[K_\mathrm{I}]\bigr)(u)G(u,v)\qquad\qquad\qquad\quad\ t_{1}\leq v\leq u\leq t_{2},\nonumber\\
G(v,v)&=&P_{\mathrm{I}},\nonumber\\
&&\nonumber\\
\partial_1H_{n,l,\varepsilon}(u,v)&=&\mathfrak{P}_{\alpha,T_\mathrm{I}}^{(n,l)}[\overline{K}_\varepsilon](u)H_{n,l,\varepsilon}(u,v)\qquad\qquad\qquad\quad\ t_{1}\leq v\leq u\leq t_{2},\nonumber\\
H_{n,l,\varepsilon}(v,v)&=&P_{\mathrm{I}}.
\end{eqnarray}
Let $s\in [t_1,t_2]$ and let $\mathfrak{T}_n(s)=(\rho^{\alpha}_{n,0}(s),\dots,\rho^{\alpha}_{n,n}(s))$ denote the partition of the interval $[t_1,s]$.\par
\vspace{3mm}
\noindent\textbf{Step 1.} We prove
\begin{equation}
\lim_{n\rightarrow\infty}\left\|\prod_{l=0}^{n-1}TF\circ(\rho^{\alpha}_{n,l+1}\times\rho^{\alpha}_{n,l})-T^n_\mathrm{I}\prod_{l=0}^{n-1}F_{n,l,0}\circ(\rho^{\alpha}_{n,l+1}\times \rho^{\alpha}_{n,l})\right\|_{C^0}=0.
\label{EQT1ST1}
\end{equation}
\vspace{2mm}
\noindent For every $s\in [t_1,t_2]$, the application of Lemma \ref{LEM4.2} and statement \textit{7)} of Theorem \ref{THM2.1} gives 
\begin{eqnarray}
&&\left\|\prod_{l=0}^{n-1}TF\bigl(\rho^{\alpha}_{n,l+1}(s),\rho^{\alpha}_{n,l}(s)\bigr)-T^n_\mathrm{I}\prod_{l=0}^{n-1}F_{n,l,0}\bigl(\rho^{\alpha}_{n,l+1}(s),\rho^{\alpha}_{n,l}(s)\bigr)\right\|\nonumber\\ 
&&\qquad\qquad\leq n\,\Omega\,|\mathfrak{T}_n(s)|^2\,h[\mathfrak{T}_n(s)]+\Omega |\mathfrak{T}_n(s)|\,h_1[\mathfrak{T}_n(s)]+\Omega\,\|T_\mathrm{C}\|^n,\nonumber\\[6pt]
&&\qquad\qquad\leq n\,\Omega\,|\mathfrak{T}_n(s)|^2\bigl((1+\tau_{\mathrm{C}})\|K\|^2_{C^0}+\|K_\mathrm{I}\|^2_{C^0}+\mathcal{O}(|\mathfrak{T}_n(s)|)\bigr)\nonumber\\[6pt]
&&\qquad\qquad\qquad+\Omega |\mathfrak{T}_n(s)|\tau_{\mathrm{C}}\bigl(1+\mathcal{O}(|\mathfrak{T}_n(s)|)\bigr)+\Omega \|T_\mathrm{C}\|^n.
\label{EQ45}
\end{eqnarray}
Here, $\tau_{\mathrm{C}}$ is equal to $(1+\|T_\mathrm{C}\|)(1-\|T_\mathrm{C}\|)^{-1}$ and 
\begin{equation}
\Omega=\exp\left(\int_{t_1}^{s}\omega(v)\, \mathrm{d}v\right),
\end{equation}
where $\omega:[t_1,t_2]\rightarrow\mathbb{R}$ dominates the weak growth bounds of both $K,K_\mathrm{I}\in C([t_1,t_2],$
$\bbanach)$. 
Since $|\mathfrak{T}_n(s)|<(t_1-t_2)\max_{0\leq l<n}\alpha_{n,l}$, if (\ref{EQ71}) and so (\ref{EQ72}) holds, (\ref{EQT1ST1}) follows.\par
\vspace{3mm}
\noindent\textbf{Step 2.} Given $\varepsilon>0$, the following estimate holds:
\begin{eqnarray}
&&\left\|T^n_\mathrm{I}\prod_{l=0}^{n-1}F_{n,l,0}\circ(\rho^{\alpha}_{n,l+1}\times\rho^{\alpha}_{n,l})-T^n_\mathrm{I}\prod_{l=0}^{n-1}F_{n,l,\varepsilon}\circ(\rho^{\alpha}_{n,l+1}\times\rho^{\alpha}_{n,l})\right\|_{C^0}\nonumber\\
&&\qquad\qquad\qquad\leq \varepsilon(t_2-t_1)\mathrm{e}^{3\|K\|_{C^0}(t_2-t_1)}\mathrm{e}^{2\varepsilon(t_2-t_1)}.
\label{EQ74}
\end{eqnarray}
Let $\varepsilon>0$ and $s\in[t_1,t_2]$ be given. We can apply Gr\"onwall's Lemma for the integral inequality (\ref{EQ73}) in case of both $F_{n,l,\varepsilon}(\cdot,\cdot)$ and $F_{n,l,0}(\cdot,\cdot)$ and use Proposition \ref{PROP3.7} to obtain
\begin{eqnarray*}
&&\left\| T^n_\mathrm{I}\prod_{l=0}^{n-1}F_{n,l,0}\bigl(\rho^{\alpha}_{n,l+1}(s),\rho^{\alpha}_{n,l}(s)\bigr)-T^n_\mathrm{I}\prod_{l=0}^{n-1}F_{n,l,\varepsilon}\bigl(\rho^{\alpha}_{n,l+1}(s),\rho^{\alpha}_{n,l}(s)\bigr)\right\|\qquad\qquad\qquad\qquad\nonumber\\
&&\qquad =\left\|\prod_{l=0}^{n-1}F_{n,l,0}\bigl(\rho^{\alpha}_{n,l+1}(s),\rho^{\alpha}_{n,l}(s)\bigr)-\prod_{l=0}^{n-1}F_{n,l,\varepsilon}\bigl(\rho^{\alpha}_{n,l+1}(s),\rho^{\alpha}_{n,l}(s)\bigr)\right\|\nonumber\\
&&\qquad \leq\sum_{l=0}^{n-1}f_{n,l}(s)g_{n,l}(s)\Bigl\|F_{n,l,0}\bigl(\rho^{\alpha}_{n,l}(s),\rho^{\alpha}_{n,l-1}(s)\bigr)-F_{n,l,\varepsilon}\bigl(\rho^{\alpha}_{n,l}(s),\rho^{\alpha}_{n,l-1}(s)\bigr)\Bigr\|\nonumber\\
&&\qquad\leq\sum_{l=0}^{n-1}f_{n,l}(s)g_{n,l}(s)w_{n,l}(s)\int_{\rho^{\alpha}_{n,l-1}(s)}^{\rho^{\alpha}_{n,l}(s)}\|\overline{K}_0(v)-\overline{K}_\varepsilon(v)\|\,\mathrm{d}v
\end{eqnarray*}
\begin{eqnarray}
\leq\varepsilon(s-t_1)\sum_{l=0}^{n-1}\alpha_{n,l}f_{n,l}(s)g_{n,l}(s)w_{n,l}(s),\qquad\qquad\qquad\qquad\qquad\quad
\end{eqnarray}
where
\begin{eqnarray}
f_{n,l}(s)&=&\exp\Bigl(\|\mathfrak{P}_{\alpha,T_\mathrm{I}}[K]+\overline{K}_0\|_{C^0}\bigl(\rho^{\alpha}_{n,n}(s)-\rho^{\alpha}_{n,l+1}(s)\bigr)\Bigr)\nonumber\\
&\leq& \exp\Bigl(\|K\|_{C^0}(s-t_1)\Bigr),
\end{eqnarray}
\begin{eqnarray}
g_{n,l}(s)&=&\exp\Bigl(\|\mathfrak{P}_{\alpha,T_{\mathrm{I}}}[K]+\overline{K}_{\varepsilon}\|_{C^0}\bigl(\rho^{\alpha}_{n,l}(s)-\rho^{\alpha}_{n,0}(s) \bigr)\Bigr)\nonumber\\
&\leq& \exp\Bigl(\|K\|_{C^0}(s-t_1)\Bigr)\exp\Bigl(\varepsilon(s-t_1)\Bigr),
\end{eqnarray}
\begin{eqnarray}
w_{n,l}(s)&=&\exp\Bigl((\|K\|_{C^0}+\varepsilon)\bigl(\rho^{\alpha}_{n,l+1}(s)-\rho^{\alpha}_{n,l}(s)\bigr)\Bigr)\nonumber\\
&\leq& \exp\Bigl(\|K\|(s-t_1)\Bigr)\exp\Bigl(\varepsilon(s-t_1)\Bigr).
\end{eqnarray}
Therefore,
\begin{eqnarray}
&&\left\|T^n_\mathrm{I}\prod_{l=0}^{n-1}F_{n,l,0}\bigl(\rho^{\alpha}_{n,l+1}(s),\rho^{\alpha}_{n,l}(s)\bigr)-T^n_\mathrm{I}\prod_{l=0}^{n-1}F_{n,l,\varepsilon}\bigl(\rho^{\alpha}_{n,l+1}(s),\rho^{\alpha}_{n,l}(s)\bigr)\right\|\nonumber\\
&&\qquad\qquad\qquad\leq \varepsilon(s-t_1)\mathrm{e}^{3\|K\|_{C^0}(s-t_1)}\mathrm{e}^{2\varepsilon(s-t_1)},
\end{eqnarray}
from which (\ref{EQ74}) follows.\par
\vspace{3mm}
\noindent\textbf{Step 3.} We prove
\begin{eqnarray}
\lim_{n\rightarrow\infty}\left\|T^n_\mathrm{I}\prod_{l=0}^{n-1}F_{n,l,\varepsilon}\circ(\rho^{\alpha}_{n,l+1}\times\rho^{\alpha}_{n,l})-T^n_\mathrm{I}H_{n,n,\varepsilon}(\cdot,t_1)G(\cdot,t_1)\right\|_{C^0}=0.
\label{T1S3}
\end{eqnarray}
Note that for any $s\in[t_1,t_2]$, $\rho^{\alpha}_{n,0}(s)=t_1$ and $\rho^{\alpha}_{n,n}(s)=s$. We write 
\begin{eqnarray}
&&T^{n}_{\mathrm{I}}\prod_{l=0}^{n-1}F_{n,l,\varepsilon}\bigl(\rho^{\alpha}_{n,l+1}(s),\rho^{\alpha}_{n,l}(s)\bigr)-T_\mathrm{I}^nH_{n,n,\varepsilon}\bigl(\rho^{\alpha}_{n,n}(s),\rho^{\alpha}_{n,0}(s)\bigr)G\bigl(\rho^{\alpha}_{n,n,\varepsilon}(s),\rho^{\alpha}_{n,0}(s)\bigr)\nonumber\\
&&\qquad\qquad=T^n_{\mathrm{I}}\sum_{k=0}^{n-1}\left(\prod_{l=k+1}^{n-1}F_{n,l,\varepsilon}\bigl(\rho^{\alpha}_{n,l+1}(s),\rho^{\alpha}_{n,l}(s)\bigl)\right)\nonumber\\
&&\qquad\qquad\qquad\times\Delta^{\varepsilon}_{n,k}\bigl(\rho^{\alpha}_{n,k+1}(s),\rho^{\alpha}_{n,k}(s),\rho^{\alpha}_{n,k-1}(s)\bigr)G\bigl(\rho^{\alpha}_{n,k-1}(s),\rho^{\alpha}_{n,0}(s)\bigr),
\label{EQ64}
\end{eqnarray}
where
\begin{eqnarray}
&&\Delta^{\varepsilon}_{n,k}\bigl(\rho^{\alpha}_{n,k+1}(s),\rho^{\alpha}_{n,k}(s),\rho^{\alpha}_{n,k-1}(s)\bigr)\nonumber\\
&&\qquad\qquad=F_{n,k,\varepsilon}\bigl(\rho^{\alpha}_{n,k+1}(s),\rho^{\alpha}_{n,k}(s)\bigr)H_{n,k-1,\varepsilon}(\rho^{\alpha}_{n,k}(s),\rho^{\alpha}_{n,0}(s)\bigr)\nonumber\\
&&\qquad\qquad\qquad\qquad-H_{n,k,\varepsilon}\bigl(\rho^{\alpha}_{n,k+1}(s),\rho^{\alpha}_{n,0}(s)\bigr)G\bigr(\rho^{\alpha}_{n,k+1}(s),\rho^{\alpha}_{n,k}(s)\bigl).
\label{EQ65}
\end{eqnarray}
The various terms in (\ref{EQ64}) and (\ref{EQ65}) can be estimated as follows. First, using Proposition \ref{prop:growth_prop}, we get
\begin{eqnarray}
\left\|\prod_{l=k+1}^{n-1}F_{n,l,\varepsilon}\bigl(\rho^{\alpha}_{n,l+1}(s),\rho^{\alpha}_{n,l}(s)\bigr)\right\|&\leq &\int_{\rho^{\alpha}_{n,k+1}(s)}^{\rho^{\alpha}_{n,n}(s)}\omega(v)\,\mathrm{d}v,\nonumber\\
\left\|G\bigl(\rho^{\alpha}_{n,k-1}(s),\rho^{\alpha}_{n,0}(s)\bigr)\right\|&\leq &\int_{\rho^{\alpha}_{n,0}(s)}^{\rho^{\alpha}_{n,k-1}(s)}\omega(v)\,\mathrm{d}v.
\label{EQ75}
\end{eqnarray}
Beyond that of Corollary \ref{COR5.3}, we have the estimates 
\begin{eqnarray}
\left\|\bigl\{T_\mathrm{I}^{-k}K_\mathrm{I}T_\mathrm{I}^{k}\bigr\}_{p}\bigl(\rho^{\alpha}_{n,k+1}(s),\rho^{\alpha}_{n,k}(s)\bigr)\right\|&\leq& \frac{\alpha_{n,k}^p(s-t_1)^p}{p!}\|K_{\mathrm{I}}\|^p_{C^0},\nonumber\\
\bigl\{\mathfrak{P}_{\alpha,T_\mathrm{I}}[K_\mathrm{I}]\bigr\}_{p}\bigl(\rho^{\alpha}_{n,k+1}(s),\rho^{\alpha}_{n,k}(s)\bigr)&\leq&\frac{\alpha_{n,k}^p(s-t_1)^p}{p!}\|\mathfrak{P}_{\alpha,T_\mathrm{I}}[K_\mathrm{I}]\|_{C^{0}}
\end{eqnarray}
for all $p\in \mathbb{N}_0$. Therefore, the application of (\ref{F4}) to (\ref{EQ65}) results in 
\begin{eqnarray}
&&\|\Delta^{\varepsilon}_{n,k}\bigl(\rho^{\alpha}_{n,k+1}(s),\rho^{\alpha}_{n,k}(s),\rho^{\alpha}_{n,k-1}(s)\bigr)\|\leq \sum_{p=1}^2\bigl(f^{(p)}_{n}(s-t_1)\bigr)^2\exp(f^{(p)}_{n}(s-t_1)),
\label{EQ76}
\end{eqnarray}
where the error functions $f^{(1)}_n$ and $f^{(2)}_n$ are given by
\begin{eqnarray}
f^{(1)}_n(x)&=&x\max_{0\leq l<n}\alpha_{n,l}\|K_{\mathrm{I}}\|_{C^0}+\frac{(2+3x)^2}{3}D_\alpha(n)\|L_\varepsilon\|_{C^{0,1}},\nonumber\\
f^{(2)}_n(x)&=&x\max_{0\leq l<n}\alpha_{n,l}\|\mathfrak{P}[K_{\mathrm{I}}]\|_{C^0}+\frac{(2+3x)^2}{3}D_\alpha(n)\|L_\varepsilon\|_{C^{0,1}}.
\end{eqnarray}
Therefore, 
\begin{eqnarray}
&&\left\|T^n_\mathrm{I}\prod_{l=0}^{n-1}F_{n,l,\varepsilon}\circ(\rho^{\alpha}_{n,l+1}\times\rho^{\alpha}_{n,l})(s)-T^n_\mathrm{I}H_{n,n,\varepsilon}(s,t_1)G(s,t_1)\right\|_{C^0}\nonumber\\
&&\qquad\leq \Omega\sum_{l=0}^{n-1}\exp\left(-\int_{\rho^{\alpha}_{n,l}(s)}^{\rho^{\alpha}_{n,l+1}(s)}\omega(v)\,\mathrm{d}v\right)\sum_{p=1}^2\bigl(f^{(p)}_{n}(s-t_1)\bigr)^2\exp(f^{(p)}_{n}(s-t_1))\nonumber\\
&&\qquad\leq n\Omega\exp\left((t_2-t_2)\|\omega\|_{L^1}\max_{0\leq l<n}\alpha_{n,l}\right)\sum_{p=1}^2\bigl(f^{(p)}_{n}(s-t_1)\bigr)^2\exp(f^{(p)}_{n}(s-t_1)).
\end{eqnarray}
Since (\ref{EQ71}) holds by assumption, when combined with Remark \ref{RK5.2}, it guarantees the convergence
\begin{equation}
\lim_{n\rightarrow\infty}f^{(p)}_n=0
\end{equation}
in the $C^0$ norm. Therefore, (\ref{T1S3}) follows.\par
\vspace{3mm}
\noindent\textbf{Step 4.}
\begin{equation}
\lim_{n\rightarrow\infty}\left\|T^n_\mathrm{I}H_{n,n,\varepsilon}(\cdot,t_1)G(\cdot,t_1)-T^n_\mathrm{I}G(\cdot,t_1)\right\|_{C^0}=0.
\label{T1S4}
\end{equation}
We can apply Proposition \ref{prop:growth_prop} and \ref{PROP3.7} to obtain 
\begin{eqnarray}
\left\|T^n_\mathrm{I}\bigl(H_{n,n,\varepsilon}(s,t_1)-\mathbbm{1}_{\banach}\bigr)G(s,t_1)\right\|&\leq& \exp\left(\int_{t_1}^{s}\omega(v)\,\mathrm{d}v\right)\|H_{n,n,\varepsilon}(s,t_1)-\mathbbm{1}_{\banach}\|\nonumber\\
&\leq& \exp\left(\int_{t_1}^{s}\omega(v)\,\mathrm{d}v\right)\|\mathfrak{P}^{(n)}_{\alpha,T_{\mathrm{I}}}[\overline{K}_{\varepsilon}](s)\|
\end{eqnarray}
for any $s\in [t_1,t_2]$. Since $\mathfrak{P}^{(n)}_{\alpha,T_{\mathrm{I}}}[\overline{K}_{\varepsilon}]\rightarrow 0$ in the $C^0$ norm and $\omega$ is integrable, we arrive to (\ref{T1S4}).
\par
\vspace{3mm}
\noindent
Finally, for any $\varepsilon>0$ we have
\begin{eqnarray}
&&\lim_{n\rightarrow\infty}\left\|T^{-n}_\mathrm{I}\prod_{l=0}^{n-1}TF\bigl(\rho^{\alpha}_{n,l+1}\times\rho^{\alpha}_{n,l}\bigr)-G(\cdot,t_1)\right\|\nonumber\\
&&\qquad\qquad\qquad\leq \varepsilon(t_2-t_1)\mathrm{e}^{3\|K\|_{C^0}(t_2-t_1)}\mathrm{e}^{2\varepsilon(t_2-t_1)},
\end{eqnarray}
therefore the Theorem is proved.
\end{proof}
\begin{Rk} \label{RK5.6}
Let $T\in\bbanach$ an isometric isomorphism, the admissible $\alpha$ defined through $\alpha_{n,l}=l/n$, $n\in\mathbb{N}$, $0\leq l<n$. Consider the interval $[0,t]$, $0<t<\infty$. Let $K\in\bbanach$ be equal to $K=M+L-T^{-1}LT$, where $M,L\in\bbanach$ and $M$ commutes with $T$. Then, the calculation in the third step of the proof can be repeated for the constant valued $K$ if $\mathfrak{P}_{\alpha,T_\mathrm{I}}[K]$ is replaced by $M$. Then, the following relation, which will be used later in section \ref{SEC6} can be derived:
\begin{eqnarray}
&&\left\|T^n\exp\left(\sum_{l=0}^{n-1}T^{-l}(L-T^{-1}LT)T^{l}\frac{t}{n}\right)\exp(Mt)-T^n\prod_{l=0}^{n-1}\exp\left(T^{-l}KT^l\frac{t}{n}\right)\right\|_\bbanach\nonumber\\
&&\qquad\qquad=\left\|T^n\exp\left((L-T^{-n}LT^{n})\frac{t}{n}\right)\exp(Mt)-T^n\prod_{l=0}^{n-1}\exp\left(T^{-l}KT^{l}\frac{t}{n}\right)\right\|_\bbanach\nonumber\\
&&\qquad\qquad\leq\exp\left(\omega t\right)\Bigl(\|M\|_{\bbanach}+2\|L\|_{\bbanach}\Bigr)^2\,\frac{2t^2}{n}\nonumber\\
&&\qquad\qquad\qquad\times\exp\left(\Bigl(|\omega|+\|M\|_{\bbanach}+2\|L\|_{\bbanach}\Bigr)\frac{t}{n}\right),
\end{eqnarray}
if $\omega$ dominates both weak growth bounds $\omega_K$ and $\omega_M$.
\end{Rk}

\section{Generalized adiabatic theorem}
\label{SEC6}
We now turn to the proof of the generalized adiabatic theorem. In order to be comprehensive, we begin with a specialized definition of ergodic means with respect to uniformly continuous isometric isomorphism groups and a standard Theorem concerning them.
\begin{Def} Let $\banach$ be a Banach space and $(T(t))_{t\in\mathbb{R}}\subset\bbanach$ be a  uniformly continuous group of isometric isomorphisms.  We say that the uniform ergodic mean of $K \in\bbanach $ with respect to $T$ exists, if 
\begin{equation}\label{EQ87}
\mathfrak{P}_{T}[K]\vcentcolon=\lim_{S\rightarrow\infty}\frac{1}{S}\int_{0}^ST^{-1}(s)KT(s)\,\mathrm{d}s 
\end{equation} 
exists in the uniform topology of $\bbanach$.
\end{Def}

The set of all $K\in\bbanach$ for which (\ref{EQ87}) exists is denoted by $\Dom(\mathfrak{P}_{T})$.

\begin{Thm}\label{THM6.2} Let $\banach$ be a Banach space, $(T(t))_{t\in\mathbb{R}}\subset\bbanach$ be a uniformly continuous group of isometric isomorphisms generated by $A\in\bbanach$.
\begin{enumerate}
\item $\Dom(\mathfrak{P}_{T})$ is a closed subspace of $\bbanach$ on which $\mathfrak{P}_{T}$ acts continuously.
\item For every $K\in\Dom(\mathfrak{P}_{T})$, $\mathfrak{P}_{T}[K]$ commutes with all members of $(T(t))_{t\in\mathbb{R}}$.
\item $\mathfrak{P}_{T}$ is an involution on $\Dom(\mathfrak{P}_{T})$.
\item The kernel of $\mathfrak{P}_{T}$ is the uniform closure of the linear space 
\begin{equation}
\mathcal{N}\vcentcolon=\mathrm{Lin}\{L-T^{-1}(v)LT(v):L\in\bbanach,\ v\in \mathbb{R}\}.
\end{equation}
\item The kernel of $\mathfrak{P}_{T}$ is the uniform closure of the linear space 
\begin{equation}
\mathcal{C}_A\vcentcolon=\{[L,A]:L\in\bbanach\}.
\end{equation}
\item Let $K\in \Dom(\mathfrak{P}_{T})$ with weak growth bound $\omega_K$. Then, the weak growth bound of $\mathfrak{P}_{T}[K]$ satisfies $\omega_{\mathfrak{P}_{T}[K]}\leq \omega_K$.
\end{enumerate}
\end{Thm}

\begin{proof}
$\ $\par
\noindent\textit{1)} Follows from the bound
\begin{equation}
\left\|\frac{1}{S}\int_{0}^{S}T^{-1}(s)KT(s)\, \mathrm{d}s\right\|\leq \|K\|
\end{equation}
and a standard $3\varepsilon$ argument.\par
\vspace{3mm}
\noindent\textit{2)} For every $S> 0$ and $t\in\mathbb{R}$, we have
\begin{eqnarray}
\frac{T(t)}{S}\int_{0}^{S}T^{-1}(s)KT(s)\,\mathrm{d}s&=&\frac{1}{S}\int_{0}^{S}T^{-1}(s-t)KT(s)\,\mathrm{d}s\nonumber\\
&=&\frac{1}{S}\int_{t}^{S+t}T^{-1}(s)KT(s+t)\,\mathrm{d}s\nonumber\\
&=&\frac{1}{S}\int_{t}^{S+t}T^{-1}(s)KT(s)\,\mathrm{d}sT(t).
\end{eqnarray}
Therefore, 
\begin{eqnarray}
&&\lim_{S\rightarrow \infty}\left\|\left[T(t),\frac{1}{S}\int_{0}^{S}T^{-1}(s)KT(s)\,\mathrm{d}s\right]\right\|\nonumber\\
&&\qquad\qquad=\lim_{S\rightarrow \infty}\left\|\frac{1}{S}\int_{t}^{S+t}T^{-1}(s)KT(s)\,\mathrm{d}s-\frac{1}{S}\int_{0}^{S}T^{-1}(s)KT(s)\,\mathrm{d}s\right\|\nonumber\\
&&\qquad\qquad=\lim_{S\rightarrow \infty}\left\|\frac{1}{S}\int_{S}^{S+t}T^{-1}(s)KT(s)\,\mathrm{d}s-\frac{1}{S}\int_{0}^{t}T^{-1}(s)KT(s)\,\mathrm{d}s\right\|\nonumber\\
&&\qquad\qquad\leq \lim_{S\rightarrow \infty}\frac{2\|K\|t}{S}=0.
\end{eqnarray}
\vspace{3mm}
\noindent\textit{3)} Straightforward consequence of \textit{2)}.\par
\vspace{3mm}
\noindent\textit{4)} Let $t\in\mathbb{R}$, $L\in\bbanach$. Then, 
\begin{eqnarray}
&&\frac{1}{S}\int_{0}^ST^{-1}(s)\bigl(L-T^{-1}(t)LT(t)\bigr)T(s)\,\mathrm{d}s\nonumber\\
&&\qquad\qquad=\frac{1}{S}\int_{0}^ST^{-1}(s)LT(s)\,\mathrm{d}s-\frac{1}{S}\int_{t}^{S+t}T^{-1}(s)LT(s)\,\mathrm{d}s\nonumber\\
&&\qquad\qquad=\frac{1}{S}\left(\int_{0}^{t}T^{-1}(s)LT(s)\,\mathrm{d}s-\int_{S}^{S+t}T^{-1}(s)LT(s)\,\mathrm{d}s\right).
\end{eqnarray}
Therefore, 
\begin{equation}
\left\|\frac{1}{S}\int_{0}^ST^{-1}(s)\bigl(L-T^{-1}(t)LT(t)\bigr)T(s)\,\mathrm{d}s\right\|\leq \frac{2\|L\|t}{S},
\end{equation}
so $\mathfrak{P}_{T}[\mathcal{N}]=0$.

Assume that $K\in \Ker(\mathfrak{P}_{T})\setminus\,\overline{\mathcal{N}}$, where $\overline{\mathcal{N}}$ is the $\|\cdot\|_{\bbanach}$-closure of $\mathcal{N}$. Since $K\neq 0$, provided by the Hahn-Banach separation theorem, there exists $\varphi\in\bbanach^{*}$ such that $\varphi(K)=1$ and $\varphi(\overline{\mathcal{N}})=0$. Let $t\in\mathbb{R}$ be arbitrary, then 
\begin{eqnarray}
\varphi(K)&=&\varphi(K)-\varphi(T^{-1}(t)KT(t))+\varphi(T^{-1}(t)KT(t))\nonumber\\
&=&\varphi(K-T^{-1}(t)KT(t))+\varphi(T^{-1}(t)KT(t))\nonumber\\
&=&\varphi(T^{-1}(t)KT(t)).
\end{eqnarray}
Therefore, 
\begin{eqnarray}
\varphi\left(\frac{1}{S}\int_{0}^{S}T^{-1}(s)KT(s)\,\mathrm{d}s\right)&=&\frac{1}{S}\int_{0}^{S}\varphi\left(T^{-1}(s)KT(s)\right)\,\mathrm{d}s\nonumber\\
&=&\frac{1}{S}\int_{0}^{S}\varphi\left(K\right)\,\mathrm{d}s\nonumber\\
&=&\varphi\left(K\right).
\end{eqnarray}
However, then we obtain
\begin{equation}
0=\varphi(0)=\varphi(\mathfrak{P}_{T}[K])=\lim_{S\rightarrow\infty}\varphi\left(\frac{1}{S}\int_{0}^{S}T^{-1}(s)KT(s)\,\mathrm{d}s\right)=\varphi(K)=1,
\end{equation}
which is impossible.\par
\vspace{3mm}
\noindent\textit{5)} First, we prove that $\mathcal{C}_A\subseteq \Ker(\mathfrak{P}_{T})$, then we show that $\mathcal{C}_A$ is dense in $\mathcal{N}$. To prove the first statement, note that for any $L\in\bbanach$ and any $\varepsilon>0$, 
\begin{equation}
\frac{L-T^{-1}(\varepsilon)LT(\varepsilon)}{\varepsilon}\in\mathcal{N}.
\end{equation}
Since $\Ker(\mathfrak{P}_{T})$ is closed and $\mathfrak{P}_{T}$ is continuous on its domain,  we have
\begin{equation}
\mathfrak{P}_{T}\bigl[[L,A]\bigr]=\lim_{\varepsilon\rightarrow 0}\mathfrak{P}_{T}\left[\frac{L-T^{-1}(\varepsilon)LT(\varepsilon)}{\varepsilon}\right]=0.
\end{equation}
Therefore $\mathcal{C}_A\subseteq \Ker(\mathfrak{P}_{T})$. To see that $\mathcal{C}_A$ is dense in $\mathcal{N}$, first observe 
\begin{equation}
L-T^{-1}(v)LT(v)=-\int_{0}^{v}\bigl(\partial_t(T^{-1}LT)\bigr)(s)\,\mathrm{d}s=-\int_{0}^{v}[T^{-1}(s)LT(s),A]\,\mathrm{d}s.
\end{equation} 
Since $T(\cdot)$ is uniformly continuous, for any $v\in\mathbb{R}$, $v> 0$, $L\in\bbanach$ and any $\varepsilon>0$, there exists $N(\varepsilon,L,v)\in\mathbb{N}$, such that  if $n>N(\varepsilon,L,v)$, then
\begin{equation}
\max_{0\leq l\leq n}\sup_{s\in [lv/n,(l+1)v/n]}\|[T^{-1}(lv/n)LT(lv/n),A]-[T^{-1}(s)LT(s),A]\|\leq \varepsilon,
\end{equation}
that is 
\begin{eqnarray}
\left\|\frac{v}{n}\sum_{l=0}^{n}[T^{-1}(lv/n)LT(lv/n),A]-\int_{0}^{v}[T^{-1}(s)LT(s),A]\right\|\leq v\varepsilon,
\end{eqnarray}
whenever $n>N(\varepsilon,L,v)$. Since 
\begin{eqnarray}
\frac{v}{n}\sum_{l=0}^{n}[T^{-1}(lv/n)LT(lv/n),A]\in\mathcal{C}_A,
\end{eqnarray}
any member of $\mathcal{N}$ of the form $L-T^{-1}(v)LT(v)$, $v>0$, $L\in\bbanach$ can be approximated by members of $\mathcal{C}_A$. For $v=0$, the approximation is trivial, while for $v<0$, a similar argument works. Since members of the form $L-T^{-1}(v)LT(v)$, $v\in\mathbb{R}$, $L\in\bbanach$ generates $\mathcal{N}$, it follows that $\mathcal{C}_A$ is dense in $\mathcal{N}$, therefore its closure is the kernel of $\mathfrak{P}_T$ provided by \textit{4)}.\par
\vspace{3mm}
\noindent\textit{6)} Let $S>0$ and $v\in\mathbb{R}_0^{+}$. For any $0\leq s\leq S$, provided that $T(s)$ is an isometric isomorphism, we have $\omega_{T^{-1}(s)KT(s)}=\omega_{K}$. Using Corollary \ref{COR3.5}, we obtain
\begin{equation}
\left\|\exp\left(\frac{v}{S}\int_{0}^{S}T^{-1}(s)KT(s)\,\mathrm{d}s\right)\right\|\leq\exp\left(\frac{v}{S}\int_{0}^{S}\omega_{T^{-1}(s)KT(s)}\,\mathrm{d}s\right)=\exp(\omega_Kv). 
\end{equation}
Therefore,
\begin{eqnarray}
\left\|\exp\left(\mathfrak{P}_{T}[K]v\right)\right\|&=&\lim_{S\rightarrow\infty}\left\|\exp\left(\frac{v}{S}\int_{0}^{S}T^{-1}(s)KT(s)\,\mathrm{d}s\right)\right\|\nonumber\\
&\leq&\exp\left(\omega_Kv\right),
\end{eqnarray}
which proves the statement.
\end{proof}

\begin{Thm}\label{THM6.3} Let $A\in\bbanach$ be the  generator of $(T(t))_{t\in\mathbb{R}}$, a uniformly continuous group of isometric isomorphisms.  Let $K\in\bbanach$ and assume that $\mathfrak{P}_{T}[K]$ exists. Then, for any $t\in \mathbb{R}_0^+$,
\begin{equation}
\lim_{0<\gamma\rightarrow\infty}\left\|\exp\bigl((\gamma A+K)t\bigr)-T(\gamma t)\exp(\mathfrak{P}_{T}[K] t)\right\|_\bbanach=0.
\label{EQ91}
\end{equation}
\end{Thm}

\begin{proof}
We define $T_\gamma(t)\vcentcolon =T(\gamma t)$ for all $\gamma>0$ and $t\in\mathbb{R}_0^{+}$. A short calculation gives
\begin{eqnarray}
\exp\bigl((\gamma A+K)t\bigr)&=&\lim_{n\rightarrow \infty}\prod_{l=0}^{n-1}\exp(\gamma At/n)\exp(K t/n)\nonumber\\
&=&T_\gamma(t)\lim_{n\rightarrow \infty}\prod_{l=0}^{n-1}\exp\left(T_\gamma^{-l}(t/n)K T^{l}_\gamma(t/n)\frac{t}{n}\right).
\label{EQ87a}
\end{eqnarray}
\noindent\textbf{Step 1.} Let $K_\varepsilon\in\bbanach$ satisfying $\|K-K_\varepsilon\|\leq \varepsilon$. Then, 
\begin{equation}
\left\|\exp\bigl((\gamma A+K)t\bigr)-\exp\bigl((\gamma A+K_\varepsilon)t\bigr)\right\|\leq \varepsilon t\exp\left(\bigl(\|K\|+\varepsilon\bigr)t\right).
\label{EQ87b}
\end{equation}
To see this, note that the application of Proposition \ref{PROP3.7} with the trivial upper bound $\|K\|+\varepsilon$ on the weak growth bounds $\omega_K$ and $\omega_{K_\varepsilon}$ gives 
\begin{eqnarray}
&&\left\|\exp\left(T_\gamma^{-l}(t/n)K T^{l}_\gamma(t/n)\frac{t}{n}\right)-\exp\left(T_\gamma^{-l}(t/n)K_\varepsilon T^{l}_\gamma(t/n)\frac{t}{n}\right)\right\|\nonumber\\
&&\qquad\leq\exp\left(\bigl(\|K\|+\varepsilon\bigr)\frac{t}{n}\right)\int^{t/n}_{0}\|T_\gamma^{-l}(t/n)(K-K_\varepsilon)T^l_\gamma(t/n)\|\,\mathrm{d}s\nonumber\\
&&\qquad\leq\varepsilon \exp\left(\bigl(\|K\|+\varepsilon\bigr)\frac{t}{n}\right)\frac{t}{n}.
\end{eqnarray}
Therefore, using (\ref{F1}), we obtain
\begin{eqnarray}
&&\Biggl\|\prod_{l=0}^{n-1}\exp\left(T_\gamma^{-l}(t/n)K T^{l}_\gamma(t/n)\frac{t}{n}\right)-\prod_{l=0}^{n-1}\exp\left(T_\gamma^{-l}(t/n)K_\varepsilon T^{l}_\gamma(t/n)\frac{t}{n}\right)\Biggr\|\nonumber\\
&&\qquad\qquad\qquad \leq \varepsilon t\exp\left(\bigl(\|K\|+\varepsilon\bigr)t\right),
\end{eqnarray}
from which, when combined with (\ref{EQ87a}), (\ref{EQ87b}) follows.\par
\vspace{3mm}\par
\noindent\textbf{Step 2.} Let $F\in\Dom(\mathfrak{P}_{T})$ be of the form $F=\mathfrak{P}_{T}[F]+[L,A]$ for some $L\in \bbanach$. Then, 
\begin{equation}
\left\|\exp\bigl((\gamma A+F)t\bigr)-T_\gamma(t)\exp\bigl(\mathfrak{P}_{T}[F]t\bigr)\right\|\leq \frac{2\|L\|}{\gamma}\exp\left(\frac{2\|L\|}{\gamma}\right).
\label{EQ90}
\end{equation}
To see this, we first use statement \textit{2)} of Theorem \ref{THM6.2} to get
\begin{eqnarray}
&&T_\gamma(t)\lim_{n\rightarrow \infty}\prod_{l=0}^{n-1}\exp\left(T_\gamma^{-l}(t/n)F T^{l}_\gamma(t/n)\frac{t}{n}\right)\nonumber\\
&&\qquad\qquad=T_\gamma(t)\lim_{n\rightarrow \infty}\prod_{l=0}^{n-1}\exp\left(\mathfrak{P}_{T}[F]\frac{t}{n}+T_\gamma^{-l}(t/n)[L,A]T^{l}_\gamma(t/n)\frac{t}{n}\right).
\end{eqnarray}
Let $\varepsilon>0$ be given. Since $T(\cdot)$ is uniformly continuous on the interval $[0,\gamma t]$, there exists $N(\varepsilon,\gamma t)\in\mathbb{N}$, such that if $n>N(\varepsilon,\gamma t)$, then 
\begin{equation}
\max_{0\leq l<n-1}\sup_{s\in \left[\frac{lt}{n},\frac{(l+1)t}{n}\right]}\left\|T(s)-T\left(\frac{lt}{n}\right)\right\|\leq \varepsilon.
\end{equation}
But in that case, a short calculation shows 
\begin{eqnarray}
&&\left\|T_\gamma^{-l}(t/n)[L,A]T^{l}_\gamma(t/n)\frac{t}{n}-\int_{lt/n}^{(l+1)t/n}T^{-1}_\gamma(s)[L,A]T_\gamma(s)\,\mathrm{d}s\right\|\nonumber\\
&&\qquad=\left\|T_\gamma^{-1}(lt/n)[L,A]T_\gamma(lt/n)\frac{t}{n}-\int_{lt/n}^{(l+1)t/n}T_\gamma^{-1}(s)[L,A]T_\gamma(s)\,\mathrm{d}s\right\|\nonumber\\
&&\qquad\leq 2\varepsilon \|[L,A]\|\frac{t}{n}.
\label{EQ88}
\end{eqnarray}
Therefore, combining the upper bounds 
\begin{eqnarray}
&&\left\|\mathfrak{P}_{T}[F]\frac{t}{n}+\int_{lt/n}^{(l+1)t/n}T^{-1}_\gamma(s)[L,A]T_\gamma(s)\,\mathrm{d}s\right\|\nonumber\\
&&\qquad\qquad\qquad=\left\|\int_{lt/n}^{(l+1)t/n}T^{-1}_\gamma(s)\bigl(\mathfrak{P}_{T}[F]+[L,A]\bigl)T_\gamma(s)\,\mathrm{d}s\right\|\nonumber\\
&&\qquad\qquad\qquad\leq \|F\|\frac{t}{n},\nonumber\\
&&\left\|T^{-1}_\gamma(lt/n)FT_\gamma(lt/n)\frac{t}{n}\right\|\leq \|F\|\frac{t}{n}
\end{eqnarray}
with (\ref{EQ88}), Proposition \ref{PROP3.7} and (\ref{F1}), we can arrive to
\begin{eqnarray}
&&\Biggl\|T_\gamma(t)\prod_{l=0}^{n-1}\exp\left(T_\gamma^{-l}(t/n)F T^{l}_\gamma(t/n)\frac{t}{n}\right)\nonumber\\
&&\qquad-T_\gamma(t)\prod_{l=0}^{n-1}\exp\left(\mathfrak{P}_{T}[F]\frac{t}{n}+\int_{lt/n}^{(l+1)t/n}T^{-1}_\gamma(s)[L,A]T_\gamma(s)\,\mathrm{d}s\right)\Biggr\|\nonumber\\
&&\leq 2\varepsilon \|[L,A]\|\exp(\|F\|t).
\end{eqnarray}
Since $\varepsilon>0$ can be chosen arbitrarily, we have
\begin{eqnarray}
&&\lim_{n\rightarrow\infty}T_\gamma(t)\prod_{l=0}^{n-1}\exp\left(T_\gamma^{-l}(t/n)F T^{l}_\gamma(t/n)\frac{t}{n}\right)\nonumber\\
&&\qquad=\lim_{n\rightarrow\infty}T_\gamma(t)\prod_{l=0}^{n-1}\exp\left(\mathfrak{P}_{T}[F]\frac{t}{n}+\int_{lt/n}^{(l+1)t/n}T^{-1}_\gamma(s)[L,A]T_\gamma(s)\,\mathrm{d}s\right).
\label{EQ89}
\end{eqnarray}
However,
\begin{eqnarray}
\int_{lt/n}^{(l+1)t/n}T^{-1}_\gamma(s)[L,A]T_\gamma(s)\,\mathrm{d}s&=&\frac{1}{\gamma}\int_{lt/n}^{(l+1)t/n}\bigl(\partial_t(T^{-1}_\gamma LT_\gamma)\bigr)(s)\,\mathrm{d}s\nonumber\\
&=&\frac{1}{\gamma}T^{-l}_\gamma(t/n)\bigr(T^{-1}_\gamma(t/n)LT_\gamma(t/n)-L\bigr)T^{l}_\gamma(t/n),
\end{eqnarray}
that is 
\begin{eqnarray}
&&T_\gamma(t)\prod_{l=0}^{n-1}\exp\left(\mathfrak{P}_{T}[F]\frac{t}{n}+\int_{lt/n}^{(l+1)t/n}T^{-1}_\gamma(s)[L,A]T_\gamma(s)\,\mathrm{d}s\right)\nonumber\\
&&\qquad=T_\gamma(t)\prod_{l=0}^{n-1}\exp\left(\mathfrak{P}_{T}[F]\frac{t}{n}+\frac{1}{\gamma}T^{-l}_\gamma(t/n)\bigr(T^{-1}_\gamma(t/n)LT_\gamma(t/n)-L\bigr)T^{l}_\gamma(t/n)\right).
\end{eqnarray}
Applying the statement of Remark \ref{RK5.6} with $\omega\vcentcolon =\|F\|$, we get
\begin{eqnarray}
&&\Biggl\|T_\gamma(t)\exp\left(\frac{1}{\gamma}\bigl(T^{-1}_\gamma(t)LT_\gamma(t)-L\bigr)\right)\exp\bigl(\mathfrak{P}_{T}[F]t\bigr)\nonumber\\
&&\qquad-T_\gamma(t)\prod_{l=0}^{n-1}\exp\left(\mathfrak{P}_{T}[F]\frac{t}{n}+\int_{lt/n}^{(l+1)t/n}T^{-1}_\gamma(s)[L,A]T_\gamma(s)\,\mathrm{d}s\right)\Biggr\|\nonumber\\
&&\leq\exp\left(\|F\| t\right)\Bigl(\|\mathfrak{P}_{T}[F]\|+2\gamma^{-1}\|L\|\Bigr)^2\,\frac{2t^2}{n}\nonumber\\
&&\qquad\times\exp\left(\Bigl(\|F\|+\|\mathfrak{P}_{T}[F]\|+2\gamma^{-1}\|L\|\Bigr)\frac{t}{n}\right).
\end{eqnarray}
Since $\gamma$, as well as $t$ is fixed, in the light of (\ref{EQ89}) this gives us 
\begin{eqnarray}
&&\lim_{n\rightarrow\infty}T_\gamma(t)\prod_{l=0}^{n-1}\exp\left(T_\gamma^{-l}(t/n)F T^{l}_\gamma(t/n)\frac{t}{n}\right)\nonumber\\
&&\qquad=T_\gamma(t)\exp\left(\frac{1}{\gamma}\bigl(T^{-1}_\gamma(t)LT_\gamma(t)-L\bigr)\right)\exp\bigl(\mathfrak{P}_{T}[F]t\bigr)
\end{eqnarray}
Using (\ref{F3}), we have
\begin{eqnarray}
\left\|\exp\left(\frac{1}{\gamma}\bigl(T^{-1}_\gamma(t)LT_\gamma(t)-L\bigr)\right)-\mathbbm{1}_{\banach}\right\|\leq \frac{2\|L\|}{\gamma}\exp\left(\frac{2\|L\|}{\gamma}\right),
\end{eqnarray}
so we arrive to 
\begin{equation}
\left\|\exp\bigl((\gamma A+F)t\bigr)-T_\gamma(t)\exp\bigl(\mathfrak{P}_{T}[F]t\bigr)\right\|\leq \frac{2\|L\|}{\gamma}\exp\left(\frac{2\|L\|}{\gamma}\right),
\end{equation}
which is just (\ref{EQ90}).\par
\vspace{3mm}
\noindent\textbf{Step 3.} Using the assumption of the Theorem, $\mathfrak{P}_{T}[K]$ exists. Using statement \textit{3)} of Theorem \ref{THM6.2}, $\overline{K}\vcentcolon =K-\mathfrak{P}_{T}[K]$ is in the kernel of $\mathfrak{P}_{T}$. Provided by statement \textit{5)} of the same Theorem, there exists $L_\varepsilon\in\bbanach$ such that $\|\overline{K}-[L_\varepsilon,A]\|\leq \varepsilon$. Using (\ref{EQ87b}) and (\ref{EQ90}) we have  
\begin{eqnarray}
&&\left\|\exp\bigl((\gamma A+K)t\bigr)-T_\gamma(t)\exp\bigl(\mathfrak{P}_{T}[K]t\bigr)\right\|\nonumber\\
&&\qquad\qquad\leq \frac{2\|L_\varepsilon\|}{\gamma}\exp\left(\frac{2\|L_\varepsilon\|}{\gamma}\right)+\varepsilon t\exp\left(\bigl(\|K\|+\varepsilon\bigr)t\right),
\label{EQ94}
\end{eqnarray}
therefore 
\begin{equation}
\lim_{0<\gamma\rightarrow\infty}\left\|\exp\bigl((\gamma A+K)t\bigr)-T_\gamma(t)\exp\bigl(\mathfrak{P}_{T}[K]t\bigr)\right\|\leq \varepsilon t\exp\left(\bigl(\|K\|+\varepsilon\bigr)t\right),
\end{equation}
Since $\varepsilon>0$ is arbitrary, we proved (\ref{EQ91}).
\end{proof}

\begin{Thm}[Generalized adiabatic theorem of contractions] Let $\banach$ be a Banach space with a decomposition $\banach=\banach_{\mathrm{I}}\oplus\banach_{\mathrm{C}}$ such that both $\banach_{\mathrm{I}}$ and $\banach_{\mathrm{C}}$ are closed subspaces of $\banach$ and the continuous projection $P_{\mathrm{I}}$ corresponding to $\banach_{\mathrm{I}}$ and its complementary projection $P_{\mathrm{C}}$ are norm one. Assume that the uniformly continuous semi-group $(T(t))_{t\geq 0}\subset\bbanach$ generated by $A\in\bbanach$ is of the form $T=T_{\mathrm{I}}\oplus T_{\mathrm{C}}$, where $(T_\mathrm{I}(t))_{t\geq 0}\subset\mathcal{B}(\banach_{\mathrm{I}})$ is a restriction of a group of isometric isomorphisms of $\banach$ to $\mathbb{R}_0^+$,  while $(T_\mathrm{C}(t))_{t\geq 0}\subset\mathcal{B}(\banach_{\mathrm{C}})$ is a semi-group of strict contractions with growth bound $w_\mathrm{C}<0$.  Let $K\in C^0([t_1,t_2],\bbanach)$, $0\leq t_1<t_2<\infty$  and let  $F_\gamma(\cdot,\cdot)$ be the solution of the initial value problem (\ref{Q2b}). Assume that for all $t\in[t_1,t_2]$, $\mathfrak{P}_{T_\mathrm{I}}[K(t)]$ exists and let $G(\cdot,t_1)$ be the solution of the initial value problem 
\begin{eqnarray}
\partial_1G(u,t_1)&=&\mathfrak{P}_{T_\mathrm{I}}[P_\mathrm{I}K(u)P_\mathrm{I}]G(u,t_1)\qquad\qquad t_{1}\leq u\leq t_{2}\nonumber\\
G(t_1,t_1)&=&P_{\mathrm{I}}.
\end{eqnarray}
Moreover, assume that for all $\varepsilon>0$, there exists $L_\varepsilon\in C^{0,1}([t_1,t_2],\mathcal{B}(\mathcal{X}_\mathrm{I}))$ such that 
\begin{equation}
\left\|\bigl(P_\mathrm{I}KP_{\mathrm{I}}-\mathfrak{P}_{T_\mathrm{I}}[P_\mathrm{I}KP_\mathrm{I}]\bigr)-[L_\varepsilon,A_\mathrm{I}]\right\|_{C^0}\leq \varepsilon,
\label{THM6.3AS}
\end{equation}
where $A_\mathrm{I}\in\mathcal{B}(\mathcal{X}_\mathrm{I})$ generates $(T_\mathrm{I}(t))_{t\geq 0}$.
Then, for any $t\in (t_1,t_2]$
\begin{equation}
\lim_{0<\gamma\rightarrow\infty}\left\|F_\gamma(t,t_1)-T_{\mathrm{I},\gamma}(t)G(t,t_1)\right\|_{\bbanach}=0,
\label{EQ96}
\end{equation}
where $T_{\mathrm{I},\gamma}(t)=T_{\mathrm{I}}(\gamma t)$ and the convergence is uniform on every compact subset of $(t_1,t_2]$, whenever $K$ is constant valued. 
\end{Thm}

\begin{proof}
Let $\omega:[t_1,t_2]\rightarrow\mathbb{R}$ be integrable on $[t_1,t_2]$ and assume that $\omega$ dominates the weak growth bounds $\omega_{K(\cdot)}$ and $\omega_{P_\mathrm{I}K(\cdot)P_{\mathrm{I}}}$. Note that $A$ generates a contraction semi-group, therefore for every $s\in[t_1,t_2]$ and for any $v\in\mathbb{R}^{+}_0$, we have 
\begin{eqnarray}
\left\|\exp\bigl((\gamma A+K(t))v\bigr)\right\|&\leq& \lim_{n\rightarrow\infty}\left\|\exp(\gamma Av/n)\right\|^n\|\left\|\exp(K(t)v/n)\right\|^n\nonumber\\
&\leq &\lim_{n\rightarrow\infty}\left\|\exp(K(t)v/n)\right\|^n\nonumber\\
&\leq &\exp(\omega_{K(t)}v).
\end{eqnarray}
That is, provided by statement \textit{1)} of Proposition \ref{lem:weak_gb1}, the weak growth bounds satisfy the inequality $\omega_{\gamma A+K(t)}\leq \omega_{K(t)}\leq \omega(t)$ for all $t\in [t_1,t_2]$. The same argument can be repeated to obtain
\begin{equation}
\omega_{\gamma A_\mathrm{I}+P_\mathrm{I}K(t)P_\mathrm{I}}\leq \omega_{P_\mathrm{I}K(t)P_\mathrm{I}}\leq \omega(t)
\label{EQ97}
\end{equation}
within $\mathcal{B}(\mathcal{X}_\mathrm{I})$. Therefore, using Proposition \ref{prop:growth_prop}, we have
\begin{eqnarray}
\|F_{\gamma}(u,v)\|&\leq& \exp\left(\int_{v}^{u}\omega_{K(s)}\,\mathrm{d}s\right)\leq \exp\left(\int_{v}^{u}\omega(s)\,\mathrm{d}s\right),\nonumber\\
\|G(u,v)\|&\leq& \exp\left(\int_{v}^{u}\omega_{P_\mathrm{I}K(s)P_{\mathrm{I}}}\,\mathrm{d}s\right)\leq \exp\left(\int_{v}^{u}\omega(s)\,\mathrm{d}s\right).
\label{EQ101}
\end{eqnarray}
 Provided by statement \textit{2)} of Proposition \ref{lem:weak_gb1}, the uniform continuity of $K$ on $[t_1,t_2]$ implies the same property of $\omega_{K(\cdot)}$. Define $\rho_{n,l}(t)\vcentcolon= t_1+l(t-t_1)/n$  for all $n\in\mathbb{N}$, $0\leq l <n$. Then, since $P_\mathrm{I}$ is a norm one projection, for any $\varepsilon>0$ we can choose $N_{\varepsilon}$ such that for any $t\in [t_1,t_2]$, the inequalities
\begin{eqnarray}
\max_{0\leq l<n}\sup_{s\in[\rho_{n,l}(t),\rho_{n,l+1}(t)]}\|K(s)-K(\rho_{n,l}(t))\|\leq \varepsilon,
\label{EQ98}
\end{eqnarray}
\begin{eqnarray}
\max_{0\leq l<n}\sup_{s\in[\rho_{n,l}(t),\rho_{n,l+1}(t)]}\|P_{\mathrm{I}}K(s)P_{\mathrm{I}}-P_{\mathrm{I}}KP_{\mathrm{I}}(\rho_{n,l}(t))\|\leq \varepsilon
\label{EQ99},
\end{eqnarray}
\begin{eqnarray}
\max_{0\leq l<n}\sup_{s\in[\rho_{n,l}(t),\rho_{n,l+1}(t)]}\|\omega_{K(t)}-\omega_{K(\rho_{n,l}(t))}\|\leq \varepsilon
\label{EQ100a}
\end{eqnarray}
and 
\begin{eqnarray}
\max_{0\leq l<n}\sup_{s\in[\rho_{n,l}(t),\rho_{n,l+1}(t)]}\|\omega_{P_\mathrm{I}K(t)P_{\mathrm{I}}}-\omega_{P_{\mathrm{I}}K(\rho_{n,l}(t))P_{\mathrm{I}}}\|\leq \varepsilon
\label{EQ100b}
\end{eqnarray}
are all satisfied, whenever $n\geq N_\varepsilon$. Note that $N_{\varepsilon}$ depends neither on $\gamma$, nor on $A$.\par
\vspace{2mm}
\noindent\textbf{Step 1}. The consideration above quarantees that for a given $\varepsilon>0$, if $n>N_{\varepsilon}$, then (\ref{EQ98})-(\ref{EQ100a}) are all satisfied, therefore their application with the help of Proposition \ref{PROP3.7}, (\ref{F1}) and (\ref{EQ101}) gives us
\begin{eqnarray}
&&\left\|F_\gamma(t,t_1)-\prod_{l=0}^{n-1}\exp\left(\Bigl(\gamma A+K(\rho_{n,l}(t))\Bigr)\frac{t-t_1}{n}\right)\right\|\nonumber\\&&\qquad=\left\|\prod_{l=0}^{n-1}F_\gamma(\rho_{n,l+1}(t),\rho_{n,l}(t))- \prod_{l=0}^{n-1}\exp\left(\Bigl(\gamma A+K(\rho_{n,l}(t))\Bigr)\frac{t-t_1}{n}\right)\right\|\nonumber\\
&&\qquad=\varepsilon(t-t_1)\exp\left(\int_{t_1}^{t}\omega(s)\,\mathrm{d}s\right)\exp\left(\varepsilon\frac{t-t_1}{n}\right),
\label{FINAL1}
\end{eqnarray}
if $n>N_{\varepsilon}$.
\par
\vspace{3mm}
\noindent\textbf{Step 2}. In order to decouple the strictly contractive part of $T_{\gamma}(t)$ from its isometric part, we first define $K_{n,l}\vcentcolon =K\circ\rho_{n,l}$ for any $n\in\mathbb{N}$, $0\leq l<n$. Then, for any $t\in [t_1,t_2]$ we can write
\begin{eqnarray}
&&\exp\left(\Bigl(\gamma A+K_{n,l}(t)\Bigr)\frac{t-t_1}{n}\right)\nonumber\\
&&\qquad\qquad=\lim_{m\rightarrow \infty}\prod_{k=0}^{m-1}T_\gamma\left(\frac{t-t_1}{nm}\right)\exp\left(K_{n,l}(t)\frac{t-t_1}{nm}\right).
\end{eqnarray}
The application of Lemma \ref{LEM4.2} shows
\begin{eqnarray}
&&\Biggl\|\prod_{k=0}^{m-1}T_\gamma\left(\frac{t-t_1}{nm}\right)\exp\left(K_{n,l}(t)\frac{t-t_1}{nm}\right)\nonumber\\
&&\qquad-T_{\mathrm{I},\gamma/n}(t-t_1)\prod_{k=0}^{m-1}T^{-k}_{\mathrm{I},\gamma/nm}(t-t_1)\exp\left(P_\mathrm{I}K_{n,l}(t)P_{\mathrm{I}}\frac{t-t_1}{nm}\right)T^{k}_{\mathrm{I},\gamma/nm}(t-t_1)\Biggr\|\nonumber\\
&&\leq \Omega_{n,l,t}\Lambda_{n,l,m}\frac{(t-t_1)^2}{n^2m}\Bigl(\|K_{n,l}(t)\|^2\mathrm{e}^{\|K_{n,l}(t)\|(t-t_1)/nm}\nonumber\\
&&\qquad+\tau_{\mathrm{C},n,l,m}\|K_{n,l}(t)\|^2\mathrm{e}^{2\|K_{n,l}(t)\|(t-t_1)/nm}+\|P_\mathrm{I}K_{n,l}(t)P_{\mathrm{I}}\|^2\mathrm{e}^{2\|P_\mathrm{I}K_{n,l}(t)P_\mathrm{I}\|(t-t_1)/nm}\Bigr)\nonumber\\
&&\qquad +\Omega_{n,l,t}\tau_{\mathrm{C},n,l,m}\Lambda_{n,l,m}\|K_{n,l}(t)\|\mathrm{e}^{\|K_{n,l}(t)\|(t-t_1)/nm}\frac{t-t_1}{nm}\nonumber\\
&&\qquad+\Omega_{n,l,t}\left\|T_{\mathrm{C},\gamma/nm}(t-t_1)\right\|^m,
\label{EQ92}
\end{eqnarray}
where
\begin{eqnarray}
\Omega_{n,l,t}&=&\exp\left(\omega_{K(\rho_{n,l}(t))}\frac{t-t_1}{n}\right)\leq \exp\left(\int_{\rho_{n,l}(t)}^{\rho_{n,l+1}(t)}\omega(s)\,\mathrm{d}s+\varepsilon\frac{t-t_1}{n}\right),\nonumber\\
&&\ \qquad\qquad\Lambda_{n,l,m}=\exp\left(|\omega_{K(\rho_{n,l}(t))}|\frac{t-t_1}{nm}\right),\nonumber\\
&&\qquad\qquad\tau_{\mathrm{C},n,l,m}=\frac{1+\left\|T_{\mathrm{C},\gamma/nm}(s-t_1)\right\|}{1-\left\|T_{\mathrm{C},\gamma/nm}(s-t_1)\bigr)\right\|},
\label{EQ102}
\end{eqnarray}
whenever $n\geq N_{\varepsilon}$ (the inequality (\ref{EQ100a}) has been used). Since 
\begin{eqnarray}
\left\|T_{\mathrm{C},\gamma/nm}(t-t_1)\right\|&\leq& \exp\left(\gamma w_{\mathrm{C}}\frac{t-t_1}{nm}\right)
\end{eqnarray}
and $x\mapsto (1+x)/(1-x)$ is monotonically increasing on $[0,1)$, we can obtain
\begin{eqnarray}
\tau_{\mathrm{C},n,l,m}&\leq& \coth\left(\gamma w_{\mathrm{C}}\frac{t-t_1}{2nm}\right),
\end{eqnarray}
thereso 
\begin{equation}
\lim_{m\rightarrow \infty}\frac{t-t_1}{nm}\tau_{\mathrm{C},n,l,m}=\frac{2}{\gamma |w_{\mathrm{C}}|}.
\end{equation}
Therefore, taking the $m\rightarrow \infty$ limit of the r.h.s of (\ref{EQ92}) gives
\begin{equation}
2\Omega_{n,l,t}\frac{1}{\gamma |w_{\mathrm{C}}|}\left(\frac{t-t_1}{n}\|K_{n,l}(t)\|^2+\|K_{n,l}(t)\|\right) +\Omega_{n,l,t}\mathrm{e}^{\gamma w_{\mathrm{C}}(t-t_1)/n}.
\end{equation}
Simultaneously performing the limit process on the l.h.s also, we obtain
\begin{eqnarray}
&&\left\|\exp\left(\Bigl(\gamma A+K_{n,l}(t)\Bigr)\frac{t-t_1}{n}\right)-P_\mathrm{I}\exp\left(\Bigl(\gamma A_\mathrm{I}+P_{\mathrm{I}}K_{n,l}(t)P_\mathrm{I}\Bigr)\frac{t-t_1}{n}\right)\right\|\nonumber\\
&&\qquad\leq 2\Omega_{n,l,t}\frac{1}{\gamma |w_{\mathrm{C}}|}\left(\frac{t-t_1}{n}\|K_{n,l}(t)\|^2+\|K_{n,l}(t)\|\right)+\Omega_{n,l,t}\mathrm{e}^{\gamma w_{\mathrm{C}}(t-t_1)/n}
\label{EQ95}
\end{eqnarray}
Since (\ref{EQ97}) holds, (\ref{EQ95}) combined with (\ref{F1}), the upper bound on $\Omega_{n,l,t}$ contained within (\ref{EQ102}) and the approximation inequality (\ref{EQ100a}) results in 
\begin{eqnarray}
&&\left\|\prod_{l=0}^{n-1}\exp\left(\Bigl(\gamma A+K_{n,l}(t)\Bigr)\frac{t-t_1}{n}\right)-P_\mathrm{I}\prod_{l=0}^{n-1}\exp\left(\Bigl(\gamma A_\mathrm{I}+P_{\mathrm{I}}K_{n,l}(t)P_\mathrm{I}\Bigr)\frac{t-t_1}{n}\right)\right\|\nonumber\\
&&\quad\leq \exp\left(\int_{t_1}^{t}\bigl(\omega(s)+\varepsilon\bigr)\,\mathrm{d}s\right)\nonumber\\
&&\quad\quad\quad\times\left[\frac{2(t-t_1)}{\gamma |w_{\mathrm{C}}|}\left(\|K_{n,l}(t)\|^2+n\|K_{n,l}(t)\|\right)+n\mathrm{e}^{\gamma w_{\mathrm{C}}(t-t_1)/n}\right]
\label{FINAL2}
\end{eqnarray}
if $n\geq N_\varepsilon$.
\par
\vspace{3mm}
\noindent\textbf{Step 3.} Given $\varepsilon>0$, $t_1<t\leq t_2$ and $N_\varepsilon$ as described before the first step of the proof, the assumptions of the Theorem quarantees the existence of $\mathfrak{P}_{T_{\mathrm{I}}}[P_\mathrm{I}K(\rho_{n,l}(t))P_\mathrm{I}]$ for all $0\leq l<n$, $n\in\mathbb{N}$ and the existence of $L_\varepsilon\in C^{0,1}([t_1,t_2],\bbanach)$ such that 
\begin{eqnarray}
\bigl\|\bigl(P_\mathrm{I}K(\rho_{n,l}(t))P_{\mathrm{I}}-\mathfrak{P}_{T_{\mathrm{I}}}[P_\mathrm{I}K(\rho_{n,l}(t))P_\mathrm{I}]\bigr)-L(\rho_{n,l}(t))\bigr\|\leq \varepsilon
\end{eqnarray}
for all $0\leq l<n$, $n\in\mathbb{N}$. Let $n\geq N_\varepsilon$. We can apply (\ref{EQ94}), (\ref{EQ97}) and (\ref{F1}) to obtain
\begin{eqnarray}
&&\Biggl\|P_\mathrm{I}\prod_{l=0}^{n-1}\exp\left(\Bigl(\gamma A_\mathrm{I}+P_{\mathrm{I}}K_{n,l}(t)P_\mathrm{I}\Bigr)\frac{t-t_1}{n}\right)\nonumber\\
&&\qquad\qquad-\prod_{l=0}^{n-1}F_{\mathrm{I},\gamma/n}(t)\exp\left(\mathfrak{P}_{T_{\mathrm{I}}}[P_\mathrm{I}K(\rho_{n,l}(t))P_\mathrm{I}]\frac{t-t_1}{n}\right)\Biggr\|\nonumber\\
&&\qquad\leq\exp\left(\int_{t_1}^{t}\bigl(\omega(s)+\varepsilon\bigr)\,\mathrm{d}s\right)\exp\left(\|\omega_{P_\mathrm{I}K(\cdot)P_{\mathrm{I}}}\|_{C^0}\frac{t-t_1}{n}\right)\nonumber\\
&&\qquad\qquad\times\Biggl[\frac{2\|L_\varepsilon\|_{C^0}}{\gamma}\exp\left(\frac{2\|L_\varepsilon\|_{C^0}}{\gamma}\right)+\varepsilon \frac{t-t_1}{n}\exp\left(\|K\|_{C^0}\frac{t-t_1}{n}\right)\Biggr],
\label{FINAL3}
\end{eqnarray}
whenever $n\geq N_\varepsilon$. Using statement \textit{3)} of Theorem \ref{THM6.2} and (\ref{EQ99}), to apply (\ref{F1}) for the same $n$, we obtain
\begin{eqnarray}
&&\left\|\prod_{l=0}^{n-1}F_{\mathrm{I},\gamma/n}(t)\exp\left(\mathfrak{P}_{T_{\mathrm{I}}}[P_\mathrm{I}K(\rho_{n,l}(t))P_\mathrm{I}]\frac{t-t_1}{n}\right)-F_{\mathrm{I},\gamma}G(t,t_1)\right\|\nonumber\\&&\qquad=\left\|\prod_{l=0}^{n-1}F_\gamma(\rho_{n,l+1}(t),\rho_{n,l}(t))- \prod_{l=0}^{n-1}\exp\left(\Bigl(\gamma A+K(\rho_{n,l}(t))\Bigr)\frac{t-t_1}{n}\right)\right\|\nonumber\\
&&\qquad=\varepsilon(t-t_1)\exp\left(\int_{t_1}^{t}\omega(s)\,\mathrm{d}s\right)\exp\left(\varepsilon\frac{t-t_1}{n}\right).
\label{FINAL4}
\end{eqnarray}
\par
\vspace{3mm}
\noindent\textbf{Step 4.} Keeping the chosen $\varepsilon>0$ $n\geq N_{\varepsilon}$, we can combine the inequalities (\ref{FINAL1}), (\ref{FINAL2}), (\ref{FINAL3}) and (\ref{FINAL4}) to obtain
\begin{equation}
\left\|F_{\gamma,t_1}-T_{\mathrm{I},\gamma}G(t,t_1)\right\|\leq f(n,\varepsilon)+g(n,\varepsilon,\gamma)\qquad n\geq \varepsilon,
\end{equation}
where  $f(n,\varepsilon)$ and $g(n,\varepsilon,\gamma)$ can be easily read out from (\ref{FINAL1}), (\ref{FINAL2}), (\ref{FINAL3}) and (\ref{FINAL4}). Set $n$ to be equal to $N_\varepsilon$. Then, provided that for any fixed $n$,
\begin{equation}
\lim_{\gamma\rightarrow\infty}g(n,\varepsilon,\gamma)=0,
\end{equation}
we can obtain that for the given $\varepsilon>0$,
\begin{eqnarray}
\lim_{\gamma\rightarrow\infty}\left\|F_{\gamma,t_1}(t,t_1)-T_{\mathrm{I},\gamma}G(t,t_1)\right\|&\leq& f(N_\varepsilon,\varepsilon)\nonumber\\
&=&2\varepsilon(t-t_1)\exp\left(\int_{t_1}^{t}\omega(s)\,\mathrm{d}s\right)\exp\left(\varepsilon\frac{t-t_1}{N_\varepsilon}\right)
\end{eqnarray}
Since $N_\varepsilon\rightarrow\infty$ can always be chosen whenever $\varepsilon\rightarrow 0$, this implies
\begin{equation}
\lim_{\gamma\rightarrow\infty}\left\|F_{\gamma,t_1}(t,t_1)-T_{\mathrm{I},\gamma}G(t,t_1)\right\|=0.
\end{equation}
\par
\vspace{3mm}
\noindent\textbf{Step 5.} If $K$ is contant, then (\ref{THM6.3AS}) is redundant due to statement \textit{5)} of Theorem \ref{THM6.2}. Moreover, the discretization steps based on (\ref{EQ98})-(\ref{EQ100b}) are unnecessary and one has to consider only the analogs of (\ref{FINAL2}) and (\ref{FINAL3}). Therefore, let $\omega\in\mathbb{R}$ such that $\omega_{P_{\mathrm{I}}KP_{\mathrm{I}}},\omega_K\leq \omega$ and choose $\varepsilon>0$. Provided by statement \textit{5)} of Theorem \ref{THM6.2}, there exists $L_\varepsilon\in\mathcal{B}(\mathcal{X}_\mathrm{I})$, such that 
\begin{equation}
\left\|\bigl(P_\mathrm{I}KP_{\mathrm{I}}-\mathfrak{P}_{T_\mathrm{I}}[P_\mathrm{I}KP_\mathrm{I}]\bigr)-[L_\varepsilon,A_\mathrm{I}]\right\|_{\mathcal{B}(\mathcal{X}_\mathrm{I})}\leq \varepsilon.
\end{equation}
A similar calculation which resulted in (\ref{EQ95}) gives
\begin{eqnarray}
&&\left\|\exp\left(\Bigl(\gamma A+K\Bigr)(t-t_1)\right)-P_\mathrm{I}\exp\left(\Bigl(\gamma A_\mathrm{I}+P_{\mathrm{I}}KP_\mathrm{I}\Bigr)(t-t_1)\right)\right\|\nonumber\\
&&\qquad\leq 2\mathrm{e}^{\omega (t-t_1)}\frac{1}{\gamma |w_{\mathrm{C}}|}\left((t-t_1)\|K\|^2+\|K\|\right)+\mathrm{e}^{\omega (t-t_1)}\mathrm{e}^{\gamma w_{\mathrm{C}}(t-t_1)},
\end{eqnarray}
while the calculation which gives (\ref{FINAL3}) can also be adapted to obtain 
\begin{eqnarray}
&&\Biggl\|P_\mathrm{I}\exp\left(\bigl(\gamma A_\mathrm{I}+P_{\mathrm{I}}P_\mathrm{I}\bigr)(t-t_1)\right)-F_{\mathrm{I},\gamma}(t)\exp\bigr(\mathfrak{P}_{T_{\mathrm{I}}}[P_\mathrm{I}KP_\mathrm{I}](t-t_1)\bigl)\Biggr\|\nonumber\\
&&\qquad\leq\mathrm{e}^{\omega(t-t_1)}\mathrm{e}^{|\omega_{P_\mathrm{I}KP_{\mathrm{I}}}|(t-t_1)}\frac{2\|L_\varepsilon\|}{\gamma}\exp\left(\frac{2\|L_\varepsilon\|}{\gamma}\right).
\end{eqnarray}
Let $t\in C$, where $C\subset (t_1,t_2]$ is compact with $c_M=\max C$ and $c_m=\min C$. Combining the inequalities above, we arrive to
\begin{eqnarray}
&&\left\|\exp\left(\Bigl(\gamma A+K\Bigr)(t-t_1)\right)-F_{\mathrm{I},\gamma}(t)\exp\bigr(\mathfrak{P}_{T_{\mathrm{I}}}[P_\mathrm{I}KP_\mathrm{I}](t-t_1)\bigl)\right\|\nonumber\\
&&\quad\leq  2\mathrm{e}^{|\omega| (c_M-t_1)}\frac{1}{\gamma |w_{\mathrm{C}}|}\left((c_M-t_1)\|K\|^2+\|K\|\right)+\mathrm{e}^{|\omega |(c_M-t_1)}\mathrm{e}^{\gamma w_{\mathrm{C}}(c_m-t_1)}\nonumber\\
&&\quad\ +\mathrm{e}^{|\omega|(c_M-t_1)}\mathrm{e}^{|\omega_{P_\mathrm{I}KP_{\mathrm{I}}}|(c_M-t_1)}\frac{2\|L_\varepsilon\|}{\gamma}\exp\left(\frac{2\|L_\varepsilon\|}{\gamma}\right).
\label{EQ103}
\end{eqnarray}
For large enough $\gamma$, the r.h.s of (\ref{EQ103}) is monotonically decreasing and tends to zero as $\gamma$ tends to infinity. Therefore, one can always choose an $L_\varepsilon$ dependent $\gamma_{\varepsilon}$ such that whenever $\gamma\geq \gamma_\varepsilon$, the r.h.s of (\ref{EQ103}) is less than $\varepsilon$. Since the bound in (\ref{EQ103}) is independent of $t\in C$ and depends only on the compact set $C$, we have that if $\gamma\geq \gamma_{\varepsilon}$, then 
\begin{equation}
\|F_\gamma(\cdot,t_1)-F_{\mathrm{\gamma,I}}(\cdot)G(\cdot,t_1)\|_{C^0}\leq \varepsilon
\end{equation}
on $C$. Thus, the last statement of the Theorem is proved.
\end{proof}

\section*{Appendix}
\setcounter{equation}{0}
\setcounter{Thm}{0}



\def\thesection{A} Some simple statements frequently used in the paper are collected here.

\begin{Lem} Let $\banach$ be a Banach space, $A_0,\dots, A_{n-1},B_0,\dots,B_{n-1}\in\bbanach$ such that $\|A_k\|_{\bbanach},$
$\|B_{k}\|_\bbanach\leq \exp(M_k)$ for all $0\leq k<n$ and with some $M_0,\dots,M_{n-1}>0$. Assume that for all $\max{0\leq k<n}$, $\exp(-M_k)\leq N$ holds. Then,
\begin{equation}
\left\|\prod_{l=0}^{n-1}A_l-\prod_{l=0}^{n-1}B_l\right\|_\bbanach\leq N\mathrm{e}^{\sum_{k=0}^{n-1}M_k}\sum_{l=0}^{n-1}\|A_l-B_l\|_{\bbanach}
\label{F1}
\end{equation}  
\end{Lem}

\begin{proof}
Follows directly from
\begin{equation}
\prod_{l=0}^{n-1}A_l-\prod_{l=0}^{n-1}B_l=\sum_{k=0}^{n-1}\left(\prod_{l=k+1}^{n-1}A_l\right)(A_k-B_k)\left(\prod_{l=0}^{k-1}B_l\right).
\end{equation}
\end{proof}

For every $K\in C([t_1,t_2],\bbanach)$, the solution of the initial value problem (\ref{EQ1}) comes in the form of the Picard series
\begin{equation}
F(u,v)=\sum_{l=0}^{\infty}\bigl\{K\bigr\}_l(u-v),
\end{equation} 
where $\bigl\{K\bigr\}_0(u,v)=\mathbbm{1}_{\banach}$ and 
\begin{equation}
\bigl\{K\bigr\}_l(u,v)=\int_{v}^{u}\mathrm{d}s_l\cdots\int_{v}^{\mathrm{s}_2}\mathrm{d}s_1 K(s_l)\cdots K(s_1)
\end{equation}
for $l>0$. The following statements can be easily verified using the definition of the $C^0$ norm and Taylor's theorem:

\begin{Lem} Let $\banach$ be a Banach space. 
\begin{enumerate}[label=\textit{\arabic*)}]
\item If $K\in C([t_1,t_2],\bbanach)$, then 
\begin{equation}
\left\|\bigl\{K\bigr\}_l(u,v)\right\|_{\bbanach}\leq\frac{(u-v)^p}{p!}\|K\|^p_{C^0}.
\label{F2}
\end{equation}
\item If $K\in C([t_1,t_2],\bbanach)$, then 
\begin{equation}
\left\|F(u,v)-\sum_{l=0}^{k-1}\bigl\{K\bigr\}_l(u,v)\right\|_{\bbanach}\leq(u-v)^k\|K\|^k_{C^0}\exp((u-v)\|K\|_{C^0}).
\label{F3}
\end{equation}
\end{enumerate}
\end{Lem}

For an arbitrary $X\in C([t_1,t_2],\bbanach)$, let us denote the propagator generated by $X$ by $F_X(\cdot,\cdot)$, that is $F_X(\cdot,\cdot)$ is the solution of the initial value problem (\ref{EQ1}), where $K$ is replaced by $X$. Assume that there exists a function $f_{X}:[t_1,t_2]\rightarrow\mathbb{R}$ such that for any $t_1\leq v<u\leq t_2$ and $l\in\mathbb{N}_0$ 
\begin{equation}
\bigl\{X\bigr\}_l(u-v)\leq \frac{f^l_X(u-v)}{l!}.
\end{equation}
The last statement is the following:
\begin{Lem} If $K,L,M\in C([t_1,t_2],\bbanach)$, $t_1\leq w<v< u\leq t_2$, then 
\begin{eqnarray}
&&\left\|F_{K+M}(u,v)F_{M}(v,w)-F_{M}(u,w)F_{K}(u,v)\right\|_\bbanach\nonumber\\
&&\qquad\qquad\leq g^2_{K+M,M}(u,v,w)\exp\bigl(g_{K+M,M}(u,v,w)\bigl)\nonumber\\
&&\qquad\qquad\qquad\qquad	+h^2_{K,M}(u,v,w)\exp\bigl(h_{K,M}(u,v,w)\bigl),
\label{F4}
\end{eqnarray}
where 
\begin{eqnarray}
g_{X,Y}(u,v,w)&=&f_X(u-v)+f_Y(v-w),\nonumber\\
h_{X,Y}(u,v,w)&=&f_X(u-v)+f_Y(u-w).
\end{eqnarray}
\end{Lem}

\begin{proof} Note that
\begin{equation}
\bigl\{K+M\bigr\}_1(u,v)+\bigl\{M\bigr\}_1(v,w)=\bigl\{M\bigr\}_1(u,w)+\bigl\{K\bigr\}_1(u,v).
\end{equation}
Therefore, 
\begin{eqnarray}
&&\left\|F_{K+M}(u,v)F_{M}(v,w)-F_{M}(u,w)F_{K}(u,v)\right\|\nonumber\\
&&\qquad\leq\sum_{k=2}^{\infty}\sum_{l=0}^{k}\left\|\bigl\{K+M\bigr\}_l(u,v)\right\|\,\left\|\bigl\{M\bigr\}_k(v,w)\right\|+\left\|\bigl\{M\bigr\}_{k}(u,w)\right\|\,\left\|\bigl\{K\bigr\}_{l-k}(u,v)\right\|\nonumber\\
&&\qquad\leq\sum_{k=2}^{\infty}\sum_{l=0}^{k}\frac{f^l_{K+M}(u-v)}{k!}\frac{f^{k-l}_{M}(v-w)}{(k-l)!}+\frac{f^l_{M}(u-w)}{k!}\frac{f^{k-l}_{K}(u-v)}{(k-l)!}\nonumber\\
&&\qquad\leq\sum_{k=2}^{\infty}\frac{\bigl(f_{K+M}(u-v)+f_{M}(v-w)\bigr)^k}{k!}+\frac{\bigl(f_{K}(u-v)+f_{M}(u-w)\bigr)^k}{k!}\nonumber\\
&&\qquad\leq\bigl(f_{K+M}(u-v)+f_{M}(v-w)\bigr)^2\exp\bigl(f_{K+M}(u-v)+f_{M}(v-w)\bigr)\nonumber\\
&&\qquad\quad\qquad+\bigl(f_{K}(u-v)+f_{M}(u-w)\bigr)^2\exp\bigl(f_{K}(u-v)+f_{M}(u-w)\bigr).
\end{eqnarray}
\end{proof}

\bibliographystyle{acm}
\bibliography{zeno_arxiv}

\signnb \signzz

\end{document}